\documentclass[acmsmall]{ec21acm}
\usepackage{acm-ec-21}
\usepackage{booktabs} 
\usepackage[ruled]{algorithm2e} 

\SetAlFnt{\small}
\SetAlCapFnt{\small}
\SetAlCapNameFnt{\small}
\SetAlCapHSkip{0pt}
\IncMargin{-\parindent}

\usepackage{color,xcolor}
\usepackage{multirow}
\usepackage{appendix}
\usepackage{cleveref,hyperref}
\usepackage{caption,subcaption}
\usepackage{enumitem}
\usepackage{algorithmic}
\usepackage{graphicx,float}

\setcitestyle{acmnumeric}

\theoremstyle{remark}
\newtheorem{remark}{Remark}
\newtheorem{observation}{Observation}
\newtheorem{definition}{Definition}
\newtheorem{assumption}{Assumption}
\theoremstyle{definition}
\newtheorem{example}{Example}

\begin{document}

\title{Equilibrium Computation of Generalized Nash Games: A New Lagrangian-Based Approach}

\author{Jong Gwang Kim}
\authornote{Email: kim2133@purdue.edu}
\affiliation{Purdue University}
\email{kim2133@purdue.edu}

\begin{abstract}
This paper presents a new primal-dual method for computing an equilibrium of generalized (continuous) Nash game (referred to as \emph{generalized Nash equilibrium problem} (GNEP)) where each player's feasible strategy set depends on the other players' strategies. The method is based on a new form of Lagrangian function with a quadratic approximation. First, we reformulate a GNEP as a saddle point computation problem using the new Lagrangian and establish equivalence between a saddle point of the Lagrangian and an equilibrium of the GNEP. We then propose a simple algorithm that is convergent to the saddle point. Furthermore, we establish global convergence by assuming that the Lagrangian function satisfies the \emph{Kurdyka-{\L}ojasiewicz} property. A distinctive feature of our analysis is to make use of the new Lagrangian as a potential function to guide the iterate convergence. Our method has two novel features over existing approaches: (i) it requires neither boundedness assumptions on the strategy set and the set of multipliers of each player, nor any boundedness assumptions on the iterates generated by the algorithm; (ii) it leads to a Jacobi-type decomposition scheme, which, to the best of our knowledge, is the first development of a \emph{distributed} algorithm to solve a \emph{general class} of GNEPs. Numerical experiments are performed on benchmark test problems and the results demonstrate the effectiveness of the proposed method.
\end{abstract}

\maketitle

\section{Introduction}
We consider a generalized (continuous) Nash game (\emph{generalized Nash equilibrium problem} (GNEP)) that describes a  broad class of non-cooperative and simultaneous-move games, in which  each player seeks to optimize her/his own objective function while subject to certain constraints that are affected by the other players' strategies. The standard Nash game \citep{nash1950equilibrium} is a subclass of GNEPs, as the strategic interactions among players in a Nash game are only reflected in their objective functions, not in the constraints. More specifically, the game features a set of $N$ players denoted by $\mathcal{N}=\left\{ 1,\ldots,N\right\}$ where each player $\nu$ has its own strategy $x^{\nu}\in\mathbb{R}^{n_{\nu}}$. Each player $\nu$ has an objective function $\theta_{\nu}(x^{\nu},x^{-\nu})$ and a finite set of ``coupling constraints'' $g_{i}^{\nu}(x^{\nu},x^{-\nu}) \leq 0$ $(i=1, \ldots, m_{\nu})$, both of which depend on player $\nu$'s own strategy $x^{\nu}$ as well as other players' strategies  $x^{-\nu} := (x^{\nu^{\prime}})_{\nu^{\prime} \neq \nu}$. Denote all players' strategies by a vector $\mathbf{x} = (x^{\nu},x^{-\nu}):=(x^1,\ldots,x^\nu,\ldots,x^N)$ with dimension $n=\sum_{\nu=1}^{N} n_{\nu}$.  The GNEP can be formally defined as a problem of simultaneously finding a solution for each of the following problem. Given other players' strategies $x^{-\nu}$, each player $\nu$ seeks to find a strategy $x^{\nu}$ that solves the optimization problem:
\begin{equation} \label{eq:of}
	\mathrm{P}_{\nu}(x^{-\nu}) \quad
	\begin{aligned}
	\underset{x^{\nu}}{\mathrm{minimize}} & \quad\theta_{\nu}(x^{\nu},x^{-\nu})\\
	\mathrm{\mathrm{subject\:to}} & \quad g_{i}^{\nu}(x^{\nu},x^{-\nu}) \leq 0,\quad i=1,\ldots,m_{\nu},\\
	&
	\quad {x}^{\nu} \in \mathcal{X}_{\nu},
	\end{aligned}
\end{equation}
where $\mathcal{X}_{\nu}\subseteq\mathbb{R}^{n_{\nu}}$ represents the ``private'' strategy set of player $\nu$ that is nonempty, closed, and convex. The feasible strategy set of each player $\nu$ can be represented by the parametric inequalities:
    \[
	\mathcal{F}_{\nu}(x^{-\nu}) := \left\{ x^{\nu}\in\mathcal{X}_{\nu}\left|\right.g_{i}^{\nu}(x^{\nu},x^{-\nu})\leq0,\;i=1,\ldots,m_{\nu} \right\} \subseteq\mathbb{R}^{n_{\nu}}.
	\]
Note that the private functional constraints $c^{\nu}_j(x^{\nu})\leq0$ for $j=1,\ldots,p_\nu$ are not explicitly highlighted in the paper for notational simplicity. They can be easily treated by the same way to deal with $g_{i}^{\nu}(x^{\nu},x^{-\nu}) \leq 0$. Here $n_\nu$, $m_\nu$ and $p_\nu$ are positive integers. The simple set $\mathcal{X}_{\nu}$ is defined by
$\mathcal{X}_{\nu} := \left\{ x^{\nu}\in \mathbb{R}^{n_{\nu}}\left|\right. l_{\nu} \leq x^{\nu} \leq u_{\nu} \right\}$, where $l_\nu$ or $u_\nu$ may be unbounded; that is, $l_\nu=-\infty$ or $u_\nu=+\infty$ or both.
    
A Nash equilibrium of the GNEP can be defined as follows.
\begin{definition}\label{def_gne}	
A collection of strategies $\mathbf{x}^{\ast}=\left(x^{1,\ast},\ldots,x^{N,\ast}\right)$ is a (pure-strategy) generalized Nash equilibrium (GNE) if for every $\nu=1,\ldots,N$,
	\[
	\theta_{\nu}\left(x^{\nu,\ast},x^{-\nu,\ast}\right)\leq\theta_{\nu} \left(x^{\nu},x^{-\nu,\ast}\right),\quad\forall x^{\nu}\in \mathcal{F}_{\nu} (x^{-\nu,\ast}),
	\]
i.e.,  $\mathbf{x}^{\ast}=(x^{1,\ast},\ldots,x^{N,\ast})$ is a GNE, if and only if no player has incentive to unilaterally deviate from $x^{\nu,\ast}$ when other players choose $x^{-\nu,\ast}$.
\end{definition}
	
We make the following assumption on the functions throughout the paper.
\begin{assumption}\label{assumption1} 
	For every $\nu \in \mathcal{N}$ and fixed ${x}^{-\nu}$, objective function $\theta_{\nu}(x^{\nu},x^{-\nu})$
	and constraint functions $g_{i}^{\nu}(x^{\nu},x^{-\nu})$, $i=1,\ldots,m_{\nu}$, are continuously differentiable and convex with respect to $x^{\nu}$. 
\end{assumption}

Note that  $\theta_{\nu}(x^{\nu},x^{-\nu})$ and $g_{i}^{\nu}(x^{\nu},x^{-\nu})$ are possibly nonconvex with respect to some other players' decisions $x^{\nu^{\prime}}\in x^{-\nu}$, and $g_{i}^{\nu}(x^{\nu},x^{-\nu})$ are not necessarily shared by all players (called \emph{non-shared coupling constraints}). 
	
Under Assumption \ref{assumption1}, Problem $\eqref{eq:of}$ is known as a very general form of GNEP \citep{dreves2011solution} (We call it \emph{general} GNEP in this work). In this paper, we aim to provide and analyze the first \emph{distributed} primal-dual algorithm, based on a novel form of Lagrangian, to compute an equilibrium of the \textit{general} GNEP, provided that equilibria of generalized Nash game exist.

We also make the following two standard assumptions; Lipschitz gradient continuity of the objective and constraint functions (smoothness) and coercivity of the objective functions.

\begin{assumption}[Uniform Lipschtz gradient continuity]\label{assumption_lipschitz}
For every $\nu=1,\ldots,N$, the gradients of $\theta_\nu(x^{\nu},x^{-\nu})$ and $g^\nu(x^{\nu},x^{-\nu})$ are uniformly Lipschitz continuous with constants; there exist constants $L_\nu(\theta_{\nu}), L_\nu(g^{\nu}), L_{-\nu}(\theta_{\nu}),  L_{-\nu}({g^{\nu}})>0$ such that 
\begin{subequations}
	\begin{align}
		\left\Vert \nabla_{x^\nu}\theta_{\nu}(x_{1}^{\nu},x^{-\nu})-\nabla_{x^{\nu}}\theta_{\nu}(x_{2}^{\nu},x^{-\nu})\right\Vert 
		& \leq  
		L_\nu(\theta_{\nu})
		\left\Vert x_{1}^{\nu}-x_{2}^{\nu}\right\Vert,  \quad \forall x_{1}^{\nu},x_{2}^{\nu}\in\mathcal{X}_{\nu}, \label{eq:assumption_lipschitz_1}\\
		\left\Vert \nabla_{x^\nu}g^{\nu}(x_{1}^{\nu},x^{-\nu})-\nabla_{x^\nu}g^{\nu}(x_{2}^{\nu},x^{-\nu})\right\Vert 
		& \leq 
		L_\nu(g^{\nu}) \left\Vert x_{1}^{\nu}-x_{2}^{\nu}\right\Vert, \quad \forall x_{1}^{\nu},x_{2}^{\nu}\in\mathcal{X}_{\nu}. \label{eq:assumption_lipschitz_2}\\
		\left\Vert \nabla_{x^{-\nu}}\theta_{\nu}(x^{\nu},x_1^{-\nu}) - \nabla_{x^{-\nu}}\theta_{\nu}(x^{\nu},x_2^{-\nu})\right\Vert 
		& \leq  
		L_{-\nu}(\theta_{\nu})
		\left\Vert x_{1}^{-\nu}-x_{2}^{-\nu}\right\Vert,  \quad \forall x_{1}^{-\nu},x_{2}^{-\nu}\in\mathcal{X}_{-\nu}, \label{eq:assumption_lipschitz_11}\\
		\left\Vert \nabla_{x^{-\nu}}g^{\nu}(x^{\nu},x_1^{-\nu}) - \nabla_{x^{-\nu}}g^{\nu}(x^{\nu},x_2^{-\nu})\right\Vert 
		& \leq  
		L_{-\nu}(g^{\nu})
		\left\Vert x_{1}^{-\nu}-x_{2}^{-\nu}\right\Vert,  \quad \forall x_{1}^{-\nu},x_{2}^{-\nu}\in\mathcal{X}_{-\nu}. \label{eq:assumption_lipschitz_21}
	\end{align}
In addition, for any fixed $x^\nu \in \mathcal{X}_\nu$, there exist constants $M_{-\nu}(\theta_\nu)>0$ and $M_{-\nu}(g^\nu)>0$ such that
\begin{align}
	\left\Vert \nabla_{x^{\nu}}\theta_{\nu}(x^{\nu},x_{1}^{-\nu})-\nabla_{x^{\nu}}\theta_{\nu}(x^{\nu},x_{2}^{-\nu})\right\Vert 
	&\leq 
	M_{-\nu}(\theta_\nu) \left\Vert x_{1}^{-\nu}-x_{2}^{-\nu}\right\Vert, \quad \forall x_{1}^{-\nu},x_{2}^{-\nu}\in\mathcal{X}_{-\nu},  \label{eq:assumption_lipschitz_3} \\
	\left\Vert 
	\nabla_{x^{\nu}} g^{\nu}(x^{\nu},x_{1}^{-\nu})-\nabla_{x^{\nu}} g^\nu(x^{\nu},x_{2}^{-\nu})\right\Vert 
	&\leq 
	M_{-\nu}(g^\nu) \left\Vert x_{1}^{-\nu}-x_{2}^{-\nu}\right\Vert, \quad \forall x_{1}^{-\nu},x_{2}^{-\nu}\in\mathcal{X}_{-\nu}.  \label{eq:assumption_lipschitz_4}
\end{align}
	\end{subequations}
Here, $\mathcal{X}_{-\nu}:=\prod_{q \neq \nu} \mathcal{X}_q$.
\end{assumption}	


\begin{assumption}[Coercivity of objective function] \label{assumption_coercive}
	For every $\nu=1,\ldots,N$, the objective function 
	$\theta_{\nu}(x^{\nu},x^{-\nu})$ is coercive with respect to $\mathbf{x}=(x^\nu, x^{-\nu}) \in  \mathcal{X}_{\nu} \times \mathcal{X}_{-\nu}$, i.e., $\mathrm{lim}_{\left\|\mathbf{x} \right\| \rightarrow  \infty} \; \theta_{\nu}(\mathbf{x})= \infty$.
\end{assumption}
Assumption \ref{assumption_coercive} implies that the level sets of the objective functions are bounded. 
It is well-known that coercivity assumption on the objective function is a standard assumption in nonconvex settings; see e.g., \cite{li2015global, bot2019proximal,boct2020proximal}. However, we do not impose the coercivity assumption on the feasible strategy sets, which is in contrast to the analysis of the interior-point algorithm for solving general GNEPs in \cite{dreves2011solution}. In \cite{dreves2011solution}, the algorithm relies on the strong assumption that the feasible strategy sets of all players are bounded, i.e., $\mathrm{lim}_{\left\|\mathbf{x} \right\| \rightarrow  \infty}  \left\| g^{\nu}_+(\mathbf{x}) \right\|= +\infty$ where $g^{\nu}_+(\mathbf{x}):=\mathrm{max}\{0,g^{\nu}(\mathbf{x})\}$ for all $\nu=1,\ldots,N$.

\subsection{Motivating Examples}

We give two practical examples that illustrate the versatility of the general GNEP model \eqref{eq:of} and motivate to develop an efficient algorithm.  We refer the readers to \cite{facchinei2010generalized} for more examples.
	
\renewcommand\thetheorem{1}
\begin{example}[Power allocation in telecommunications] \label{example_1}
This model is described in detail in  \cite{pang2008distributed} and represents a realistic communication system subject to Quality-of-Service (QoS) constraints. There are $N$ links transmitting to $K$ different Base Stations by using $K$ different channels. Link $\nu$ transmits with power $x^\nu = (x^\nu_1, \ldots, x^\nu_K)$, and denote by $\mathbf{x} = (x^1,\dots,x^N)$ the power allocation of all links. The game-theoretical model is defined by 
		\[
		\underset{x^{\nu}}{\mathrm{minimize}}  \quad\sum_{i=1}^{K}x_{i}^{\nu} \ \
		\mathrm{\mathrm{subject\:to}} 
		\ \ \sum_{i=1}^{K}\mathrm{log}_{2}\left(1+\frac{h_{i}^{\nu\nu}x_{i}^{\nu}}{\left(\sigma_{i}^{\nu}\right)^2+\underset{\mu\neq\nu}{\sum}h_{i}^{\nu\mu}x_{i}^{\mu}}\right)\geq L^{\nu}, \quad x^{\nu}\geq0,
		\]
where $h^{\nu \mu}_i$ denotes the power gain between transmitter $\mu$ and receiver $\nu$ on the $i$-th channel, $\left( \sigma_i^\nu \right)^2$ is the noise of link $\nu$ on the $i$-th channel, and $L^\nu$ is the minimum transmission rate (target rate) for link $\nu$.	

In this game-theoretical model, the QoS constraints are nonconvex in the other links' decision variables, and each link has different coupling QoS constraints. Thus, this problem can be viewed as a case of the general GNEP \eqref{eq:of}. Moreover, due to the complicated coupling constraints, it is hard to compute a Nash equilibrium of the model efficiently.
\end{example}	

\renewcommand\thetheorem{2}
\begin{example}[Arrow-Debreu general equilibrium model] \label{example_2}
This equilibrium model is introduced by \cite{arrow1954existence} and also described as a general GNEP in \cite{facchinei2010generalized}. There are $I$ consumers, $J$ firms, and one market player (a fictitious player) in this equilibrium model. Both consumers and firms deal with $K$ goods. The market player sets (normalized) prices $p \in \mathbb{R}^K_+$ for solving a market clearing problem. The $j$-th firm maximizes its profit by deciding how much to produce $y^j \in Y_j$, where $Y_j \subseteq \mathbb{R}^K$ is a production set. The $i$-th consumer decides how much of each good to buy $x^i \in X_i$ to maximize its utility, where $X_i \subseteq \mathbb{R}^K$ is a consumption set. The GNEP is defined as the following set of problems of the three types of players: 
\begin{xxalignat}{3}
\underset{y^j}{\mathrm{max}} & \ \ p^T y^j 
& \underset{x^i}{\mathrm{max}} & \ \ u_i(x^i) 
& \underset{p}{\mathrm{max}} & \ \ p^T \left(\sum_{i=1}^I x^i - \sum_{j=1}^J y^j - \sum_{i=1}^I \xi^i\right) \\
\mathrm{\mathrm{s.t.}}   & \ \ y^j \in Y_j,
&	\mathrm{\mathrm{s.t.}} & \ \ p^T x^i\leq p^T \left(\xi^i+\sum_{j=1}^J q_{ij} y^j\right), \ \ x^i \in X_i
& \mathrm{\mathrm{s.t.}} & \ \ \sum_{k=1}^K p_k=1, \ \  p_k \geq 0. 
\end{xxalignat}
The first problem corresponds to firm $j$'s production problem, the second problem corresponds to the consumption problem of consumer $i$, and the last problem corresponds to the  market player's problem. Here, $q_{ij}\geq0$ represents the fraction of the profit of the $j$-th production owned by consumer $i$ such that $\sum_{i=1}^{I}q_{ij}=1$, and $\xi^i \in \mathbb{R}_{+}^{K}$ is an initial endowment of goods.
	
\renewcommand\thetheorem{2}
\begin{definition}[Walrasian Equilibrium; see e.g., \cite{jofre2007variational}] 
A Walrasian equilibrium consists of a price vector $\overline{p}$, consumption vectors $\overline{x}^{i}$ for $i=1,\ldots,I$, and production vectors $\overline{y}^{j}$ for $j=1,\ldots,J$, such that 
			
(E1) (Market Nontriviality). $\overline{p}\geq0,$ $\overline{p}\neq0$, 
			
(E2) (Utility Optimization). $\overline{x}^i$ maximizes $u(x^i)$ over $x^i \in X_i$ s.t.  $\overline{p}^T x^i \leq \overline{p}^T \xi^i + \sum_{j=1}^J q_{ij} \overline{p}^T\overline{y}^j,$
			
(E3) (Profit Optimization). $\overline{y}^j$ maximizes $\overline{p}^T y^j$  over $y^j \in Y_j$,
			
(E4) (Market Clearing). Supplies and demands are balanced in the sense that
  \[
  \overline{z}\leq0 \ \ \mathrm{and} \ \ \overline{p}\cdot\overline{z}=0 \ \ \mathrm{for} \ \ \overline{z}=\sum_{i=1}^{I}\overline{x}^{i}-\sum_{j=1}^{J}\overline{y}^{j}-\sum_{i=1}^{I}\xi^{i}.
  \]
\end{definition}
For any $\overline{p}$, conditions $(E2)$ and $(E3)$ can be expressed in terms of 
\[
\begin{aligned}
&
\overline{x}^{i}\in X_{i}\left(\overline{p},\overline{y}^{1},\ldots,\overline{y}^{J}\right) \ \  \mathrm{with} \ \ \overline{y}^{j}\in Y_{j}\left(\overline{p}\right),
\end{aligned}
\]
where $Y_{j}\left(p\right)=\underset{y^{j}\in Y_{j}\left(p\right)}{\mathrm{argmax}}  \ p^{T}y^{j}$ and $X_{i}\left(p,y^{1},\ldots,y^{J}\right)=\underset{x^{i}\in X_{i}\left(p\right)}{\mathrm{argmax}}\left\{ p^T x^{i}\leq p^T \xi^i+ \sum_{j=1}^{J}q_{ij}p^T y^j\right\}$. Note that (E4) comes out as a linear complimentarity condition on $\bar{p}$ and $\bar{z}$. The above expressions lead to the idea of capturing all conditions for equilibrium, including (E4), in terms of a mapping from  $p$ to $z$:
$Z(p) := \left\{z = \sum_{i=1}^{I} x_i + \sum_{j=1}^{J} y_j - \sum_{i=1}^{I} \xi_i \left| \ y_j \in Y_j(p), x_i \in X_i (p, y_1, \ldots, y_J) \right. \right\}.$
Clearly, we have  
\[
\bar{p}\in{P} \textrm{\ \ yields \ equilibrium}\Longleftrightarrow\exists\bar{z}\in{Z}(\bar{p})\ \ \textrm{such \ that} \ \ \bar{z}\leq0, \ \bar{p}\cdot\bar{z}=0.
\] 
We thus see that the feasible strategy set of each player depends on the other players' selfish decisions. Furthermore, each player's feasible set does not depend on all other players' decisions (\textit{i.e., non-shared coupling constraints}). Hence, the Arrow-Debreu equilibrium model is a case of the \emph{general} GNEPs \eqref{eq:of}. 	

\end{example}

\subsection{Literature Review}
	
The concept of GNEP was originally addressed by \citet{debreu1952social} and \citet{arrow1954existence} in the early 1950s, where a GNEP was called a social equilibrium problem or abstract economy. An important subclass of GNEPs, known as jointly-convex GNEPs, was first investigated by  \citet{rosen1965existence} where all players share the same convex coupling constraints (i.e., $g^{1}=\cdots=g^{N}$). Although early studies on GNEPs have been primarily concerned with economics, recent decades have witnessed a growing interest in the GNEP as a modeling framework and a solution concept in various application areas. Some examples include electricity market models \cite{jing1999spatial,contreras2004numerical,hobbs2007nash}, power allocation in telecommunications \cite{pang2008distributed}, environmental pollution applications \cite{krawczyk2000relaxation,breton2006game}, transportation systems \cite{stein2018noncooperative}, and mobile cloud computing \cite{cardellini2016game}, to name a few.
		
Many approaches have been proposed to compute a GNE in the literature. A common approach is to transform a GNEP into a variational inequality (VI) and to apply algorithms designed to find a solution of a VI reformulation (see e.g., \cite{facchinei2007generalized, nabetani2011parametrized, yin2011nash, kulkarni2012variational}). Another approach is to reformulate a GNEP into a global optimization problem using Nikaido-Isoda (NI) function and then solve the resulting optimization problem by the so-called relaxation algorithms \cite{uryas1994relaxation,von2009relaxation} or gradient-based algorithms \cite{von2009optimization}. However, the theoretical and algorithmic properties of both approaches are only established for the class of jointly-convex GNEPs. In particular, VI-based methods require the monotonicity assumption on the variational mapping that generally does not hold in \emph{general} GNEPs. 

The equilibrium computation of GNEPs beyond the class of jointly-convex GNEPs remains a very challenging task. This is mainly due to interdependence between each player's strategy and some other players' strategies through coupling constraints and potential nonconvexity of each player's optimization problem with regard to the strategies chosen by other players. A few algorithms have indeed been proposed, including penalty-type methods \cite{pang2005quasi, facchinei2010penalty}, interior point algorithm \cite{dreves2011solution} and augmented Lagrangian method \cite{kanzow2016augmented}. In all such methods, it is assumed that the extended Mangasarian-Fromovitz constraint qualification (EMFCQ), an extension of the MFCQ for infeasible points, holds with respect to $\mathbf{x}=(x^\nu,x^{-\nu})$ for each player $\nu$, at every limit point of the sequence $\left\lbrace \mathbf{x}^k \right\rbrace $ generated by the algorithms. This MFCQ is a restrictive assumption since it is equivalent to the boundedness of the multiplier set of each player, and it cannot be justified to hold for every player in the GNEP in general due to the nature of GNEPs. More specifically, $\nabla_{\bf{x}}g^\nu (x^\nu,x^{-\nu})\lambda^\nu = 0$ and  the multipliers $\lambda^\nu \rightarrow \infty$ might occur since it is possible that $\nabla_{\bf{x}}g^\nu (x^\nu,x^{-\nu})=0$ from the strategic interactions among players through the coupling constraint. Thus, GNEPs may have an unbounded Lagrange multiplier set of each player.
	
Penalty-based algorithms reduce the GNEP to a standard Nash equilibrium problem (NEP) by penalizing coupling constraints and focus on updating the penalty parameter. In particular, the exact penalty method in \cite{facchinei2010penalty} results in nonsmooth subproblems, so it obtains a GNE under various differentiability assumptions on the objective functions and constraints. This lack of differentiability is a serious problem for designing efficient algorithms. To address the drawbacks of penalty-based methods, Kanzow and Steck proposed an augmented Lagrangian method in \cite{kanzow2016augmented}. This approach requires an algorithmic assumption that there exists a limit point of the primal sequence $\{{\mathbf{x}^k}\}$. However, this assumption is not clear without the compactness of each player's private set.

Furthermore, it is noteworthy to point out that even with our coercivity assumption on the objective functions, the augmented Lagrangian (AL) algorithm in \cite{kanzow2016augmented} does not guarantee that the primal sequence and/or dual sequence remains bounded. To ensure the boundedness of the sequences with the coercivity assumption, the bounded level sets of the AL functions are required. However, these level sets are typically unbounded. This is mainly related to the behavior of the multiplier sequence $\{\lambda^{,\nu,k}\}$. Specifically, the AL method in \cite{kanzow2016augmented} is of min-max dynamics (due to the increase in the dual variable) and by nature, the AL function value alternatively increases and decreases, and the dual sequence $\{\lambda^{\nu,k}\}$ might be unbounded. Hence, the coercivity does not imply the boundedness of the primal and dual sequences in the AL algorithm. This can be illustrated by a simple example.

\begin{example}
Consider the two players GNEP with $n_1=n_2=1$:
	\[ \label{eq:example3}
	\mathrm{P}_1(x_2) \quad 
	\begin{aligned}
	\underset{x_1 \in \mathbb{R}}{\mathrm{min}} & \ \  x_1^2   \\
	\mathrm{s.\:t.} & \ \  x_1^2+x_2^2 \leq 1, 
	\end{aligned} 
	\qquad \qquad \quad
	\mathrm{P}_2(x_1) \quad
	\begin{aligned}
	\underset{x_2 \in \mathbb{R}}{\mathrm{min}} & \ \   x_2^2  \\
	\mathrm{s.\:t.} & \ \ (x_1-2)^2 + x_2^2 \leq 1, 
	\end{aligned}
	\]
where both players' objective functions are coercive. We see that at the unique equilibrium (and only feasible point) given by $\mathbf{x}^\ast=(x_1^\ast, x_2^\ast) = (1,0)$, the gradients of the constraints are linearly dependent, and hence the MFCQ does not hold. As a result, the AL algorithm in \cite{kanzow2016augmented} may generate an unbounded multiplier sequence, which results in an unbounded sequence of the AL function values and thus failure to have limit points. 
\end{example}
In a computational perspective, a major limitation of the methods in \cite{facchinei2009penalty,kanzow2016augmented} is that they cannot be implemented in a distributed way since they need to solve subproblems represented by a large system of nonlinear equations (variational inequality) at each iteration. Thus, they are computationally expensive.

\subsection{Our Contributions}
	
The objective of this paper is to propose a new algorithmic framework for computing an equilibrium of a \textit{general} GNEP under general assumptions without imposing boundedness assumptions on the generated primal and dual sequences and the (feasible) strategy sets of all players. To achieve such a goal, we introduce a novel form of Lagrangian, termed as \emph{Proximal-Perturbed Lagrangian} (P-Lagrangian), and utilize a quadratic approximation of the P-Lagrangian.
	
The key ideas underlying our approach are as follows.
First, \emph{perturbation} variables $z_{i}^{\nu}$ are introduced to form constraints $g_{i}^{\nu}\left(x^{\nu},{x}^{-\nu,\ast}\right)\leq z_{i}^{\nu}$ and  $z_{i}^{\nu}=0$, which can be relaxed into the objective function with Lagrange multipliers. This reformulation allows the use of $\frac{\alpha_\nu}{2}\left\Vert z_{i}^{\nu}\right\Vert ^{2}$ as simple penalty terms and exploiting a \textit{proximal regularization} on the Lagrange multipliers, which provides a strongly concave function in the multipliers. The next ingredient is to employ a quadratic approximation that is strongly convex in all the players' strategies. This enables to simply deal with the nonconvexity of  $\theta_{\nu}\left(x^{\nu},{x}^{-\nu}\right)$ and $g_{i}^{\nu}\left(x^{\nu},{x}^{-\nu}\right)$ in some $x^{\nu^{\prime}}\in {x}^{-\nu}$, and further leads to a Jacobi-type \textit{decomposition} scheme for updating primal variables at each iteration.

This paper makes the following contributions to the literature:
\vspace{-0.015in}
\begin{itemize} 
	\item We introduce a new form of Lagrangian function that has a favorable structure; it is strongly concave with respect to the Lagrange multipliers, and it does not include penalty terms for handling coupling constraints. Consequently, the proposed algorithm guarantees to generate the bounded sequence of the Lagrange multipliers with the bounded primal sequence without requiring the MFCQ assumption. Moreover, this leads to an easy-to-implement algorithm by removing the computational effort in updating the penalty parameter as in \cite{facchinei2010penalty} and \cite{kanzow2016augmented}.
    
	\item The proposed algorithm can deal with the nonconvexity of each player's functions in other players' decisions by employing a quadratic approximation of the original P-Lagrangian in $\mathbf{x}$. More importantly, the use of the quadratic approximation offers a Jacobi-type decomposition scheme that allows distributed simultaneous updates of primal variables, which, to the best of our knowledge, leads to the first  \emph{distributed} algorithm to solve \textit{general} GNEPs.
    
	\item We prove that our algorithm is convergent to a saddle point of P-Lagrangian under standard assumptions. In contrast to the existing methods for solving general GNEPs, our analysis does not impose any boundedness assumptions on the iterates generated by the algorithm. In particular, we do not make use of a priori assumption that a limit point of primal iterates $\mathbf{x}^k$  exists, and safeguarding technique \cite{andreani2007augmented,andreani2008augmented} to bound multiplier iterates as in \cite{kanzow2016augmented}. In addition, we establish the global convergence under an additional assumption that the objective and constraint functions satisfy the \emph{Kurdyka-{\L}ojasiewicz} inequality.
\end{itemize}

\subsection{Notation and Outline of the Paper}
	
\subsubsection*{Notation.} We use $\mathbb{R}^{n_\nu}$ and $\mathbb{R}^{m_\nu}$ to denote the $n_\nu$-dimensional Euclidean vector space and $m_\nu$-dimensional Euclidean vector space, respectively. For two vectors $x,y \in \mathbb{R}^{n_\nu}$, the inner product is denoted by $x^T y$, and the standard Euclidean norm is denoted by $\left\| x  \right\| = \sqrt{x^T x}$. For a real scalar $z \in \mathbb{R}$, we define $\left[ z \right]^+ =\textrm{max}\left( z,0 \right)$. We use $\mathbb{R}^{m_\nu}_+$ to denote the nonnegative orthant of $\mathbb{R}^{m_\nu}$. 
\vspace{-0.05in}	
\subsubsection*{Outline of the paper.}This paper is organized as follows. In section 2, we introduce the P-Lagrangian function, describe its characteristics, and reformulate the GNEP as a saddle point problem using the P-Lagrangian. Section 3 presents a distributed primal-dual algorithm based on a quadratic approximation. In Section 4, we establish convergence of the proposed algorithm. Numerical results are presented in Section 5.

\section{Proximal-Perturbed Lagrangian Formulation} \label{sec2}

Before introducing Proximal-Perturbed Lagrangian (P-Lagrangian), we recall that under Assumption \ref{assumption1} and suitable constraint qualifications, a GNE $\mathbf{x}^{\ast}=\left(x^{1,\ast},\ldots,x^{N,\ast}\right)$ can be characterized by the Karush-Kuhn-Tucker (KKT) conditions (see e.g., \citeauthor{dreves2011solution} \cite{dreves2011solution} and \citeauthor{bueno2019optimality}  \cite{bueno2019optimality}): 
\vspace{0.05in}

\hspace{-0.1in}\textbf{The KKT conditions}.
Assume that a suitable constraint qualification holds. If there exists a point $\mathbf{x}^{\ast}=(x^{1,\ast},\ldots,x^{N,\ast})$ together with some Lagrange multipliers $\eta^{\nu,\ast}$ satisfying the KKT conditions
	\begin{equation} \label{eq:KKT}
	\begin{cases}
	0\in\nabla_{x^{\nu}}L_{0}^{\nu}(x^{\nu,\ast},x^{-\nu,\ast},\eta^{\nu,\ast}) + \mathcal{N}_{\mathcal{X}_{\nu}}\left( {x}^{\nu,\ast} \right), \quad x^{\nu,\ast} \in \mathcal{X}_{\nu}, \\
	\eta_{i}^{\nu,\ast}\geq0,\quad g_{i}^{\nu}(x^{\nu,\ast},x^{-\nu,\ast})\leq0,\quad \eta_{i}^{\nu,\ast}g_{i}^{\nu}(x^{\nu,\ast},x^{-\nu,\ast})=0, \quad\forall i=1,\ldots,m_{\nu},
	\end{cases}
	\end{equation}
	for every $\nu=1,\ldots,N$, then $\mathbf{x}^{\ast}=(x^{1,\ast},\ldots,x^{N,\ast})$ is a generalized Nash equilibrium (GNE). Here,  $L_{0}^{\nu}(x^{\nu},x^{-\nu},\eta^{\nu}) := \theta_{\nu}(x^{\nu},x^{-\nu})+\sum_{i=1}^{m_{\nu}}(\eta_{i}^{\nu}) g_{i}^{\nu}(x^{\nu},x^{-\nu})$ is each player $\nu$'s  Lagrangian function, and {\small $\mathcal{N}_{\mathcal{X}_{\nu}}( {x}^{\nu,\ast}) := \left\{\mathrm{d}_{\nu}\in\mathcal{X}_{\nu}\left|\right. \mathrm{d}_{\nu}^{T}(x^{\nu}-x^{\nu,\ast})\leq0, \forall x^{\nu} \in \mathcal{X}_{\nu} \right\}$} is the normal cone to $\mathcal{X}_{\nu}$ at $\mathbf{x}^{\ast}$.  
	
It is well known  \cite[Theorem 4.6]{facchinei2010generalized} that under the convexity assumption and a constraint qualification(CQ), the KKT conditions \eqref{eq:KKT} are necessary and sufficient optimality conditions for problem $\eqref{eq:of}$. In addition, convex optimization problem $\eqref{eq:of}$ is equivalent to solving the dual formulation, i.e., 
\begin{equation}
	\theta_{\nu}\left(\mathbf{x}^{\ast}\right) =\underset{\eta^{\nu}\geq0}{\mathrm{max}}\left(D_{0}^{\nu}\left(\eta^{\nu}\right) :=\underset{x^{\nu}\in\mathcal{X}_{\nu}}{\mathrm{min}}L_{0}^{\nu}\left(x^{\nu},x^{-\nu,\ast},\eta^{\nu}\right)\right).
\end{equation}

In general GNEP model, the set of Lagrange multipliers of each player is possibly unbounded (assuming it is nonempty) even if it satisfies the KKT conditions. This is due to the interdependency between $x^{\nu}$ and $x^{-\nu}$ through the coupling constraints $g_{i}^{\nu}\left(x^{\nu}, x^{-\nu} \right)\leq 0$, $i=1,\ldots,m_\nu$. The general GNEP setting thus requires a CQ weaker than MFCQ that allows unbounded multiplier sets (e.g., Guignard CQ, cone continuity property, and constant positive linear dependence CQ); see \cite{bueno2019optimality} for a detailed discussion about various CQs for GNEPs. This unboundedness makes the computation of a GNE very hard, and thus boundedness of the multipliers is one of the key issues when solving GNEPs. 
Our motivation for introducing a new Lagrangian is to address this challenge.
	
This section first introduces a new form of Lagrangian that has a desirable structure for equilibrium computation. We then present a reformulation of problem \eqref{eq:of} as a P-Lagrangian dual problem and show that computing a saddle point of the P-Lagrangian is equivalent to finding an equilibrium of the GNEP \eqref{eq:of}.

\subsection{The Proximal-Perturbed Lagrangian}
	
Motivated by the reformulation techniques in  \cite[Chapter~3.4]{bertsekas1989parallel}  and \cite[Chapter~3.2]{bertsekas2014constrained}, we start by transforming problem \eqref{eq:of} into an equivalent extended formulation by introducing perturbation variables $z^{\nu}=(z_{1}^{\nu},\ldots,z_{m_{\nu}}^{\nu})=0$ as additional constraints and letting $g^{\nu}(x^{\nu},x^{-\nu})\leq z^{\nu}$ given $x^{-\nu}$:
    \begin{equation} \label{eq:ef}
	\mathrm{EP}_{\nu}(x^{-\nu}) \quad\begin{aligned}
	\underset{x^{\nu} \in \mathcal{X}_{\nu}, \, z^{\nu} \in \mathbb{R}^{m_{\nu}}} {\mathrm{minimize}} & \quad\theta_{\nu}(x^{\nu},x^{-\nu})\\
	\mathrm{\mathrm{subject\:to}} & \quad g^{\nu}(x^{\nu},x^{-\nu})\leq z^{\nu}, \\
	& \hspace*{1em}z^{\nu}=0.
	\end{aligned}
    \end{equation}
Obviously, for $z^{\nu}=0$, the above extended formulation $\eqref{eq:ef}$ is equal to problem \eqref{eq:of}. 
	
Noting that the reformulation \eqref{eq:ef} allows the use of $\frac{\alpha_\nu}{2}\left\Vert{z}^{\nu}\right\Vert^{2}$ as a penalty term, first consider the following partially augmented Lagrangian for every $\nu=1,\ldots,N$:  
    \[
    {L}_{\alpha}^{\nu}(x^{\nu},x^{-\nu},z^{\nu},\lambda^{\nu},\mu^{\nu})=\theta_{\nu}(x^{\nu},x^{-\nu})+\left(\lambda^{\nu}\right)^{T}(g^{\nu}\left(x^{\nu},x^{-\nu}\right)-z^{\nu})+(\mu^{\nu})^{T}z^{\nu}+\frac{\alpha_{\nu}}{2}\left\Vert z^{\nu}\right\Vert^{2},
    \]
where $\lambda^{\nu}=(\lambda_{i}^{\nu},\ldots,\lambda_{m_{\nu}}^{\nu}) \in \mathbb{R}^{m_\nu}_{+}$ and $\mu^{\nu}=(\mu_{i}^{\nu},\ldots,\mu_{m_{\nu}}^{\nu}) \in \mathbb{R}^{m_\nu}$ are the Lagrange multipliers associated with constraints $g^{\nu}(x^{\nu},x^{-\nu})-z^{\nu}\leq0$ and $z^{\nu}=0$, respectively. $\alpha_{\nu}>0$ is a penalty parameter. Observe that given $(\lambda^{\nu},\mu^{\nu})$, minimizing $L^{\nu}_{\alpha}$ with respect to $z^{\nu}$ gives $$z^{\nu}(\lambda^{\nu},\mu^{\nu})=\frac{1}{\alpha_{\nu}}(\lambda^{\nu}-\mu^{\nu}),$$ 
which implies that $\lambda^\nu=\mu^\nu$ at the unique solution $z^{\nu,\ast}=0$. Based on this relation of $\lambda^{\nu}$ and $\mu^{\nu}$ from the optimality condition for $z^{\nu}$, we add a proximal term $-\frac{\beta_{\nu}}{2}\left\Vert\lambda^{\nu}-\mu^{\nu}\right\Vert^{2}$ to define \emph{Proximal-Perturbed Lagrangian} (P-Lagrangian): 
	
	\begin{equation} \label{eq:MAL}
	\begin{aligned}
	\mathcal{L}_{\alpha\beta}^{\nu}(x^{\nu},x^{-\nu},z^{\nu},\lambda^{\nu},\mu^{\nu}) 
	& := 
	\theta_{\nu}(x^{\nu},x^{-\nu})+\left(\lambda^{\nu}\right)^{T} (g^{\nu}\left(x^{\nu},x^{-\nu}\right)-z^{\nu}) + (\mu^{\nu})^{T}z^{\nu} \\
	& \quad 
	+ \frac{\alpha_{\nu}}{2}\left\Vert z^{\nu}\right\Vert ^{2} - \frac{\beta_{\nu}}{2}\left\Vert \lambda^{\nu} - \mu^{\nu}\right\Vert^{2},
	\end{aligned}
	\end{equation}
where $\beta_{\nu}>0$ is a proximal regularization parameter.
	
We observe that the structure of the P-Lagrangian $\mathcal{L}_{\alpha\beta}^{\nu}$ in \eqref{eq:MAL} differs from the standard augmented Lagrangian and its variants (\citealp{hestenes1969multiplier,powell1969method,bertsekas2014constrained,birgin2014practical}, and references therein). First, it is characterized by the absence of penalty term for handling the coupling constraint $g^{\nu}(x^{\nu},x^{-\nu})-z^{\nu}\leq0$. Only additional constraint $z^{\nu}=0$ is penalized with a quadratic penalty term $\frac{\alpha_{\nu}}{2}\left\Vert z^{\nu}\right\Vert ^{2}$, while the constraint $g^{\nu}(x^{\nu},x^{-\nu})-z^{\nu}\leq0$ is merely relaxed into the objective with the corresponding multiplier. Second, the P-Lagrangian is strongly concave in $\lambda^{\nu}$ (for fixed $\mu^{\nu}$) and in $\mu^{\nu}$ (for fixed $\lambda^{\nu}$) due to the presence of the negative quadratic term $-\frac{\beta_{\nu}}{2}\left\Vert \lambda^{\nu}-\mu^{\nu}\right\Vert ^{2}$. Note that, as will be shown later, this quadratic term $-\frac{\beta_{\nu}}{2}\left\Vert  \lambda^{\nu}-\mu^{\nu}\right\Vert^{2}$ plays an important role that it does not allow for the next iterate $\lambda^{\nu,k+1}$ to deviate far from $\mu^{\nu,k}$ when updating the multiplier $\lambda^\nu$ via an exact maximization scheme.

\subsection{Equivalence between a Saddle Point of P-Lagrangian  and a GNE}

Now consider the following P-Lagrangian dual problem for given $x^{-\nu}$:
	\begin{equation}\label{eq:modfiedLagrangeDual}
	\underset{\lambda^{\nu}\in\mathbb{R}_{+}^{m_{\nu}},\mu^{\nu}\in\mathbb{R}^{m_{\nu}}}{\mathrm{max}}\left\{ \mathcal{D}_{\alpha\beta}^{\nu} (\lambda^{\nu},\mu^{\nu}):= \underset{x^{\nu}\in \mathcal{X}_{\nu},z^{\nu}\in\mathbb{R}^{m_{\nu}}}{\mathrm{min}}\mathcal{L}_{\alpha\beta}^{\nu} (x^{\nu},x^{-\nu},z^{\nu},\lambda^{\nu},\mu^{\nu})\right\}.
	\end{equation}
Since $\mathcal{L}_{\alpha\beta}^{\nu}(\bullet,x^{-\nu,\ast},z^{\nu},\lambda^{\nu},\mu^{\nu})$ is convex, the primal-dual solutions of problem \eqref{eq:modfiedLagrangeDual}, $(x^{\nu,\ast},x^{-\nu,\ast},z^{\nu,\ast})$ and $(\lambda^{\nu,\ast},\mu^{\nu,\ast})$ given $x^{-\nu}=x^{-\nu,\ast}$, can be characterized by the saddle point of the P-Lagrangian.
	
\renewcommand\thetheorem{3}
\begin{definition} 
Given $x^{-\nu,\ast}$, a point  $(x^{\nu,\ast},x^{-\nu,\ast},z^{\nu,\ast},\lambda^{\nu,\ast},\mu^{\nu,\ast})$ is said to be a (parametrized) saddle point of the Proximal-Perturbed Lagrangian for $\alpha_{\nu}>0$ and $\beta_{\nu}>0$ if for every $\nu=1, \ldots, N,$
\begin{equation} \label{eq:saddle}
	\mathcal{L}_{\alpha\beta}^{\nu}(x^{\nu,\ast},x^{-\nu,\ast},z^{\nu,\ast},\lambda^{\nu},\mu^{\nu})
	\leq\mathcal{L}_{\alpha\beta}^{\nu}(x^{\nu,\ast},x^{-\nu,\ast},z^{\nu,\ast},\lambda^{\nu,\ast},\mu^{\nu,\ast})
	\leq\mathcal{L}_{\alpha\beta}^{\nu}(x^{\nu},x^{-\nu,\ast},z^{\nu},\lambda^{\nu,\ast},\mu^{\nu,\ast}),
\end{equation}
for all $(x^{\nu},z^{\nu},\lambda^{\nu},\mu^{\nu})\in \mathcal{X}_{\nu}(x^{-\nu,\ast})\times\mathbb{R}^{m_{\nu}}\times\mathbb{R}_{+}^{m_{\nu}}\times\mathbb{R}^{m_{\nu}}$. Here, $x^{-\nu,\ast}$ are viewed as parameters. 
\end{definition}
	
We establish the equivalence between computing a saddle point of $\mathcal{L}_{\alpha\beta}^{\nu}$ and finding an equilibrium of the GNEP \eqref{eq:of}, by proving the following two theorems.  
	
\renewcommand\thetheorem{1}
\begin{theorem}\label{lem_saddle_to_eq} 
	Let  $\left(x^{\nu,\ast},\boldsymbol{x}^{-\nu,\ast},z^{\nu,\ast},\lambda^{\nu,\ast},\mu^{\nu,\ast}\right)$ be a saddle point of  $\mathcal{L}_{\alpha\beta}^{\nu}\left(x^{\nu},\boldsymbol{x}^{-\nu},z^{\nu},\lambda^{\nu},\mu^{\nu}\right)$ for a given $\boldsymbol{x}^{-\nu}=\boldsymbol{x}^{-\nu,\ast}$ and for some $\alpha_{\nu},\beta_{\nu}>0$. Then,  $\mathbf{x}^{\ast}=\left(x^{\nu,\ast},\boldsymbol{x}^{-\nu,\ast}\right)$ is an equilibrium of the GNEP \eqref{eq:of}.
\end{theorem}
\begin{proof} 
     See Appendix \ref{A.1}.
\end{proof}

\renewcommand\thetheorem{2}
\begin{theorem}\label{lem_eq_to_saddle}
	Assume that $\mathbf{x}^{\ast}=\left(x^{1,\ast},\ldots,x^{N,\ast}\right)$ is an equilibrium of the GNEP \eqref{eq:of} at which the KKT conditions \eqref{eq:KKT} hold with some Lagrange multipliers $\eta^{\nu,\ast}$ for all players' optimization problems, given $\boldsymbol{x}^{-\nu}=\boldsymbol{x}^{-\nu,\ast}$. Then for every $\nu=1,\ldots,N$, there exist Lagrange multipliers $\left(\lambda^{\nu,\ast},\mu^{\nu,\ast}\right)$ such that 
\begin{equation}\label{eq:lem_eq_to_saddle_sp}
	\mathcal{L}_{\alpha\beta}^{\nu}\left(\mathbf{x}^{\ast},z^{\nu,\ast},\lambda^{\nu},\mu^{\nu}\right)\leq\mathcal{L}_{\alpha\beta}^{\nu}\left(\mathbf{x}^{\ast},z^{\nu,\ast},\lambda^{\nu,\ast},\mu^{\nu,\ast}\right)\leq\mathcal{L}_{\alpha\beta}^{\nu}\left(x^{\nu},\boldsymbol{x}^{-\nu,\ast},z^{\nu},\lambda^{\nu,\ast},\mu^{\nu,\ast}\right)
\end{equation}
for any $(x^{\nu},z^{\nu},\lambda^{\nu},\mu^{\nu})\in  \mathcal{X}_{\nu}(x^{-\nu,\ast})\times\mathbb{R}^{m_{\nu}}\times\mathbb{R}_{+}^{m_{\nu}}\times\mathbb{R}^{m_{\nu}}$.
\end{theorem}
\begin{proof}
	See Appendix \ref{A.2}.
\end{proof}

\section{Algorithm}

In this section, we propose a simple primal-dual algorithm for computing a saddle point of $\mathcal{L}_{\alpha\beta}^{\nu}$ based on a quadratic approximation of $\mathcal{L}_{\alpha\beta}^{\nu}$ for every $\nu=1,\ldots,N$.
	
\subsection{Motivation for Approximation of Subproblems}

We begin by describing briefly why we need to consider an approximation scheme for updating $\mathbf{x} = (x^{\nu}, x^{-\nu})$. To compute a saddle point of $\mathcal{L}_{\alpha\beta}^{\nu}(\mathbf{x}, z^{\nu}, \lambda^{\nu}, \mu^{\nu})$ for every $\nu=1,\ldots,N$, we should be able to  determine a point  $\widetilde{\mathbf{x}} = (\widetilde{x}^{\nu}, \widetilde{x}^{-\nu})$ that satisfies the following first-order optimality (or simultaneous stationarity) condition of subproblems for fixed $(z^{\nu},\lambda^{\nu},\mu^{\nu})$:
    \[
    \nabla_{x^{\nu}} \mathcal{L}_{\alpha\beta}^{\nu} (\widetilde{\mathbf{x}},z^{\nu},\lambda^{\nu},\mu^{\nu})^T(x^{\nu} -\widetilde{x}^{\nu}) \geq 0, \ \ \ \ \forall x^\nu \in \mathcal{X}_\nu \ \ \textrm{for} \ \textrm{all} \ \nu=1,\ldots,N.
    \] 
It is well known (Facchinei and Pang \cite{facchinei2007finite}) that for given $(z^{\nu,k},\lambda^{\nu,k},\mu^{\nu,k})$, computing such a stationary point is equivalent to the  variational inequality (VI) problem  of finding $\widetilde{\mathbf{x}}\in \mathbf{X}$ such that
    \[
    \mathbf{L} \left(\widetilde{\mathbf{x}},z^k, \lambda^k, \mu^k \right)^{T}
    \left(\mathbf{x} -\widetilde{\mathbf{x}} \right) \geq 0, \quad \forall \mathbf{x} \in \mathbf{X},
    \]
where $\mathbf{X} := \prod_{\nu=1}^{N} \mathcal{X}_{\nu}$, the Cartesian product of the private strategy sets of all players, and the mapping $\mathbf{L} \left(\mathbf{x},z^k, \lambda^k, \mu^k \right): \mathbf{X} \rightarrow \mathbb{R}^n$ is given by 
    \[
    \mathbf{L} \left(\mathbf{x},z^k, \lambda^k, \mu^k \right) =
    \begin{bmatrix}
    \nabla_{x^{1}} \mathcal{L}_{\alpha\beta}^{1} \left( x^{1},x^{-1},z^{1,k},\lambda^{1,k},\mu^{1,k}\right)\\
    \vdots \\
    \nabla_{x^{N}} \mathcal{L}_{\alpha\beta}^{N} \left( x^{N},x^{-N},z^{N,k},\lambda^{N,k},\mu^{N,k}\right)
    \end{bmatrix},
    \]
with $z= \left[ (z^1)^T,\ldots,(z^N)^T\right]$, $\lambda=\left[ (\lambda^1)^T,\ldots,(\lambda^N)^T\right]^T$ and  $\mu=\left[ (\mu^1)^T,\ldots,(\mu^N)^T\right]^T$.

However, it may be difficult to compute the point $\widetilde{\mathbf{x}}$ using descent methods. In the GNEP setting, the monotonicity of the mapping  $\mathbf{L} (\mathbf{x},z^{k}, \lambda^{k}, \mu^{k})$ with respect to $\mathbf{x}=\left(x^{\nu},x^{-\nu}\right)$ does not hold in general \cite[Section 5.2]{facchinei2010generalized} even if  each component $\nabla_{x^{\nu}} \mathcal{L}_{\alpha\beta}^{\nu} ( x^{\nu},x^{-\nu},z^{\nu},\lambda^{\nu},\mu^{\nu})$ is convex in $x^{\nu}$. The nonconvexity of each P-Lagrangian with respect to the other players' variables makes it hard to preserve a descent direction for the convergence to the stationary point $\widetilde{\mathbf{x}}$ that satisfies all components of the variational inequality.

\subsection{Construction of Quadratic Approximation Model}

To overcome such a computational difficulty, we consider a  monotone approximation, denoted by $\widehat{\mathbf{L}}^k$, to the nonmonotone mapping $\mathbf{L}$ in $\mathbf{x}$. The monotone approximation $\widehat{\mathbf{L}}^k$ of the mapping  $\mathbf{L}$ can be always chosen even if $\mathbf{L}$ is nonmonotone (see e.g., \citet{chung2010subproblem, luna2014class}). Furthermore, strongly monotone approximation mapping can be derived by replacing each player's $\mathcal{L}_{\alpha\beta}^{\nu}$ by a simple approximate function and then constructing an approximation $\widehat{\mathbf{L}}^k$.  

To this end, inspired by \citet{beck2009fast} and \citet{bolte2014proximal}, we first employ the following quadratic approximation $\widehat{\mathcal{L}_{\alpha\beta}^{\nu}}$ in only $\mathbf{x}$, at a given point $\mathbf{y}$:
\begin{multline}\label{eq:approx_AL}
\widehat{\mathcal{L}_{\alpha\beta}^{\nu}}(\mathbf{x},z^{\nu},\lambda^{\nu},\mu^{\nu};\mathbf{y})  
:= \mathcal{L}_{\alpha\beta}^{\nu}(\mathbf{y},z^{\nu},\lambda^{\nu},\mu^{\nu})
+\nabla_{x^{\nu}}\mathcal{L}_{\alpha\beta}^{\nu}(\mathbf{y},z^{\nu},\lambda^{\nu},\mu^{\nu})^{T} (x^{\nu}-y^{\nu}) +\frac{\gamma_{\nu}}{2}\left\Vert x^{\nu}-y^{\nu}\right\Vert ^{2}\\
+\sum_{\nu^{\prime}\neq\nu}\nabla_{x^{\nu^{\prime}}}\mathcal{L}_{\alpha\beta}^{\nu}(\mathbf{y},z^{\nu},\lambda^{\nu},\mu^{\nu})^{T} (x^{\nu^{\prime}}-y^{\nu^{\prime}})+\frac{\gamma_{\nu}}{2}\sum_{\nu^{\prime}\neq\nu}\left\Vert x^{\nu^{\prime}}-y^{\nu^{\prime}}\right\Vert ^{2},
\end{multline}
namely, the linearized P-Lagrangian $\mathcal{L}_{\alpha\beta}^{\nu}$ at the point $\mathbf{y}$ combined with quadratic proximal terms that measure the local error in the linear approximation. Here, $\gamma_\nu>0$ is a proximal parameter. The term $\sum_{\nu^{\prime}\neq\nu}\nabla_{x^{\nu^{\prime}}}\mathcal{L}_{\alpha\beta}^{\nu}\left(\mathbf{y},z^{\nu},\lambda^{\nu},\mu^{\nu}\right)=\nabla_{x^{-\nu}}\mathcal{L}_{\alpha\beta}^{\nu}\left(\mathbf{y},z^{\nu},\lambda^{\nu},\mu^{\nu}\right)$ represents the gradient at a given point $\mathbf{y}\in\mathbb{R}^{n}$ in other players' strategies, and $\nabla_{x^{\nu}}\mathcal{L}_{\alpha\beta}^{\nu}\left(\mathbf{y},z^{\nu},\lambda^{\nu},\mu^{\nu}\right)$ denotes the gradient of $\mathcal{L}_{\alpha\beta}^{\nu}$ with respect to $x^\nu$ at the point $\mathbf{y}$. 

From the conditions \eqref{eq:assumption_lipschitz_1}--\eqref{eq:assumption_lipschitz_21} in  Assumption \ref{assumption_lipschitz}, we know that  $\theta_\nu$ and $g^\nu$ have Lipschitz continuous gradients; there exist Lipschitz constants $L_{\nabla\theta_{\nu}}>0$ and $L_{\nabla{g^{\nu}}}>0$ such that 
	\begin{subequations}
		\begin{align}
		\left\Vert \nabla_{\mathbf{x}}{\theta_{\nu}\left(\mathbf{x}_{1}\right)}-\nabla_{\mathbf{x}}{\theta_{\nu}\left(\mathbf{x}_{2}\right)}\right\Vert 
		&\leq L_{\nabla\theta_{\nu}}\left\Vert \mathbf{x}_{1}-\mathbf{x}_{2}\right\Vert ,\quad\forall\mathbf{x}_{1},\mathbf{x}_{2}\in\mathbf{X}, \label{eq:uniform_lipsch1} \\
		\left\Vert \nabla_{\mathbf{x}}{g^{\nu}\left(\mathbf{x}_{1}\right)}-\nabla_{\mathbf{x}}{g^{\nu}\left(\mathbf{x}_{2}\right)}\right\Vert 
		&\leq L_{\nabla{g^{\nu}}}\left\Vert \mathbf{x}_{1}-\mathbf{x}_{2}\right\Vert ,\quad\forall\mathbf{x}_{1},\mathbf{x}_{2}\in\mathbf{X}, \label{eq:uniform_lipsch2} 
	\end{align}
\end{subequations}		
where  $L_{\nabla{\theta_{\nu}}}=L_\nu(\theta_\nu) + L_{-\nu}(\theta_{\nu})$ and $L_{\nabla{g^{\nu}}}=L_\nu(g^\nu) + L_{-\nu}(g^{\nu})$ (see \citet[Lemma 2]{nesterov2012efficiency}). Here,  $\nabla_{\mathbf{x}}\theta_{\nu}(\mathbf{x})$ and $\nabla_{\mathbf{x}}g^{\nu}(\mathbf{x})$ represent $\left[\nabla_{x^{1}}\theta_{\nu}(\mathbf{x})^T,\ldots,\nabla_{x^{N}}\theta_{\nu}(\mathbf{x})^T\right]^T$ and $\left[\nabla_{x^{1}}g^{\nu}(\mathbf{x})^T,\ldots,\nabla_{x^{N}}g^{\nu}(\mathbf{x})^T\right]^T$, respectively. 
As a direct consequence of the above Lipschitz  continuity of $\nabla_{\mathbf{x}}\theta_\nu(\mathbf{x})$ and $\nabla_{\mathbf{x}}g^\nu(\mathbf{x})$,   \eqref{eq:uniform_lipsch1} and  \eqref{eq:uniform_lipsch2} respectively, we have the well-known descent Lemma. 
\renewcommand\thetheorem{1}
\begin{lemma}[Bertsekas {\cite[Proposition A.24]{bertsekas1999nonlinear}}]\label{lem_descent} 
For $\nu=1,\ldots,N$ and for any fixed  $\left(z^{\nu},\lambda^{\nu},\mu^{\nu}\right)$, $\nabla_{\mathbf{x}} \mathcal{L}_{\alpha\beta}^{\nu}$ is Lipschitz continuous with constant $L_{\nu}>0$. We thus have
\[
	\mathcal{L}_{\alpha\beta}^{\nu}(\mathbf{x}_{1}) \leq\mathcal{L}_{\alpha\beta}^{\nu}(\mathbf{x}_{2}) +\nabla_{\mathbf{x}}\mathcal{L}_{\alpha\beta}^{\nu}(\mathbf{x}_{2})^{T}(\mathbf{x}_{1}-\mathbf{x}_{2}) +\frac{L_{\nu}}{2}\left\Vert \mathbf{x}_{1}-\mathbf{x}_{2}\right\Vert ^{2},\quad \forall\mathbf{x}_{1},\mathbf{x}_{2}\in\mathbf{X}.
\]
Here, we omit fixed $\left(z^{\nu},\lambda^{\nu}, \mu^{\nu}\right)$ for notational simplicity.
\end{lemma}

Then, with the proximal parameter $\gamma_\nu$ large enough such that $\gamma_\nu \geq L_\nu$, $\widehat{\mathcal{L}_{\alpha\beta}^{\nu}}(\mathbf{x},z^{\nu},\lambda^{\nu},\mu^{\nu};\mathbf{y})$ in \eqref{eq:approx_AL} is an \emph{upper quadratic} approximation of $\mathcal{L}_{\alpha\beta}^{\nu}(\bullet,z^{\nu},\lambda^{\nu},\mu^{\nu})$ around the point $\mathbf{y}$ with respect to $\mathbf{x}=(x^\nu, x^{-\nu})$ and it has the following properties (see e.g., \cite{beck2009fast,razaviyayn2013unified,scutari2016parallel}).
\renewcommand\thetheorem{1}
\begin{remark} [Properties of $\widehat{\mathcal{L}_{\alpha\beta}^{\nu}}$] \label{rem_prop_approx_AL}
The approximation function $\widehat{\mathcal{L}_{\alpha\beta}^{\nu}}$ with $\gamma_\nu \geq L_\nu$ satisfies the properties:
	
\begin{enumerate}[label=(\textrm{P\arabic*})]	
		\item \label{itm:p1} $\widehat{\mathcal{L}_{\alpha\beta}^{\nu}}\left(\mathbf{y},z^{\nu},\lambda^{\nu},\mu^{\nu};\mathbf{y}\right)
		=
		\mathcal{L}_{\alpha\beta}^{\nu}\left(\mathbf{y},z^{\nu},\lambda^{\nu},\mu^{\nu}\right)$
		for $\forall\mathbf{y}\in\mathbf{X}$. 
		
		\item \label{itm:p2}
		$\widehat{\mathcal{L}_{\alpha\beta}^{\nu}}\left(\mathbf{x},z^{\nu},\lambda^{\nu},\mu^{\nu};\mathbf{y}\right)
		\geq
		\mathcal{L}_{\alpha\beta}^{\nu}\left(\mathbf{y},z^{\nu},\lambda^{\nu},\mu^{\nu}\right)$
		for $\forall\mathbf{x},\mathbf{y}\in\mathbf{X}$. 
		
		\item \label{itm:p3} $\widehat{\mathcal{L}_{\alpha\beta}^{\nu}}(\bullet,z^{\nu},\lambda^{\nu},\mu^{\nu};\mathbf{y}) $ is strongly convex in all the players' decisions $\mathbf{x}=(x^{\nu},x^{-\nu})$ with constant $c_{\nu}>0$, i.e., for any $\mathbf{x}_{1},\mathbf{x}_{2}\in\mathbf{X}$,
		 \[
		\left(\nabla_{\mathbf{x}}\widehat{\mathcal{L}_{\alpha\beta}^{\nu}}(\mathbf{x}_{1},z^{\nu},\lambda^{\nu}, \mu^{\nu};\mathbf{y})-
		\nabla_{\mathbf{x}}\widehat{\mathcal{L}_{\alpha\beta}^{\nu}}(\mathbf{x}_{2},z^{\nu},\lambda^{\nu}, \mu^{\nu};\mathbf{y})\right)^{T}\left(\mathbf{x}_{1}-\mathbf{x}_{2} \right) 
		\geq 
		c_{\nu}\left\|\mathbf{x}_{1}-\mathbf{x}_{2}\right\|^{2}.
		\]
	
		\item \label{itm:p4}   $\nabla_{\mathbf{x}}\widehat{\mathcal{L}_{\alpha\beta}^{\nu}}
		=\left[\nabla_{x^{1}}\widehat{\mathcal{L}_{\alpha\beta}^{\nu}}^T,\ldots,
		\nabla_{x^{N}}\widehat{\mathcal{L}_{\alpha\beta}^{\nu}}^T\right]^{T}$ is Lipschitz continuous on $\mathbf{X}$  with some Lipschitz constant $\widehat{L}_{\nu} \geq \gamma_{\nu}$, i.e., for any $\mathbf{x}_{1},\mathbf{x}_{2}\in\mathbf{X}$, 
		\[
		\left\Vert \nabla_{\mathbf{x}}\widehat{\mathcal{L}_{\alpha\beta}^{\nu}}\left(\mathbf{x}_{1},z^{\nu},\lambda^{\nu}, \mu^{\nu};\mathbf{y}\right)-\nabla_{\mathbf{x}}\widehat{\mathcal{L}_{\alpha\beta}^{\nu}}\left(\mathbf{x}_{2},z^{\nu},\lambda^{\nu},\mu^{\nu};\mathbf{y}\right)\right\Vert 
		\leq
		\widehat{L}_{\nu}\left\Vert \mathbf{x}_{1}-\mathbf{x}_{2}\right\Vert.
		\]
\end{enumerate}
\end{remark}
The properties \ref{itm:p1} and \ref{itm:p2} imply that $\widehat{\mathcal{L}_{\alpha\beta}^{\nu}}$ with  $\gamma_\nu \geq L_\nu$ is a tight upper bound of  $\mathcal{L}_{\alpha\beta}^{\nu}$  around the given point $\mathbf{y}$. The properties \ref{itm:p3} and \ref{itm:p4} are from the structure of $\widehat{\mathcal{L}_{\alpha\beta}^{\nu}}$ that is the first-order approximation of $\mathcal{L}_{\alpha\beta}^{\nu}$ in $\mathbf{x}$ at $\mathbf{y}$ with quadratic term $\frac{\gamma_{\nu}}{2}\left\Vert \mathbf{x}-\mathbf{y}\right\Vert ^{2}$.

Given the current iterates
$\mathbf{y}=\mathbf{x}^{k}$ and  $(z^{\nu,k},\lambda^{\nu,k},\mu^{\nu,k})$, since $\widehat{\mathcal{L}_{\alpha\beta}^{\nu}}(\bullet,z^{\nu,k},\lambda^{\nu,k},\mu^{\nu,k};\mathbf{x}^{k})$ is uniformly strongly convex on $\mathbf{X}$, there must exist a unique  minimizer $\widehat{\mathbf{x}}^k = (\widehat{x}^{\nu,k}, \widehat{x}^{-\nu,k})$ at each iteration $k$ such that  
\[
\nabla_{x^{\nu}} \widehat{\mathcal{L}_{\alpha\beta}^{\nu}} \left(\widehat{\mathbf{x}}^k, z^{\nu,k},\lambda^{\nu,k},\mu^{\nu,k};\mathbf{x}^k\right)^T
\left(x^{\nu} - \widehat{x}^{\nu,k}\right)  
\geq 0, \quad \nu=1,\ldots,N.
\]
It also follows from \ref{itm:p1} that
\[
\widehat{\mathcal{L}_{\alpha\beta}^{\nu}} \left(\widehat{\mathbf{x}}^k, z^{\nu,k},\lambda^{\nu,k},\mu^{\nu,k};\mathbf{x}^k\right) \leq  \mathcal{L}_{\alpha\beta}^{\nu} \left(\mathbf{x}^k, z^{\nu,k},\lambda^{\nu,k},\mu^{\nu,k}\right).
\]
We can construct a (strongly) monotone approximation mapping $\widehat{\mathbf{L}}^{k}: \mathbf{X} \rightarrow \mathbb{R}^{n}$ given by
\[
 \widehat{\mathbf{L}}^{k} \left(\mathbf{x},z^{k}, \lambda^{k}, \mu^{k};\mathbf{x}^{k} \right) := 
 \begin{bmatrix}
 \nabla_{x^{1}} \widehat{\mathcal{L}_{\alpha\beta}^{1}} \left( x^{1},x^{-1},z^{1,k},\lambda^{1,k},\mu^{1,k}; \mathbf{x}^{k}\right) \\
 \vdots \\
 \nabla_{x^{N}} \widehat{\mathcal{L}_{\alpha\beta}^{N}} \left( x^{N},x^{-N},z^{N,k},\lambda^{N,k},\mu^{N,k};\mathbf{x}^{k}\right) 
 \end{bmatrix}.
\]
Let us now consider solving the following approximate variational inequality problem  $\mathrm{VI}^k(\mathbf{X},\widehat{\mathbf{L}}^k)$  of finding  $\widehat{\mathbf{x}}^{k}$: 
 \begin{equation} \label{eq:approx_VI}
  \mathrm{VI}^k(\mathbf{X},\widehat{\mathbf{L}}^k): \quad
    \widehat{\mathbf{L}}^k \left(\widehat{\mathbf{x}}^k, z^k, \lambda^k, \mu^k; \mathbf{x}^k \right)^T \left(\mathbf{x} -\widehat{\mathbf{x}}^k \right) \geq 0, \quad \forall \mathbf{x} \in \mathbf{X}.
 \end{equation}
It is well known (\cite[Proposition 1.5.8]{facchinei2007finite}) that $\widehat{\mathbf{x}}^{k}$ is also a solution to the system of fixed-point subproblem (or system of nonlinear projected equations) at iteration $k$:
\begin{equation} \label{eq:fixed_point}
    \widehat{\mathbf{x}}^{k}
    -\mathcal{P}_{\mathbf{X}}\left[\widehat{\mathbf{x}}^{k} - \sigma \widehat{\mathbf{L}}^{k}\left( \widehat{\mathbf{x}}^{k}, z^{k}, \lambda^{k}, \mu^{k}; \mathbf{x}^{k}\right) \right]  = 0,
\end{equation}
where $\mathcal{P}_{\mathbf{X}}(x)=\textrm{argmin}\left\{\left\| x - y \right\| \left|\right.   y \in \mathbf{X} \right\} $ denotes the Euclidean projection operator onto the set $\mathbf{X}$ and $\sigma>0$ is a constant. 

For fixed $(\mathbf{x}^{k},z^k, \lambda^{k}, \mu^{k})$ at iteration $k$, we use the following gradient projection to generate a sequence  $\left\lbrace \mathbf{u}^{k,l} \right\rbrace$ in  inner iterations $l=0,1,2,\ldots$
 \begin{equation} \label{eq:PD_VI}
  {\mathbf{u}}^{k,l+1}=\mathcal{P}_{\mathbf{X}}\left[\mathbf{u}^{k,l} - \sigma \widehat{\mathbf{L}}^{k}\left( \mathbf{u}^{k,l}, z^{k}, \lambda^{k}, \mu^{k}; \mathbf{x}^{k}\right)  \right],
 \end{equation}
equivalently,
\begin{equation} \label{eq:distributed_pd_vi}
\begin{aligned} 
\mathbf{u}^{k,l+1} = 
\begin{pmatrix}
u^{1,k,l+1} \\
\vdots \\
u^{\nu,k,l+1}\\
\vdots \\
u^{N,k,l+1}
\end{pmatrix}=
\begin{pmatrix}
& \mathcal{P}_{\mathcal{X}_{1}} \left[ u^{1,k,l} - \sigma  \left(\nabla_{x^{1}} \mathcal{L}_{\alpha\beta}^{1}\left(\mathbf{x}^{k},z^{1,k},\lambda^{1,k},\mu^{1,k}\right) + \gamma_{1} \left(u^{1,k,l}-x^{1,k}\right) \right) \right] \\
& \vdots \\
& \mathcal{P}_{\mathcal{X}_{\nu}} \left[ u^{\nu,k,l} - \sigma  \left(\nabla_{x^{\nu}} \mathcal{L}_{\alpha\beta}^{\nu}\left(\mathbf{x}^{k},z^{\nu,k},\lambda^{\nu,k},\mu^{\nu,k}\right) + \gamma_{\nu} \left(u^{\nu,k,l}-x^{\nu,k}\right) \right) \right]  \\
&  \vdots \\ 
& \mathcal{P}_{\mathcal{X}_{N}} \left[ u^{N,k,l} - \sigma  \left(\nabla_{x^{N}} \mathcal{L}_{\alpha\beta}^{N}\left(\mathbf{x}^{k},z^{N,k},\lambda^{N,k},\mu^{N,k}\right) + \gamma_{N} \left(u^{N,k,l}-x^{N,k}\right) \right) \right]
\end{pmatrix}.
\end{aligned}
\end{equation}

Notice that the structure of $\nabla_{x^{\nu}} \widehat{\mathcal{L}_{\alpha\beta}^{\nu}}$ allows for the inner gradient projection scheme \eqref{eq:distributed_pd_vi} to be implemented in a distributed way since each player $\nu$ can update its own $u^{\nu,k,l}$ while keeping the current primal iterates $\mathbf{x}^k=(x^{\nu,k}, x^{-\nu,k})$ fixed. Thus we can allow each player $\nu$ to choose its own step size $\sigma_\nu$, $\nu=1,\ldots, N$.

We also note that when the private strategy set of each player $\nu$ includes functional constraints $c_{j}^\nu(x^\nu)\leq0, \ j=1,\ldots,p_\nu$, they are treated in the same way to handle $g^{\nu}(x^{\nu},x^{-\nu})\leq0$ via the P-Lagrangian. It follows that only the set $\mathcal{X}_{\nu}$ remains as a simple constraint, and hence the projection onto $\mathcal{X}_{\nu}$ is computationally cheap.

The following Lemma shows that the inner gradient projection scheme \eqref{eq:distributed_pd_vi} converges to the solution $\widehat{\mathbf{x}}^{k}$ of the subproblem \eqref{eq:fixed_point} at each iteration $k$ and thus enables us to compute a point satisfying the decrease property for every  $\mathcal{L}_{\alpha\beta}^{\nu}$ during inner iterations..

\renewcommand\thetheorem{2}
\begin{lemma} \label{prop_converge_inner}
Let $\widehat{\mathbf{x}}^{k}$ be the unique solution to the subproblem  \eqref{eq:fixed_point} and $\mathbf{x}^{k} \neq \widehat{\mathbf{x}}^{k}$. Let $\{ \mathbf{u}^{k,l} \}_{l \geq 1}$ be the sequence generated by the inner gradient projection  \eqref{eq:distributed_pd_vi} with the
step size $\sigma_\nu$ for each player $\nu$. Suppose that the parameter $\gamma_\nu>0$ of proximal term $\frac{\gamma_{\nu}}{2}\left\Vert \mathbf{x}-\mathbf{x}^{k}\right\Vert^{2}$ in $\widehat{\mathcal{L}_{\alpha \beta}^{\nu}}$ is chosen such that $\gamma_\nu \geq L_{\nu}$, where $L_{\nu}$ is the Lipschitz constant of  $\nabla_{\mathbf{x}}\mathcal{L}_{\alpha\beta}^{\nu}$. Then,

\begin{enumerate}[label=(\alph*)]

\item \label{itm:inner_a} 
for $\widehat{\sigma}:=\underset{\nu=1,\ldots,N}{\mathrm{max}}\sigma_{\nu}$ satisfying
$0<\widehat{\sigma}<(2\gamma_{\mathrm{min}}^2)/\widehat{L}_{\mathrm{max}}$, where  $\gamma_{\mathrm{min}}=\underset{\nu=1,\ldots,N}{\mathrm{min}}\gamma_{\nu}$,  $\widehat{L}_{\mathrm{max}}=\underset{\nu=1,\ldots,N}{\mathrm{max}}\widehat{L}_{\nu}$, and   $\widehat{L}_{\nu}>0$ is the Lipschitz constant of $\nabla_{\mathbf{x}}\widehat{\mathcal{L}_{\alpha \beta}^{\nu}}$, the sequence $\{ \mathbf{u}^{k,l} \}_{l \geq 1}$ converges to the solution $\widehat{\mathbf{x}}^{k}$. That is, 
\begin{equation} \label{eq:prop_converge_result}
	  \left\| {\mathbf{u}}^{k,l+1} - \widehat{\mathbf{x}}^{k}  \right\|
	  \leq
	  \boldsymbol{\tau} \left\| \mathbf{u}^{k,l} - \widehat{\mathbf{x}}^{k}  \right\|, \quad 0<\boldsymbol{\tau}<1,
	\end{equation}
where $\boldsymbol{\tau} = \sqrt{1-2 \gamma_{\mathrm{min}} \widehat{\sigma} + \widehat{\sigma}^2 \widehat{L}_{\mathrm{max}}}$. 
		
\item \label{itm:inner_b}  thus, the inner gradient projection \eqref{eq:distributed_pd_vi} can compute $\mathbf{u}^{k,l+1}$ close to $\widehat{\mathbf{x}}^{k}$ such that
\begin{equation} \label{eq:inner_b}
	\mathcal{L}_{\alpha\beta}^{\nu}\left(\mathbf{u}^{k,l+1},z^{\nu,k},\lambda^{\nu,k},\mu^{\nu,k}\right)
	<
	\mathcal{L}_{\alpha\beta}^{\nu}\left(\mathbf{x}^{k},z^{\nu,k},\lambda^{\nu,k},\mu^{\nu,k}\right) 
\end{equation}
for every $\nu=1,\ldots,N$ in a finite number of inner iterations.
\end{enumerate}		
\end{lemma}

\begin{proof}
	See Appendix \ref{B}.
\end{proof}

\subsection{Description of Algorithm}

We are ready to formally present our distributed algorithm that exploits all the features discussed. The steps of the proposed algorithm are summarized in Algorithm \ref{algorithm1}.

\begin{algorithm} 
	\caption{P-Lagrangian based Alternating Direction Algorithm (PL-ADA)}\label{algorithm1}
	
	{Set $k=0$ and define initial variables $\left({x}^{\nu,0},z^{\nu,0},\lambda^{\nu,0},\mu^{\nu,0}\right)$ with $\lambda^{\nu,0}=\mu^{\nu,0}$, $\nu=1,\ldots,N$. \\
	Set $\sigma_\nu>0$ and parameters $\alpha_{\nu}>0$ and  $\beta_{\nu}>0$. \\
	} 
	\vspace{0.05in}
		
	\begin{algorithmic}	
		\STATE\label{step1}{\textbf{Step 1}.	
		Let iteration $k$ be fixed, and let $\mathbf{u}^{k,0}=\mathbf{x}^{k}$.  \\
		\vspace{0.025in}
		For every $\nu=1,\ldots,N$, and for fixed $\left(\mathbf{x}^{k},z^{\nu,k},\lambda^{\nu,k},\mu^{\nu,k}\right)$, compute  ${u}^{\nu,k,l+1}$  according to the following gradient projection step for inner iteration $l=0,1,2,\ldots$} \\
		\vspace{0.05in}
		
		\textbf{while}
		{\small $\left\| \mathcal{P}_{\mathbf{X}} \left[ \mathbf{u}^{k,l+1} - \sigma \widehat{\mathbf{L}}^{k}\left(\mathbf{u}^{k,l+1},z^{k},\lambda^{k},\mu^{k}; \mathbf{x}^{k}\right) \right] - \mathbf{u}^{k,l+1} \right\| > \varepsilon$} or \\
		\vspace{0.05in}
		
		\hspace{1.83in}
		{\small $\widehat{\mathcal{L}_{\alpha\beta}^{\nu}}\left(\mathbf{u}^{k,l+1},z^{\nu,k},\lambda^{\nu,k},\mu^{\nu,k}; \mathbf{x}^{k}\right) -\mathcal{L}_{\alpha\beta}^{\nu}\left(\mathbf{x}^{k},z^{\nu,k},\lambda^{\nu,k},\mu^{\nu,k}\right) \geq  0$} \textbf{do}
		\vspace{0.025in}
		\[
		\begin{aligned}
		{u}^{\nu,k,l+1} 
		= \mathcal{P}_{\mathcal{X}_{\nu}} \left[ u^{\nu,k,l} - \sigma_{\nu}  \left(\nabla_{x^{\nu}} \mathcal{L}_{\alpha\beta}^{\nu}\left(\mathbf{x}^{k},z^{\nu,k},\lambda^{\nu,k},\mu^{\nu,k}\right) + \gamma_{\nu} \left(u^{\nu,k,l}-x^{\nu,k}\right) \right) \right] \\
		\end{aligned}
		\]
		\vspace{0.025in}
		\textbf{end while}
		\vspace{0.05in}
		
		{Set 
		$\mathbf{x}^{k+1}={\mathbf{u}}^{k,l+1}
		:=\left[({u}^{1,k,l+1})^T,\ldots, ({u}^{N,k,l+1})^T\right]^{T},$
		and go to \textbf{Step 2}.}
		\vspace{0.05in}
		
		\STATE\label{step2} {\textbf{Step 2}. For $\nu=1,\ldots,N$, compute $z^{\nu,k+1}$ by an exact minimization step on $\mathcal{L}_{\alpha\beta}^{\nu}$}
		\[\label{eq:update_z}
		z^{\nu,k+1}=\underset{z^{\nu}\in\mathbb{R}^{m_{\nu}}}{\mathrm{arg\,min}}\left\{ \mathcal{L}_{\alpha\beta}^{\nu}\left(\mathbf{x}^{k+1},z^{\nu},\lambda^{\nu,k},\mu^{\nu,k}\right)\right\} 
		={\left(\lambda^{\nu,k}-\mu^{\nu,k}\right)}/{\alpha_{\nu}}.
		\]
		\STATE\label{step3} {\textbf{Step 3}. For $\nu=1,\ldots,N$, update $(\lambda^{\nu,k+1}, \mu^{\nu,k+1})$
			by exact maximization steps on $\mathcal{L}_{\alpha\beta}^{\nu}$}
		\[\begin{aligned} \label{eq:update_lambda}
		\lambda^{\nu,k+1} & =  \underset{\lambda^{\nu}\in\mathbb{R}^{m_{\nu}}_{+}}{\mathrm{arg\,max}}\left\{\mathcal{L}_{\alpha\beta}^{\nu}\left(\mathbf{x}^{k+1},z^{\nu,k+1},\lambda^{\nu},\mu^{\nu,k}\right)\right\}=\left[\mu^{\nu,k}+\frac{1}{\beta_{\nu}} g^{\nu}(\mathbf{x}^{k+1}) \right]^{+}. \\
		\mu^{\nu,k+1} &= \underset{\mu^{\nu}\in\mathbb{R}^{m_{\nu}}}{\mathrm{arg\,max}}\left\{\mathcal{L}_{\alpha\beta}^{\nu}\left(\mathbf{x}^{k+1},z^{\nu,k+1},\lambda^{\nu,k+1},\mu^{\nu}\right)\right\} 
		=\lambda^{\nu,k+1}.
		\end{aligned}\]
		
		\STATE{\textbf{Step 4}. Set $k\leftarrow k+1$ and go to Step 1.}
		\vspace{0.05in}
		
	\end{algorithmic}
\end{algorithm}
	
The main computational effort of our algorithm is involved in Step \ref{step1} to update primal iterates from $\mathbf{x}^{k}$ to $\mathbf{x}^{k+1}$. If $\widehat{\mathbf{x}}^k \neq \mathbf{x}^k$, by Lemma \ref{prop_converge_inner}, we can always find a point $\mathbf{u}^{k,l+1}$ that satisfies both conditions
\begin{equation} \label{eq:step1_con1}
\left\| \mathcal{P}_{\mathbf{X}} \left[ \mathbf{u}^{k,l+1} - \sigma \widehat{\mathbf{L}}^{k}\left(\mathbf{u}^{k,l+1},z^{k},\lambda^{k},\mu^{k}; \mathbf{x}^{k}\right) \right] - \mathbf{u}^{k,l+1} \right\| \leq \varepsilon
\end{equation}
and
\begin{equation}\label{eq:step1_con2}
\widehat{\mathcal{L}_{\alpha\beta}^{\nu}} \left( {\mathbf{u}}^{k,l+1},z^{\nu,k},\lambda^{\nu,k},\mu^{\nu,k};\mathbf{x}^{k}\right) <   \mathcal{L}_{\alpha\beta}^{\nu} \left( \mathbf{x}^{k},z^{\nu,k},\lambda^{\nu,k},\mu^{\nu,k}\right), \quad  \nu=1,\ldots,N,
\end{equation}
in a finite number of inner iterations. When the descent condition  \eqref{eq:step1_con2} is satisfied, ${\mathbf{u}}^{k,l+1}$ is set to $\mathbf{x}^{k+1}$. Consequently, the decrease of every ${\mathcal{L}_{\alpha\beta}^{\nu}}\left(\mathbf{x}^{k},z^{\nu,k},\lambda^{\nu,k},\mu^{\nu,k}\right)$ value is obtained, that is, 
\begin{equation}
\mathcal{L}_{\alpha\beta}^{\nu} \left( \mathbf{x}^{k+1},z^{\nu,k},\lambda^{\nu,k},\mu^{\nu,k}\right) 
\leq
\widehat{\mathcal{L}_{\alpha\beta}^{\nu}} \left( {\mathbf{x}}^{k+1},z^{\nu,k},\lambda^{\nu,k},\mu^{\nu,k};\mathbf{x}^{k}\right)
<   \mathcal{L}_{\alpha\beta}^{\nu} \left( \mathbf{x}^{k},z^{\nu,k},\lambda^{\nu,k},\mu^{\nu,k}\right). \notag
\end{equation}
for $\nu=1,\ldots,N$ (see Lemma \ref{prop_converge_inner} \ref{itm:inner_b}).

We remark that a point satisfying the (approximate) fixed-point condition \eqref{eq:step1_con1} does not necessarily guarantee that the descent condition \eqref{eq:step1_con2} holds. Hence, the algorithm keeps updating the iterates $\mathbf{u}^{k,l}$ until condition  \eqref{eq:step1_con2} is satisfied even after condition  \eqref{eq:step1_con1} is met, which may require many inner iterations.

The next step consists of each player $\nu$ updating $z^{\nu}$ by taking a simple minimization step (Step \ref{step2}) on  $\mathcal{L}_{\alpha \beta}^{\nu}$. This update depends on only the current iterates of the Lagrange multipliers $\lambda^{\nu,k}$ and $\mu^{\nu,k}$, but is independent of the primal variables $\mathbf{x}$. 

After the minimization steps have been carried out,
given $\left(\mathbf{x}^{k+1}, z^{\nu,k+1}\right)$,  the multipliers are updated by exact maximization steps on $\mathcal{L}_{\alpha\beta}^{\nu}$. The updates of $\lambda^{\nu}$ and $\mu^{\nu}$ take the explicit forms:
\[
\lambda^{\nu,k+1}  =\left[\mu^{\nu,k}+ \frac{1}{\beta_{\nu}}g^{\nu}(\mathbf{x}^{k+1}) \right]^{+}, \quad \; \mu^{\nu,k+1}=\lambda^{\nu,k+1},
\]
which can be viewed as a proximal point scheme. The multipliers $\left(\lambda^{\nu},\mu^{\nu}\right)$ are always updated whenever the corresponding $\mathbf{x}=(x^{\nu},x^{-\nu})$ is updated. 

\section{Convergence Analysis} \label{sec4}
In this section, we establish the convergence results of Algorithm \ref{algorithm1}. We prove that the sequence generated by Algorithm 1 converges  to a saddle point of $\mathcal{L}_{\alpha\beta}^{\nu}(x^{\nu},x^{-\nu},z^{\nu},\lambda^{\nu},\mu^{\nu})$ for $\nu=1,\ldots,N$. In particular, our analysis proceeds with the steps:
\begin{enumerate}
	\item We first derive an important result that $\left\|\lambda^{\nu,k+1} -\lambda^{\nu,k} \right\|$ can be bounded by $\left\|\mathbf{x}^{k+1} -\mathbf{x}^{k} \right\|$ (Lemma \ref{lem_bound_multi}). The result, together with Lemma \ref{prop_converge_inner}, is exploited to show that the sequence $\{ \mathcal{L}_{\alpha\beta}^{\nu}\}$ is monotonically decreasing and convergent (Lemma \ref{lem_lag_behavior}). 
	\vspace{0.025in}
	
	\item We then establish the key results; the boundedness of $\left\{\mathbf{x}^{k}\right\}$ and $\mathrm{lim}_{k \rightarrow \infty}\left\|\mathbf{x}^{k+1} -\mathbf{x}^{k} \right\| =0$,  followed by the boundedness of $\left\{\lambda^{\nu,k}\right\}$ (Theorem \ref{thm2_x}).
	\vspace{0.025in}
	
	\item With the bounded sequences, convergence to an equilibrium of the GNEP is proven; we show that any limit point of the sequence is a saddle point of $\mathcal{L}_{\alpha\beta}^{\nu}$ (Theorem \ref{thm_limit_saddle}).
	\vspace{0.025in}
	
	\item Finally, we establish the global convergence that the \emph{whole} sequence generated by the algorithm converges to a saddle point of $\mathcal{L}_{\alpha\beta}^{\nu}$ by assuming that the P-Langangian satisfies the \emph{Kurdyka-{\L}ojasiewicz} (K\L) property (Theorem \ref{thm_global_convergence}).
\end{enumerate}

\subsection{Key Properties of Algorithm 1} \label{sec4.1}

We show that the sequence $\{ \mathcal{L}_{\alpha\beta}^{\nu}\}$ can be a nonincreasing sequence. To this end, we first derive an important relation on the dual iterates  $\lambda^{\nu,k}$ and $\mu^{\nu,k}$ with the primal iterates $ \mathbf{x}^{k}$ that the difference of two consecutive iterates of the multipliers can be bounded by that of the primal iterates.

\renewcommand\thetheorem{3}
\begin{lemma}\label{lem_bound_multi} 
Let $\left\{(x^{\nu,k},z^{\nu,k},\lambda^{\nu,k},\mu^{\nu,k}) \right\} _{\nu=1}^{N}$ be the sequence generated by Algorithm 1. Then, 
\begin{equation} \label{eq:multiplier_x_relation}
\left\Vert \lambda^{\nu,k+1}-\lambda^{\nu,k}\right\Vert^2 \leq \frac{L_{g^{\nu}}^{2}}{\beta_{\nu}^{2}} \left\Vert \mathbf{x}^{k+1} - \mathbf{x}^{k} \right\Vert^2, 
\end{equation}
where $L_{g^{\nu}}$ is the Lipschitz constant of $g^{\nu}$ and $\beta_{\nu}>0$ is the parameter of  $-\frac{\beta_{\nu}}{2}\left\Vert \lambda^{\nu}-\mu^{\nu}\right\Vert ^{2}$ in $\mathcal{L}_{\alpha\beta}^{\nu}$.
\end{lemma}
 \begin{proof}
 	See Appendix \ref{C.1}
 \end{proof}

Equipped with  Lemmas \ref{prop_converge_inner} and  \ref{lem_bound_multi}, we prove that the sequence of function values $\{ \mathcal{L}_{\alpha \beta}^{\nu}\} $ can be monotonically decreasing and convergent.

\renewcommand\thetheorem{4}
\begin{lemma} [Sufficient Decrease and Convergence of $\{\mathcal{L}_{\alpha \beta}^{\nu}\}$]\label{lem_lag_behavior}
Suppose that Assumptions \ref{assumption1} and \ref{assumption_lipschitz} hold. Let $\left\{\left({x}^{\nu,k},z^{\nu,k},\lambda^{\nu,k},\mu^{\nu,k}\right) \right\}_{\nu=1}^{N}$ be the sequence generated by Algorithm 1. Then for $\nu=1,\ldots,N$, we have 
\[\label{eq:lem_lag_behavior}
\mathcal{L}_{\alpha \beta}^{\nu}\left(\mathbf{x}^{k+1},z^{\nu,k+1},\lambda^{\nu,k+1},\mu^{\nu,k+1}\right) 
\leq
\mathcal{L}_{\alpha\beta}^{\nu}\left(\mathbf{x}^{k},z^{\nu,k},\lambda^{\nu,k},\mu^{\nu,k}\right) 
-\frac{1}{2}\left(\gamma_{\nu} -L_{\nu} -\frac{3L_{g^{\nu}}^{2}}{\beta_{\nu}}\right) \left\Vert \mathbf{x}^{k+1}-\mathbf{x}^{k}\right\Vert ^{2},
\]
where $L_\nu>0$ is the Lipschitz gradient constant of $\mathcal{L}_{\alpha \beta}^{\nu}$, $\gamma_\nu > 0$ is the parameter of proximal term  $\frac{\gamma_{\nu}}{2}\left\Vert \mathbf{x}-\mathbf{x}^{k}\right\Vert^{2}$ in  $\widehat{\mathcal{L}_{\alpha \beta}^{\nu}}$, and $\beta_\nu > 0$ is the parameter of quadratic term $-\frac{\beta_{\nu}}{2}\left\Vert \lambda-\mu\right\Vert ^{2}$ in $\mathcal{L}_{\alpha \beta}^{\nu}$. In particular, if  $\gamma_{\nu}$ is chosen large enough such that  $\gamma_{\nu} \geq\ {L_{\nu}} + \frac{3L_{g^{\nu}}^{2}}{\beta_{\nu}}$, then the sequence $\{ \mathcal{L}_{\alpha\beta}^{\nu} \} $ is nonincreasing and convergent. 
\end{lemma}
\begin{proof}
	See Appendix \ref{C.2}
\end{proof}

Next, we provide our key results that the generated sequence is bounded and asymptotic regular. The boundedness of the sequence follows by combining the above decrease property of the P-Lagrangian with the coercivity assumption on the only objective functions (Assumption \ref{assumption_coercive}).

\renewcommand\thetheorem{3}
\begin{theorem}\label{thm2_x} 
Assume that there exists a GNE of the GNEP \eqref{eq:of} satisfying the KKT conditions \eqref{eq:KKT} for every $\nu=1,\ldots,N$. Let $\left\{\left(x^{\nu,k},z^{\nu,k},\lambda^{\nu,k},\mu^{\nu,k}\right) \right\}_{\nu=1}^{N}$ be the sequence generated by Algorithm \ref{algorithm1} with the parameters set to $\gamma_{\nu}>0$ large enough so that  $\gamma_{\nu} \geq L_{\nu} +\frac{3L_{g^{\nu}}^2}{\beta_{\nu}}$. Then,		
	\begin{enumerate}[label=(\alph*)]		
		\item \label{itm:thm2_a} the primal sequence $\left\lbrace  \mathbf{x}^k \right\rbrace $ is bounded; 
		
		\item \label{itm:thm2_b} the sequence of the multiplier $\left\{ \lambda^{\nu,k}\right\} $ is  bounded;
		
		\item \label{itm:thm2_c} it holds that  
		$\sum_{k=1}^{\infty}\left\Vert \mathbf{x}^{k+1}-\mathbf{x}^k\right\Vert ^2< \infty \:$  and 
		$\: \sum_{k=1}^{\infty}\left\Vert \lambda^{\nu,k+1}-\lambda^{\nu,k}\right\Vert ^2< \infty$, and hence
	\begin{equation} \label{eq:thm2_x_c_result}
		\underset{k\rightarrow\infty}{\mathrm{lim}}\left\Vert\mathbf{x}^{k+1}-\mathbf{x}^k \right\Vert =0, \quad  
		\underset{k\rightarrow\infty}{\mathrm{lim}}\left\Vert \lambda^{\nu,k+1}-\lambda^{\nu,k} \right\Vert =0, 
		\quad \mathrm{and} \quad  \underset{k\rightarrow\infty}{\mathrm{lim}}\left\Vert \mu^{\nu,k+1}-\mu^{\nu,k} \right\Vert = 0.
	\end{equation}
	\end{enumerate}
\end{theorem}
\begin{proof}
	See Appendix \ref{C.3}
\end{proof}

\subsection{Main Convergence Results} \label{sec4.2}
We are ready to establish our main convergence results. We first show that any limit point of the sequence produced by Algorithm \ref{algorithm1} is a saddle point of $\mathcal{L}_{\alpha\beta}^{\nu}$ for every $\nu=1,\ldots,N$.

\renewcommand\thetheorem{4}
\begin{theorem}[Subsequence Convergence]\label{thm_limit_saddle}
Let $\left\lbrace ( {x}^{\nu,k},z^{\nu,k},\lambda^{\nu,k},\mu^{\nu,k}) \right\rbrace_{\nu=1}^{N}$ be the sequence generated by Algorithm \ref{algorithm1}. Then, the sequence $\left\lbrace \left({x}^{\nu,k},z^{\nu,k},\lambda^{\nu,k},\mu^{\nu,k}\right) \right\rbrace_{\nu=1}^{N}$ converges to a point $(\overline{\mathbf{x}},\overline{z}^{\nu},\overline{\lambda}^{\nu},\overline{\mu}^{\nu})$ that satisfies the saddle point condition \eqref{eq:saddle}.
\end{theorem}
\begin{proof}
	See Appendix \ref{D.1}
\end{proof}

We now strengthen the above subsequence convergence result under an additional assumption that the functions satisfy the \emph{Kurdyka-{\L}ojasiewicz (K\L{}) property} (see \citet{lojasiewicz1963propriete} and \citet{kurdyka1998gradients}): The K{\L} property, along with the sufficient decrease of the P-Lagrangian and the boundedness of generated sequence, enables us to establish global convergence of the \emph{whole} sequence $\left\lbrace({x}^{\nu,k},z^{\nu,k},\lambda^{\nu,k},\mu^{\nu,k})\right\rbrace_{\nu=1}^{N}$ to a saddle-point of $\mathcal{L}_{\alpha\beta}^{\nu}$  by showing that the sequence has finite length.

\begin{definition}[K{\L} Property \& K{\L} function] \label{def_KL_pro_func}
Let $\delta \in \left(0, +\infty\right]$. Denote by $\Phi_\delta$ the class of all concave and continuous functions $\varphi:\left[0, \delta\right) \rightarrow \mathbb{R}_{+}$, which satisfy the following conditions:
	\begin{enumerate}[label=(\roman*)]
		\item $\varphi(0)=0$;
		\item $\varphi$ is continuously differentiable ($C^{1}$) on $\left[0, \delta\right)$ and continuous at 0;
		\item for all $s \in (0, \delta): \varphi^{\prime}>0.$
	\end{enumerate}	
A proper and lower semicontinuous function $\Psi: \mathbb{R}^n \rightarrow \left(-\infty, +\infty \right]$ is said to have the \emph{Kurdyka-{\L}ojasiewicz (K{\L}) property} at $\overline{u} \in \mathrm{dom} \: \partial \Psi :=\left\lbrace u \in \mathbb{R}^n: \partial \Psi(u)=\emptyset \right\rbrace $ if there exist $\delta \in \left(0, +\infty\right]$, a neighborhood $U$ of $\overline{u}$ and a function $\varphi \in \Phi_\delta$, such that 
	\[
	  \varphi^{\prime}(\Psi(u) -\Psi(\overline{u})) \cdot \mathrm{dist}(0,\partial \Psi(u)) \geq 1
	\]
for all $u \in U(\overline{u}) \cap \{u: \Psi(\overline{u}) < \Psi({u}) <\Psi(\overline{u}) +\delta \}$. The function $\Psi$ satisfying the K\L{} property at each point of dom$\: \partial \Psi$ is called a \emph{K\L{} function}. 
\end{definition}

\renewcommand\thetheorem{5}
\begin{theorem}[Global Convergence] \label{thm_global_convergence}
	Suppose that $\theta_\nu$ and $g^\nu_i$, $\nu=1,\ldots,N$, $i=1,\ldots,m_\nu$, satisfy the K\L{} property. Let  $\left\lbrace \mathbf{w}^{\nu,k}:=(\mathbf{x}^{k},z^{\nu,k},\lambda^{\nu,k},\mu^{\nu,k})\right\rbrace_{\nu=1}^N$ be the sequence generated by Algorithm \ref{algorithm1}. Then the sequence  $\left\lbrace \mathbf{w}^{\nu,k}= (\mathbf{x}^{k},z^{\nu,k},\lambda^{\nu,k},\mu^{\nu,k})\right\rbrace_{\nu=1}^N $ has finite length, i.e.,  
	$$\sum_{k=1}^{\infty}\left\Vert \mathbf{w}^{\nu,k+1}-\mathbf{w}^{\nu,k}\right\Vert<+ \infty,$$
	and the whole sequence $\left\lbrace ( {x}^{\nu,k},z^{\nu,k},\lambda^{\nu,k},\mu^{\nu,k}) \right\rbrace_{\nu=1}^{N}$ converges to a saddle point  $(\overline{\mathbf{x}},\overline{z}^{\nu},\overline{\lambda}^{\nu},\overline{\mu}^{\nu})$ of $\mathcal{L}_{\alpha \beta}^{\nu}$.
\end{theorem}
\begin{proof}
	See Appendix \ref{D.2}
\end{proof}
We note that verifying the K{\L} property of a function might be a difficult task. However, it is well-known that \emph{semi-algebraic} and \emph{real-analytic} functions, which capture many applications, are classes of functions that satisfy the K{\L} property; see e.g., \cite{attouch2009convergence,attouch2010proximal,attouch2013convergence,xu2013block, li2018calculus} for an in-depth study of the K\L{} functions and illustrating examples.

\section{Computational Results}

In this section, we present computational results to demonstrate the effectiveness of the proposed method. We conducted numerical experiments on test problems using Algorithm 1. The experiments were carried out using MATLAB (R2018a) on a laptop with a Intel Core i5-6300U CPU 2.50GHz 8GB RAM. All the test problems were taken from a library of GNEPs used in the literature \cite{facchinei2010penalty,dreves2011solution,kanzow2016augmented}, and two classes of instances were considered; \emph{general} GNEPs (A.1-A.10) and \emph{jointly-convex} GNEPs (A.11-A.18). We refer the readers to \cite{facchinei2009penalty} for a detailed description of the problems with data. 

In the numerical test, we used the starting points listed in \cite{facchinei2010penalty}, and the other variables' initial points were set to $\left(z^{\nu,0},\lambda^{\nu,0},\mu^{\nu,0} \right)=\left(0,0,0 \right)$ for every $\nu=1,\ldots,N$. As for the parameters, we used fixed parameters set to $\alpha_{\nu}=10$ and $\beta_{\nu}=1$ for each player's P-Lagrangian and for all test problems. In addition, diminishing step size $\sigma_\nu$ was simply used for every player $\nu$ in each problem. The stopping criterion is set as 
\[
  \underset{\nu=1,\ldots,N}{\mathrm{max}} \left\lbrace  \left\|  x^{\nu,k+1} - x^{\nu,k} \right\|_{\infty}, \left\|  \lambda^{\nu,k+1} - \lambda^{\nu,k} \right\|_{\infty} \right\rbrace \leq 10^{-4}. 
\]
The computational results of our algorithm for the test problems are presented in Table 1, where we used the following notations;
the number of players ‘$N$’, the number of variables ‘$n$’, the number of constraints ‘$m$’, starting point ‘$\mathbf{x}^0$’,
total (cumulative) number of inner iterations ‘$Iter.$’, and computation time of CPU seconds ‘Time(s)’.

We summarize the computational results in the following:
\begin{enumerate}
	\item Algorithm 1 was able to solve all the test problems. The experimental results show that Algorithm 1 comes out favorably in terms of the number of problems solved compared to the other methods. In particular, the exact penalty algorithm in \cite{facchinei2010penalty} was unable to find solutions to the problems A.2, A.7, and A.8. In addition, the interior-point algorithm in \cite{dreves2011solution} and the augmented Lagrangian method in \cite{kanzow2016augmented} were unable to find a GNE of the instance A.8 with starting points $\mathbf{x}^{0}=10$ and $\mathbf{x}^{0}=0$, respectively. On the other hand, Algorithm 1 converges to a GNE for the three problems A.2, A.7, and A.8, starting from those points.
	\vspace{0.025in}
	
	\item It is worth noting that Algorithm 1 converges to the same GNE from different starting points in each problem, while the exact penalty algorithm \cite{facchinei2010penalty} converges to different equilibria or generates unbounded sequences in some cases. This difference is because each player's problem is convex in its own variables, and each player solves strongly convex subproblems while keeping the other players' variables fixed. On the other hand, the exact penalty algorithm \cite{facchinei2010penalty} solves nondifferentiable (possible nonconvex) subproblems. This, along with sensitivity to the penalty parameters and starting points, may lead to the convergence to different equilibria or the failure of convergence.
	\vspace{0.025in}
	
	\item Our distributed algorithm required very short CPU times to reach equilibrium for all test instances. The results confirm the efficiency of our algorithm that the computation time per iteration $k$ to solve each subproblem is significantly short for every instance. This advantage is mainly due to the Jacobi decomposition scheme for the $\mathbf{x}$-update with the cheap projection onto the simple set $\mathcal{X}_{\nu}$.
	\vspace{0.025in}
	
	\item For Arrow-Debreu equilibrium problems A.10 (a-e), it is noteworthy that our problem setting is different from test problems in \cite{facchinei2009penalty}. In our setting, the production variables are added to the constraints of consumers' problems, that is $p^{T}x^{i}\leq p^{T} \xi^{i}+\sum_{j=1}^{J}q_{ij}p^{T}y^{j}$, whereas  the constraints were set to $p^{T}x^{i}\leq p^{T}\xi^{i}$ in the test setting in \cite{facchinei2009penalty}. This reflects  the original Arrow-Debreu model better and computational results have shown that Algorithm 1 performs well on the modified instances.
	
\end{enumerate}

\begin{table} 
	\centering{}	
	\caption{Computational results for Algorithm 1.} \label{table 1}
		\begin{tabular}{|c|cccc|c|c|}
			\hline 
			\multirow{1}{*}{\emph{general} GNEP} & 
			$N$ & $n$ & $m$ & $\mathbf{x}^{0}$ & $Iter.$ & {Time (s)}  \tabularnewline
			\hline 
			A.1 & 10 & 10 & 20 & 0.01 & 38 & $<0.01$  \tabularnewline
			& & & &               0.1 & 36 & $<0.01$   \tabularnewline
			& & & &                 1 & 38& $<0.01$   \tabularnewline
			\hline 
			A.2 & 10 & 10 & 24 & 0.01 & 610 & 0.04  \tabularnewline
			& & & &               0.1 & 536 & 0.04    \tabularnewline
			& & & &                 1 & 683 & 0.05    \tabularnewline
			\hline 
			A.3 & 3  & 7  & 18 &    0 & 51  & 0.01  \tabularnewline
			& & & &                 1 & 51  & 0.01    \tabularnewline
			& & & &                10 & 51  & 0.01    \tabularnewline
			\hline 
			A.4 & 3  & 7  & 18 &    0 &  7  &$<0.01$  \tabularnewline
			& & & &                 1 &  7  &$<0.01$    \tabularnewline
			& & & &                10 &  7  &$<0.01$    \tabularnewline
			\hline 
			A.5 & 3  & 7  & 18 &    0 & 82  & 0.02    \tabularnewline
			& & & &                 1 & 82  & 0.02      \tabularnewline
			& & & &                10 & 82  & 0.02     \tabularnewline
			\hline 
			A.6 & 3  & 7 &  21 &    0 &  49  & 0.02   \tabularnewline
			& & & &                 1 &  49  & 0.02     \tabularnewline
			& & & &                10 &  49  & 0.02     \tabularnewline
			\hline 
			A.7 & 4 & 20 & 44 &     0 & 48   & 0.02   \tabularnewline
			& & & &                 1 & 48   & 0.02     \tabularnewline
			& & & &                10 & 48   & 0.02     \tabularnewline
			\hline 
			A.8 & 3 & 3  & 8  &     0 & 45  &$<0.01$ \tabularnewline
			& & & &                 1 & 45  &$<0.01$   \tabularnewline
			& & & &                10 & 45  &$<0.01$   \tabularnewline
			\hline 
			A.9 (a) & 7 & 56 & 63 & 0 & 108 & 0.32  \tabularnewline
			A.9 (b) & 7 & 112& 119& 0 & 135 & 1.24    \tabularnewline
			\hline 
			A.10 (a)& 8 & 24 & 33 & 0 & 780 & 0.10  \tabularnewline
			A.10 (b)& 25& 125& 151& 1 & 1374& 0.67    \tabularnewline
			A.10 (c)& 37& 222& 260& 0 & 2154& 1.12    \tabularnewline
			A.10 (d)& 37& 370& 408& 1 & 3251& 1.35    \tabularnewline
			A.10 (e)& 48& 576& 625& 1 & 4728& 2.54    \tabularnewline
			\hline 
			\hline 
		
			\multirow{1}{*}{\emph{jointly-convex} GNEP} & 
			$N$ & $n$ & $m$ & $\mathbf{x}^{0}$ & $Iter.$ & {Time (s)}  \tabularnewline
			\hline 
			A.11 & 2 & 2 & 2 & 0     & 12 & $<0.01$  \tabularnewline
			\hline 
			A.12 & 2 & 2 & 4 & (2,0) & 10 & $<0.01$  \tabularnewline
			\hline 
			A.13 & 3 & 3 & 9 &    0 & 15 & $<0.01$  \tabularnewline
			\hline 
			A.14 & 10& 10& 20&  0.01&  38  & $<0.01$  \tabularnewline
			\hline 
			A.15 & 3 &  6& 12&    0 & 145 & $<0.01$  \tabularnewline
			\hline 
			A.16 (P=75) & 5  & 5 &  10 & 10 &  52  & 0.02   \tabularnewline
			A.16 (P=100)& 5  & 5 &  10 & 10 &  52  & 0.02   \tabularnewline
			A.16 (P=150)& 5  & 5 &  10 & 10 &  52  & 0.02   \tabularnewline
			A.16 (P=200)& 5  & 5 &  10 & 10 &  52  & 0.02   \tabularnewline
			\hline 
			A.17 & 2 & 3 & 7 &       0 & 9   & $<0.01$   \tabularnewline
			\hline 
			A.18& 2 & 12  & 28  &    0 & 114  & 0.02 \tabularnewline
			\hline 
	\end{tabular}
\end{table}

\subsection*{Illustrative Examples}
To see how well Algorithm 1 performs on GNEPs, we provide numerical results for three important and practical instances with graphical illustrations.  
		
\subsection*{Problem A.9 (a) (Example 1 revisited,  Power allocation in telecommunications)}
This instance sets $\sigma^{\nu}_{i}=0.3162$ for all $\nu$ and $i$, $K=8$, $L^{\nu}=8$ for all players, and the starting point was set to $\left(0,\ldots,0\right)$. The data of coefficient $h$ is given in \cite{facchinei2009penalty}. As shown in \Cref{fig_ex1_1}, the P-Lagrangian function values $\{ \mathcal{L}_{\alpha \beta}^{\nu}\}$, $\nu=1,\ldots,7$, are monotonically decreasing and convergent, as expected. In addition, \Cref{fig_ex1_2} illustrates that the iterates of ${x}^{\nu}$, $\nu=1,3,5$, converge to a limit point satisfying the minimum target rate of 8.
		
Clearly, the computation time crucially depends on how the subproblems are solved efficiently. It is noteworthy that since the nonlinear coupling constraints are relaxed into the objective with the multipliers, the projection on the set  $\mathcal{X}_{\nu} = \left\{ x^{\nu}\in \mathbb{R}^{n_{\nu}}\left|\right.  x^{\nu} \geq 0 \right\}$ can be performed efficiently, which leads to the convergence to a GNE within a significantly short CPU time of 0.32 seconds.
\begin{figure}  [H]
\includegraphics[scale=0.38]{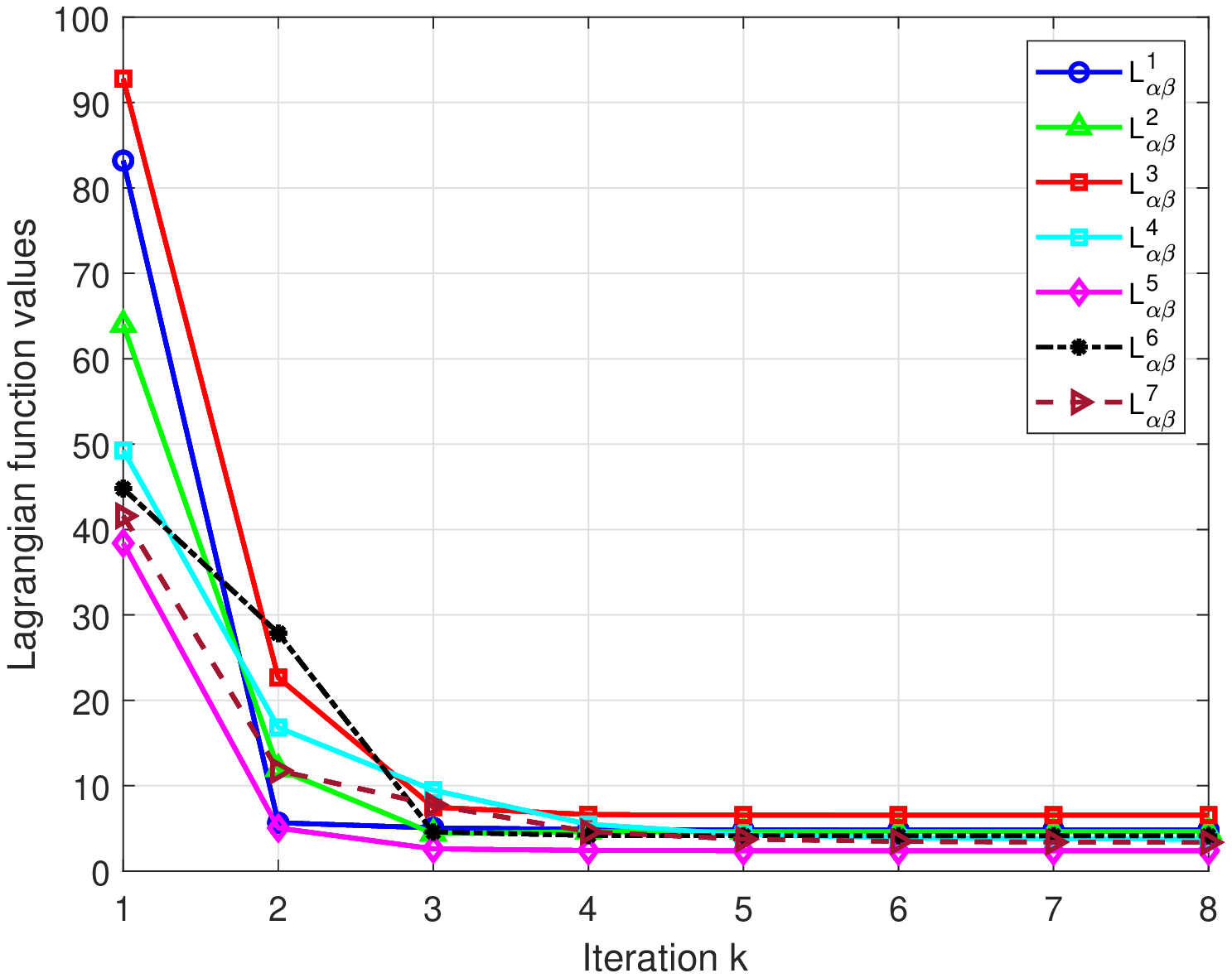}
\caption{Convergence of $\mathcal{L}_{\alpha \beta}^{\nu}$ values, $\nu=1,\ldots,7$.} \label{fig_ex1_1}	
\end{figure}
\begin{figure}[H]	
	\begin{subfigure}[b]{0.35\linewidth}
	\includegraphics[width=\linewidth]{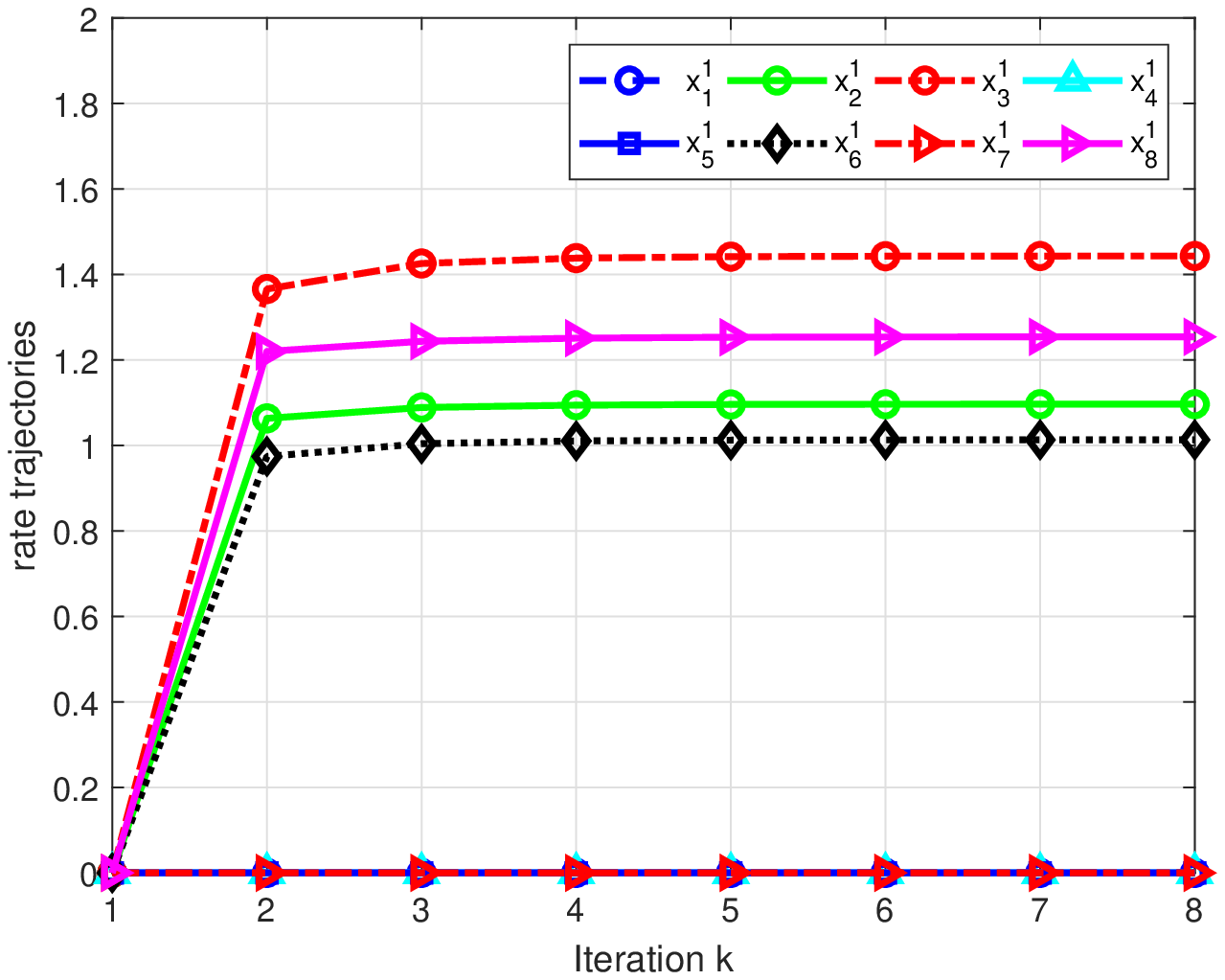}
	\end{subfigure}
	\hspace{-0.19in}\begin{subfigure}[b]{0.35\linewidth}
	\includegraphics[width=\linewidth]{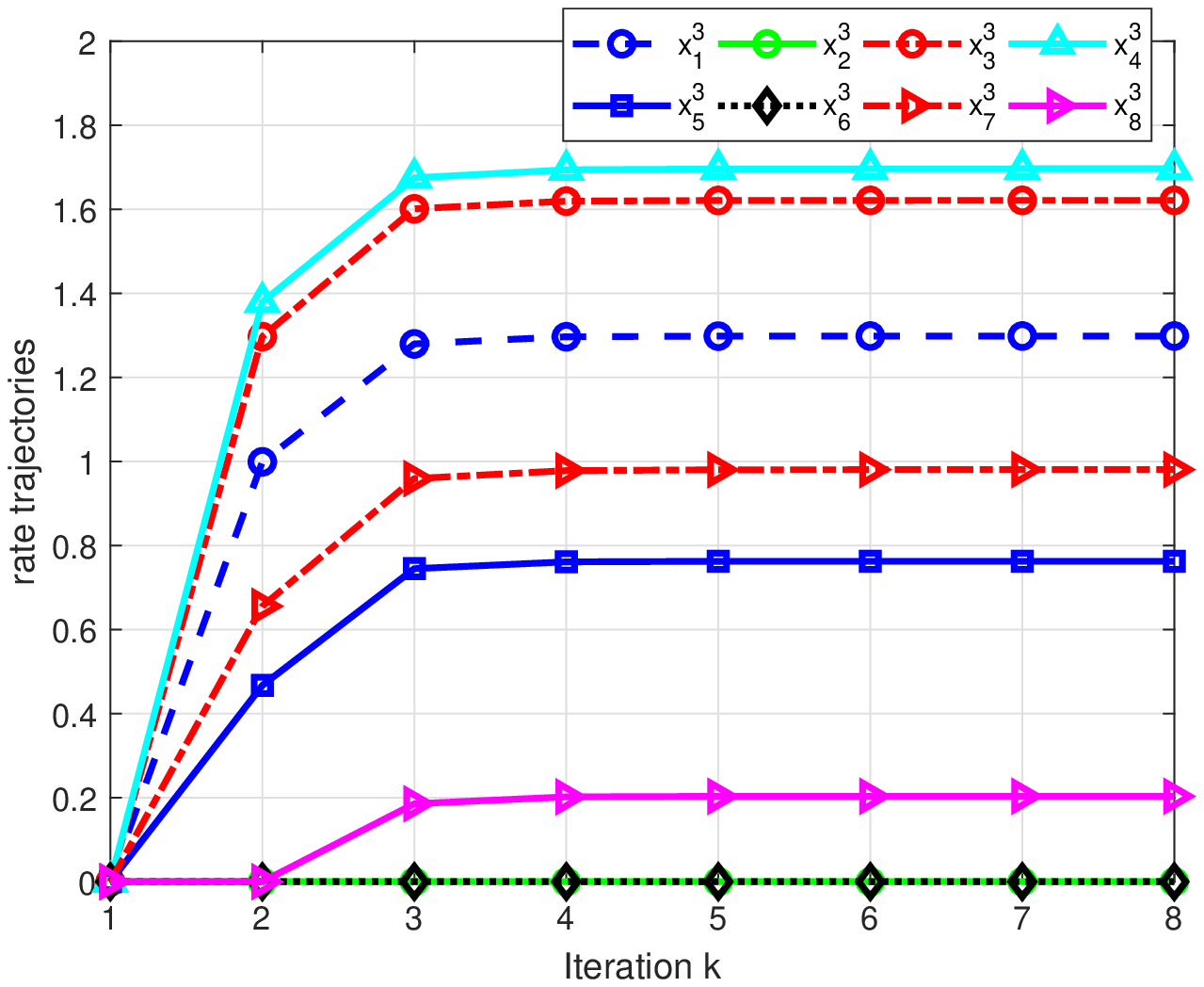}
	\end{subfigure}
	\hspace{-0.19in}\begin{subfigure}[b]{0.35\linewidth}
	\includegraphics[width=\linewidth]{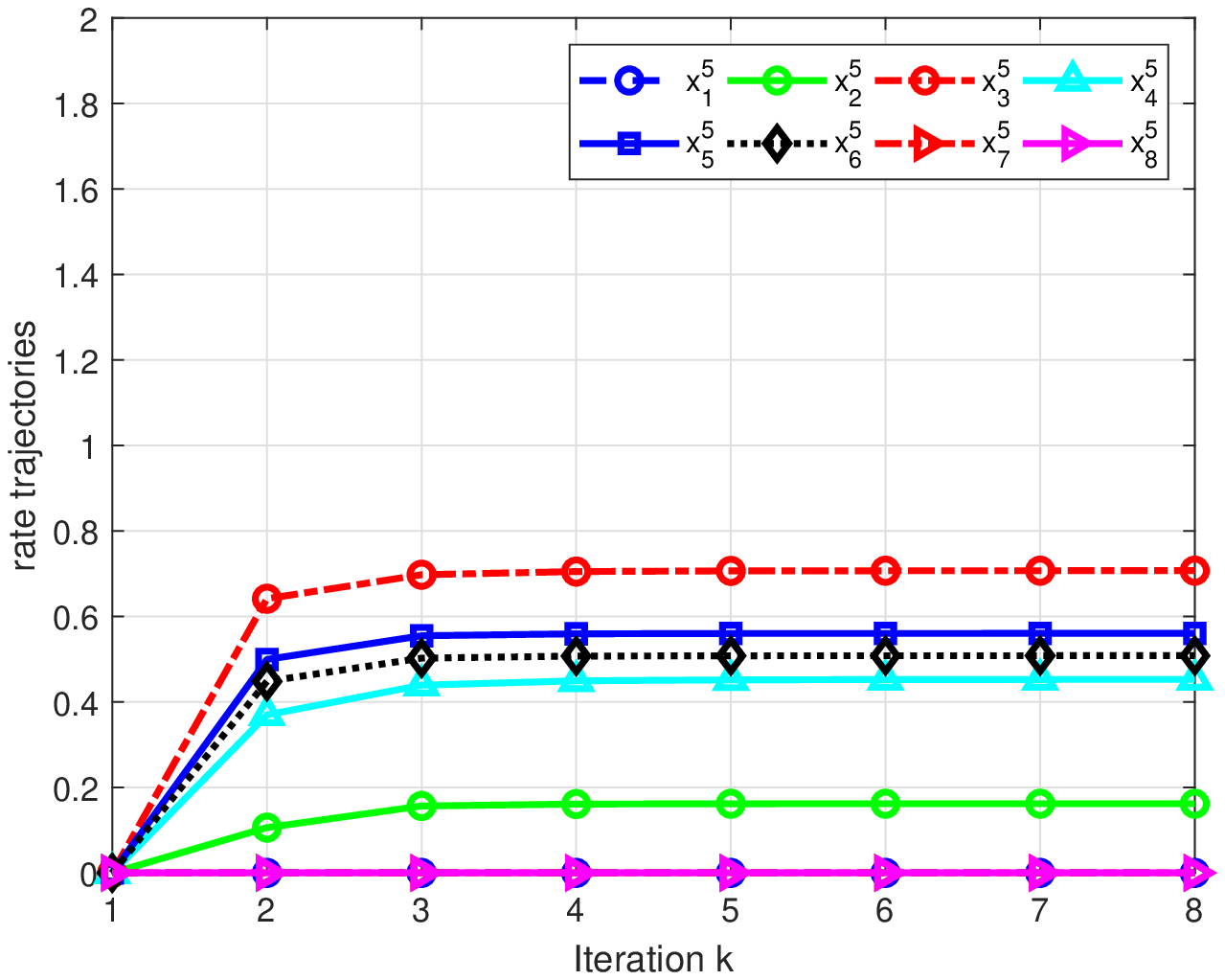}
	\end{subfigure}
	\hspace{-0.19in}\begin{subfigure}[b]{0.35\linewidth}		 
	\includegraphics[width=\linewidth]{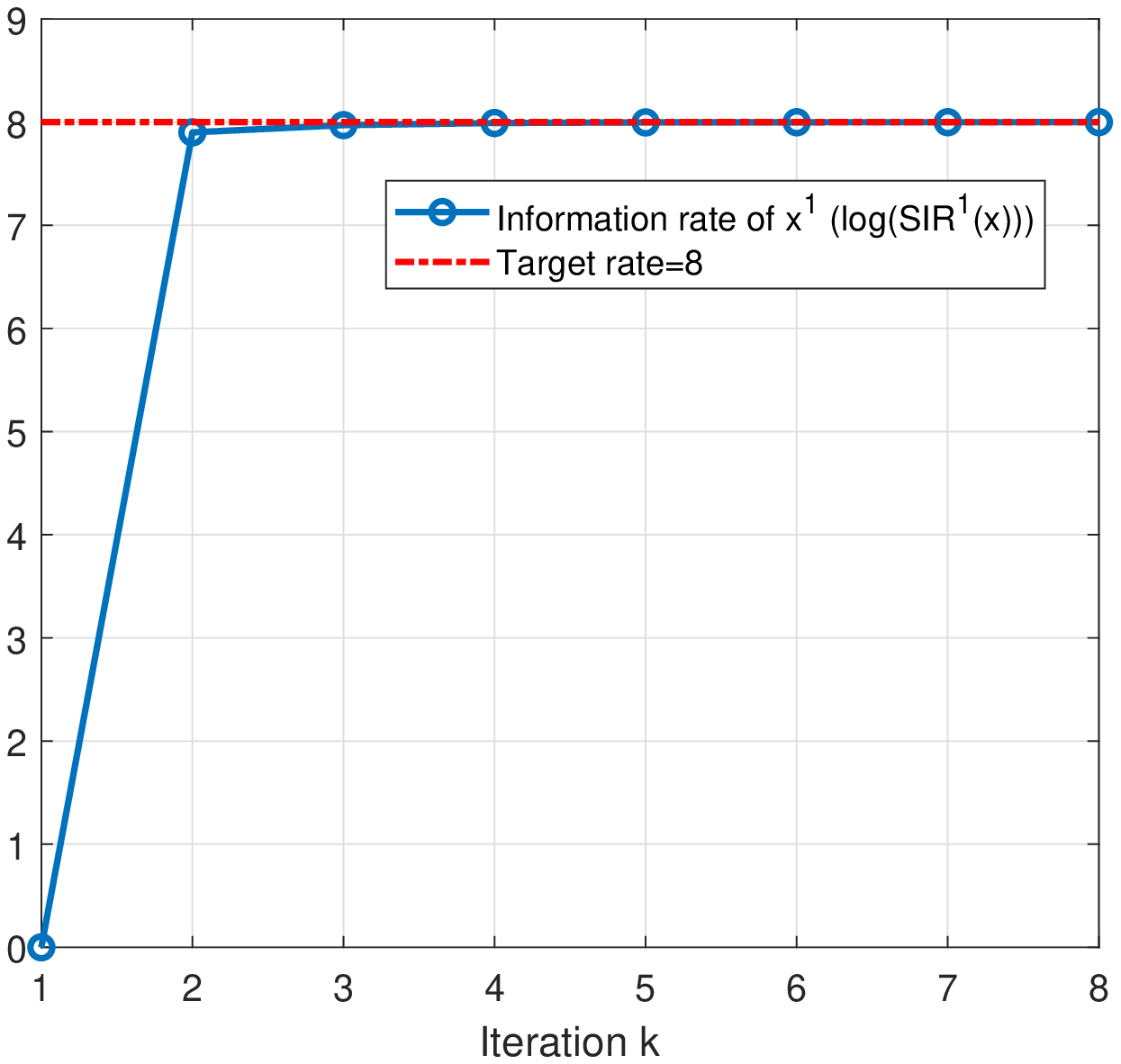}
	\subcaption{{\footnotesize Link 1 ${x}^{1}_{i} $}}
	\end{subfigure}
	\hspace{-0.19in}\begin{subfigure}[b]{0.35\linewidth}
	\includegraphics[width=\linewidth]{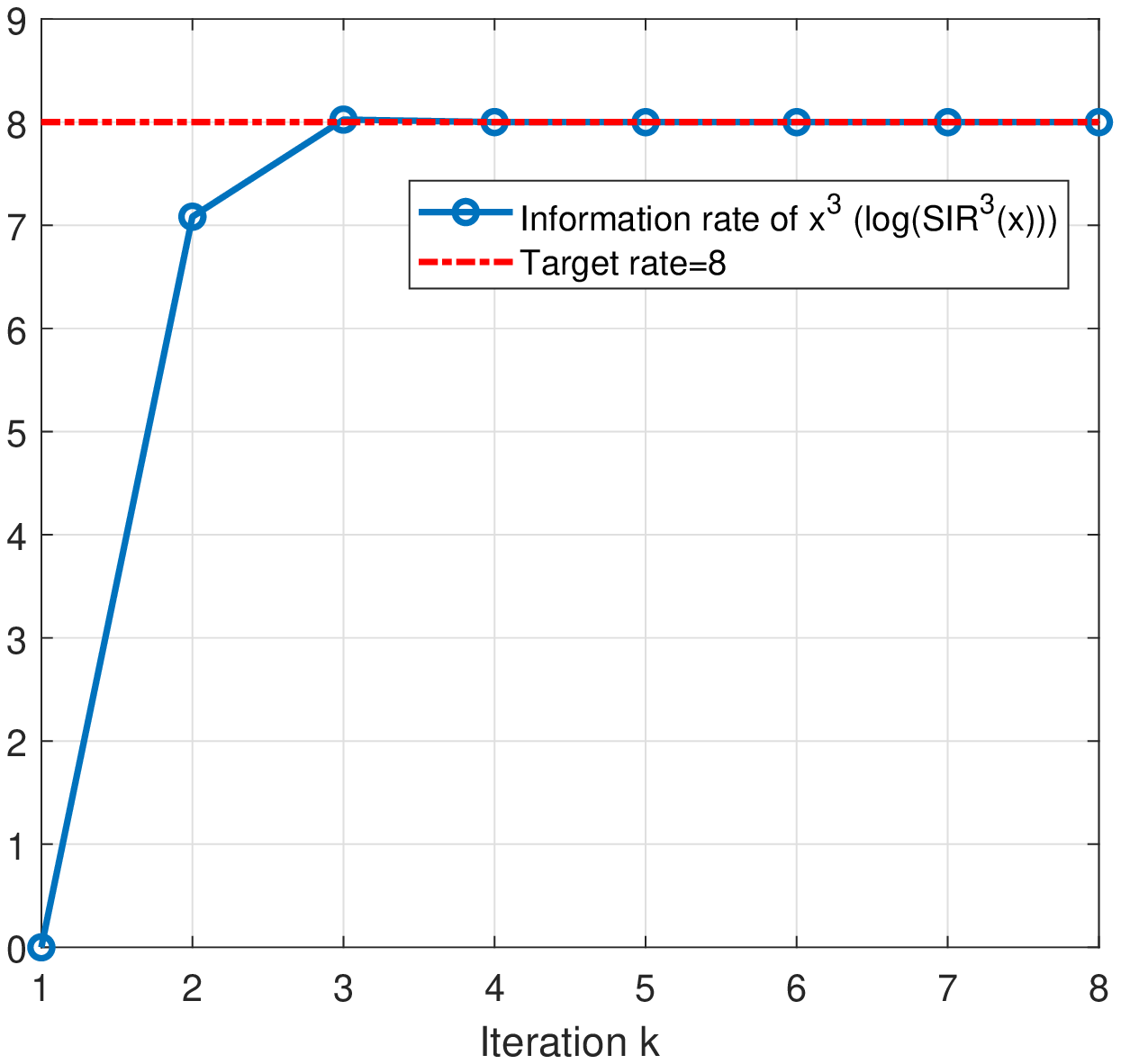}
	\subcaption{{\footnotesize Link 3 ${x}^{3}_{i} $}}
	\end{subfigure}
	\hspace{-0.19in}\begin{subfigure}[b]{0.35\linewidth}			\includegraphics[width=\linewidth]{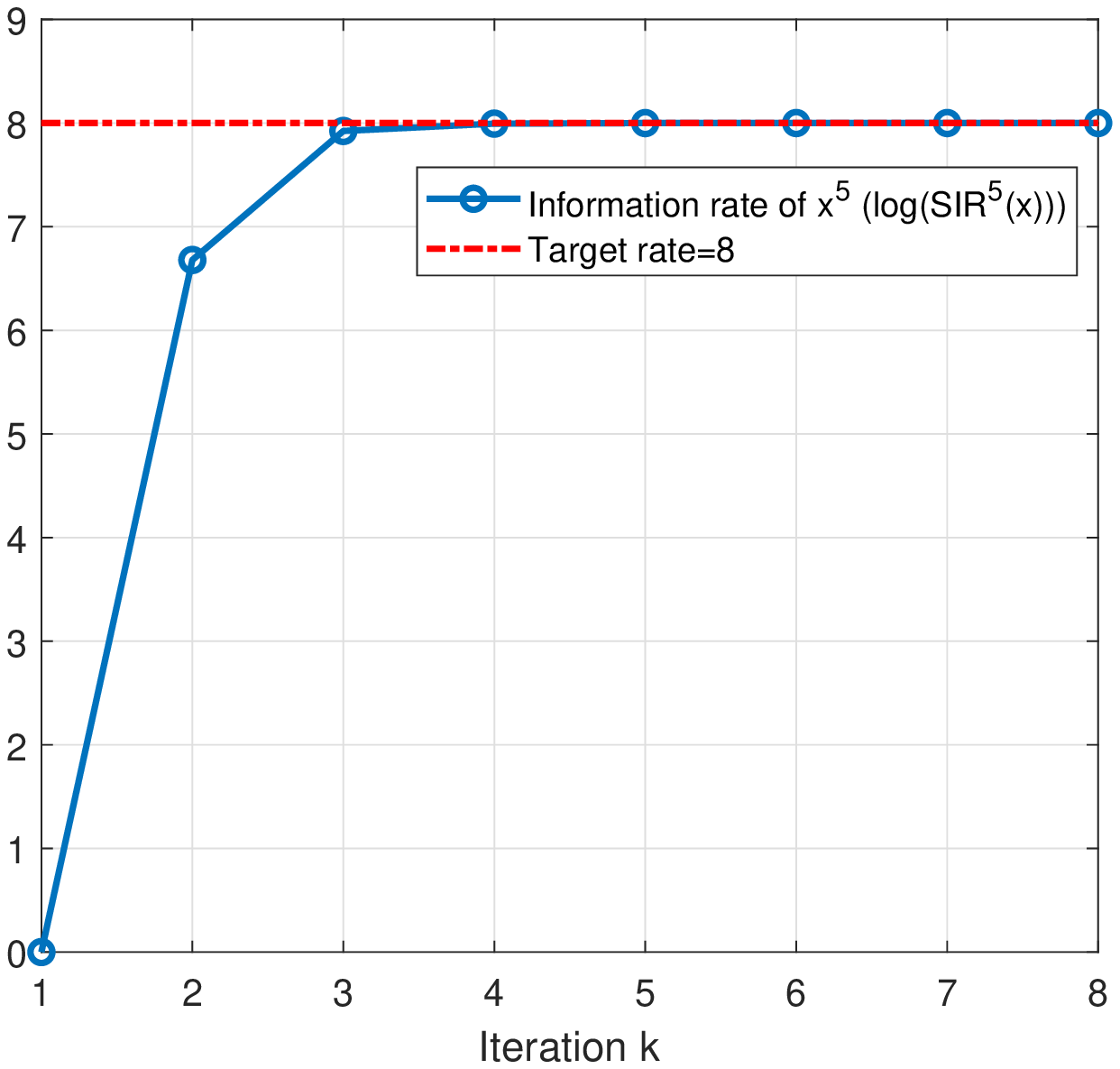}
	\subcaption{{\footnotesize Link 5 ${x}^{5}_{i} $}}
	\end{subfigure}
\caption{Trajectories of the iterates $x^1_i, x^3_i$, and $x^5_i$, $i=1,\ldots,8,$ with information sum-rates.} \label{fig_ex1_2}
\end{figure} 
\vspace{-0.1in}

\subsection*{Problem A.10 (a) (Example 2 revisited, Arrow-Debreu general equilibrium model)}

In this example, there are 8 players ($I=5$, $J=2$ and one market player) and 3 goods ($K=3$). The utility functions $u_{i}$ are quadratic and concave, $u_i(x^i)=-\frac{1}{2} (x^i)^T Q^i x^i + (b^i)^T x^i,$ and $j$-th firm's production set is defined by {\small$Y_j=\left\{y^j \left| y^j \geq 0, \ \sum_{k=1}^{K} (y^j_k)^2 \leq 10*j \right.\right\}$}. The detailed data is given in  \cite{facchinei2009penalty}.

The convergence results are shown in Figures \ref{fig_ex2_1} and \ref{fig_ex2_2}. We see that the numerical results on this example also verifies our theoretical findings. Starting point is set to $x^{i,0}=0$, $y^{j,0}=0$, and $p^{0}=\left(1/3,1/3,1/3\right)$. Figure \ref{fig_ex2_1} demonstrates that all players' P-Langrangian values are decreasing and convergent to finite values. It can be also seen in Figure \ref{fig_ex2_2} that the iterates generated by Algorithm \ref{algorithm1} converge to the equilibrium price vector $\overline{p}=(0.1441,0.5270,0.3289)$ that clears market as well as to the equilibrium productions and consumptions.
	\begin{figure}[H]
		\centering
		\hspace{-0.1in}\begin{subfigure}[b]{0.35\linewidth}
			\includegraphics[width=\linewidth]{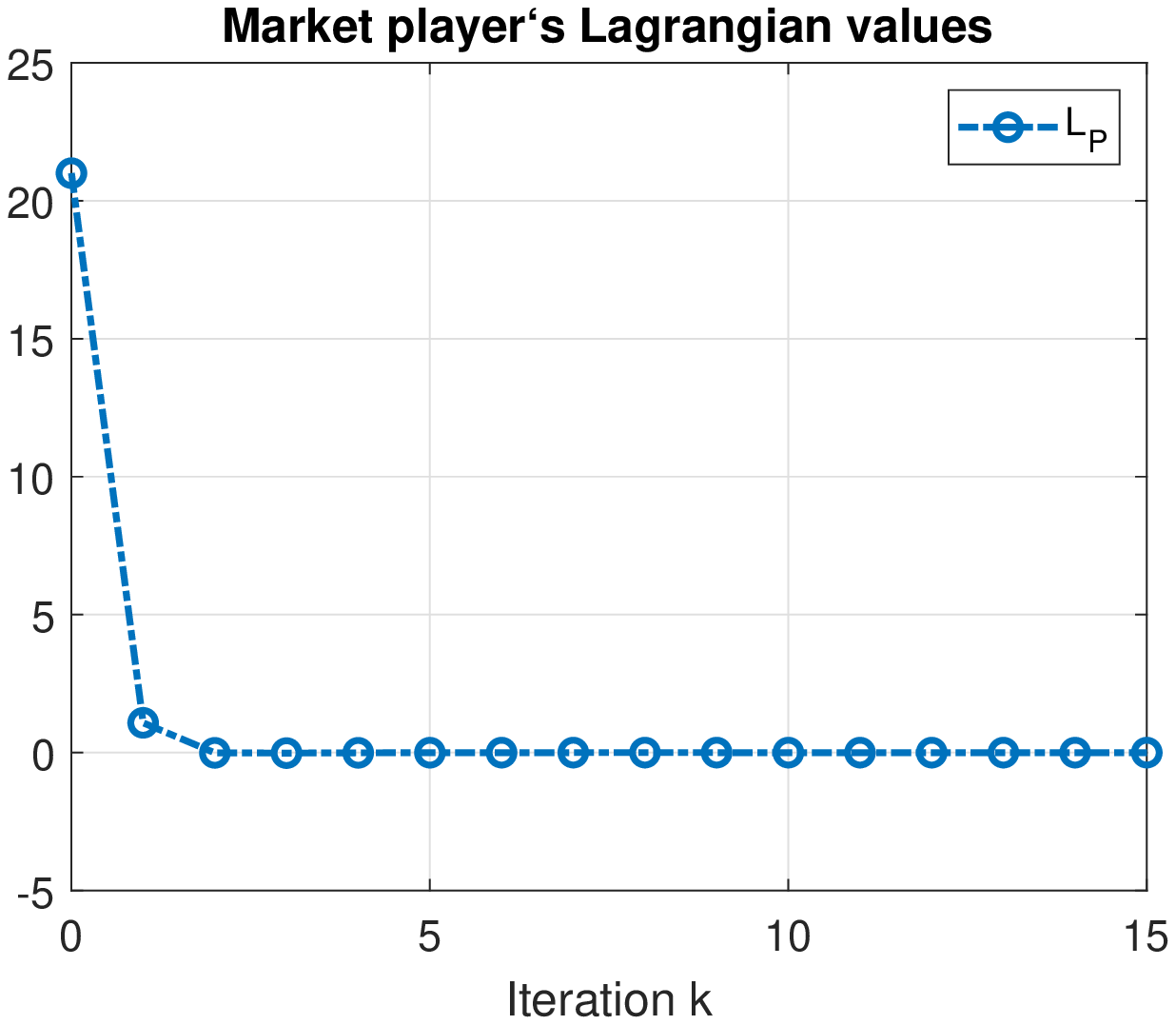}
		\end{subfigure}
		\hspace{-0.2in}\begin{subfigure}[b]{0.35\linewidth}
			\includegraphics[width=\linewidth]{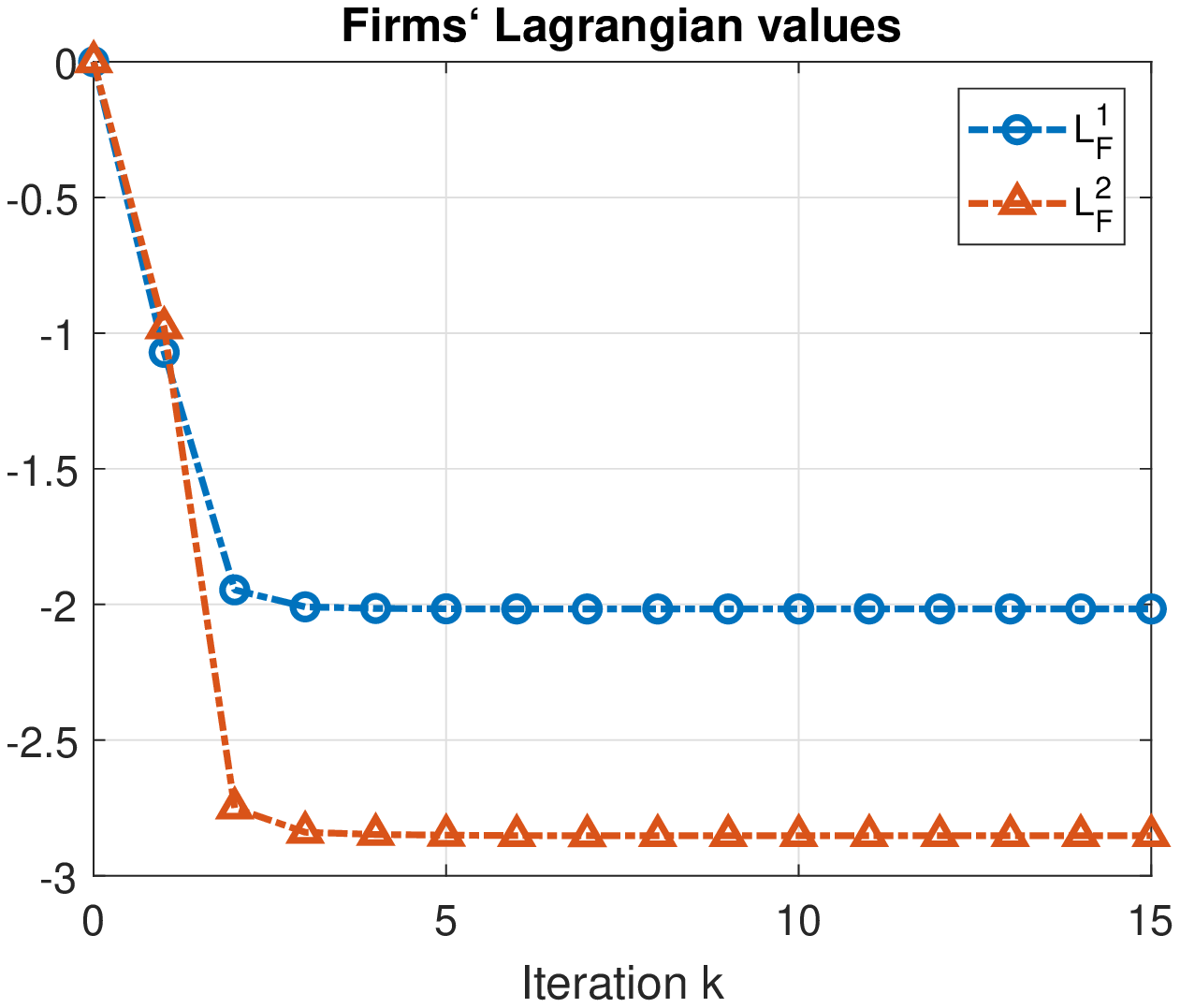}
		\end{subfigure}	
		\hspace{-0.2in}\begin{subfigure}[b]{0.35\linewidth}
			\includegraphics[width=\linewidth]{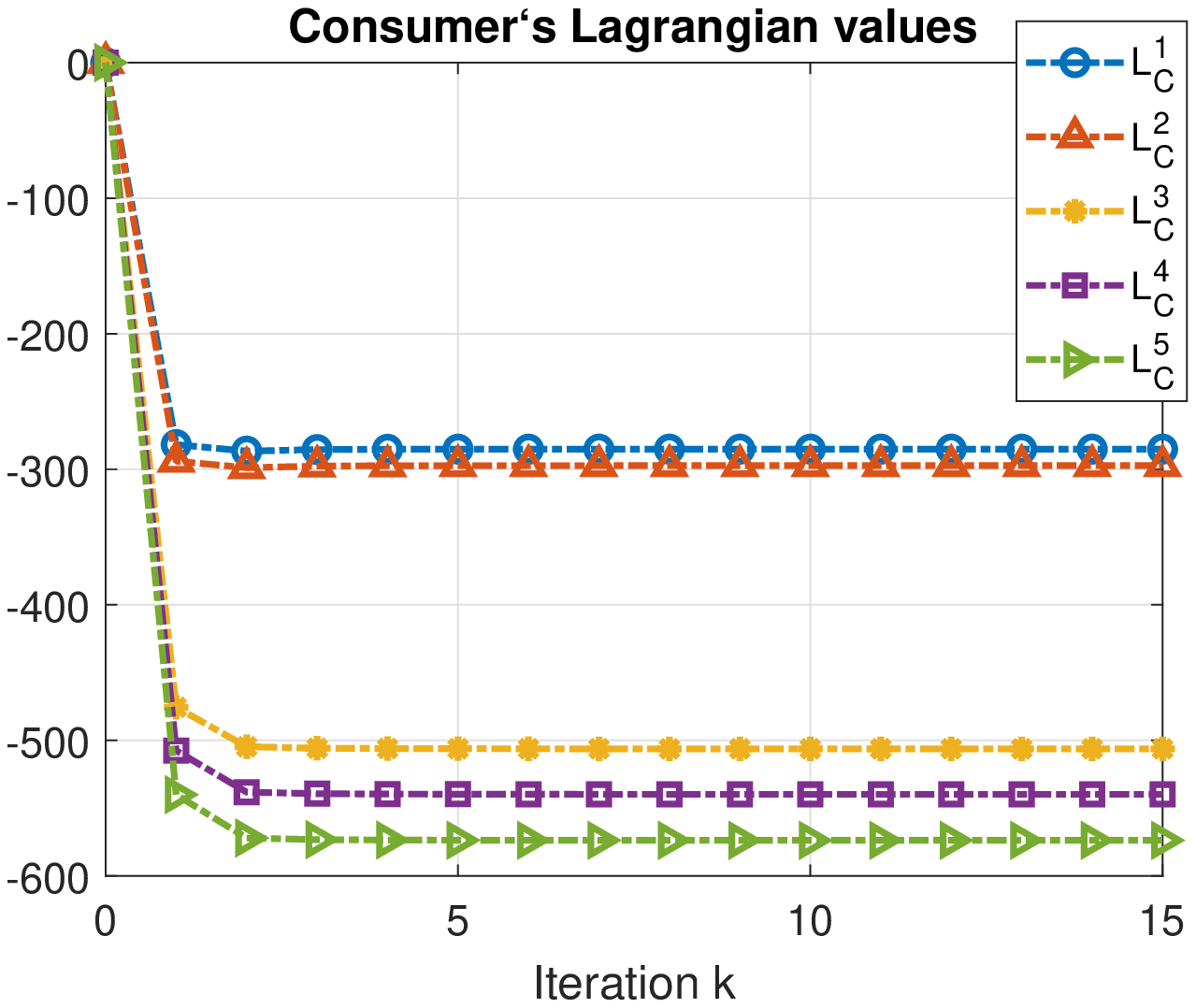}
		\end{subfigure}
		\caption{Convergence of all players' P-Lagrangian values.} \label{fig_ex2_1}
	\end{figure}
\vspace{-0.0in}
	\begin{figure}[H]
		\centering
		\hspace{-0.1in}\begin{subfigure}[b]{0.35\linewidth}
			\includegraphics[width=\linewidth]{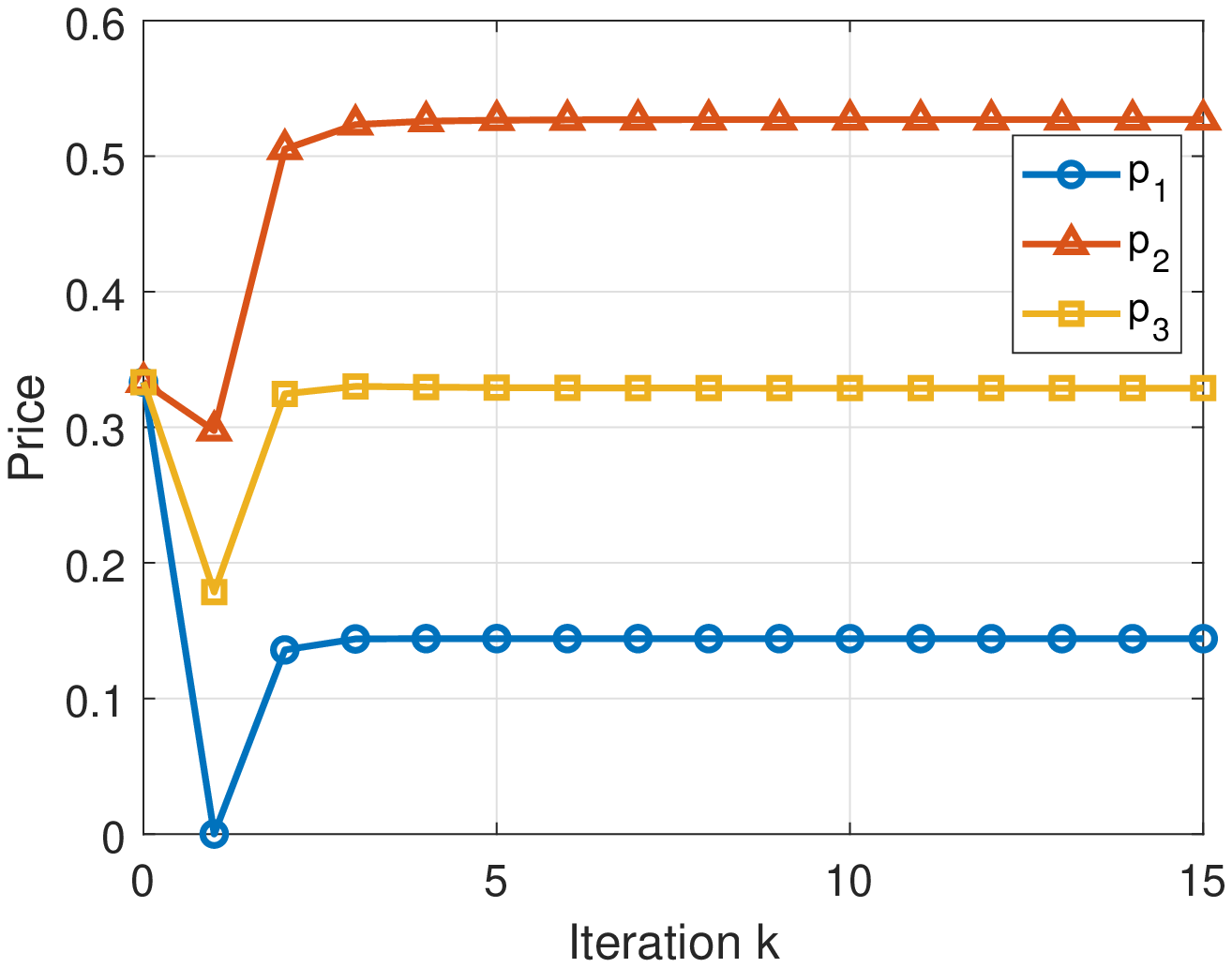}
			\caption{Market player}
		\end{subfigure}
		\hspace{-0.2in}\begin{subfigure}[b]{0.35\linewidth}
			\includegraphics[width=\linewidth]{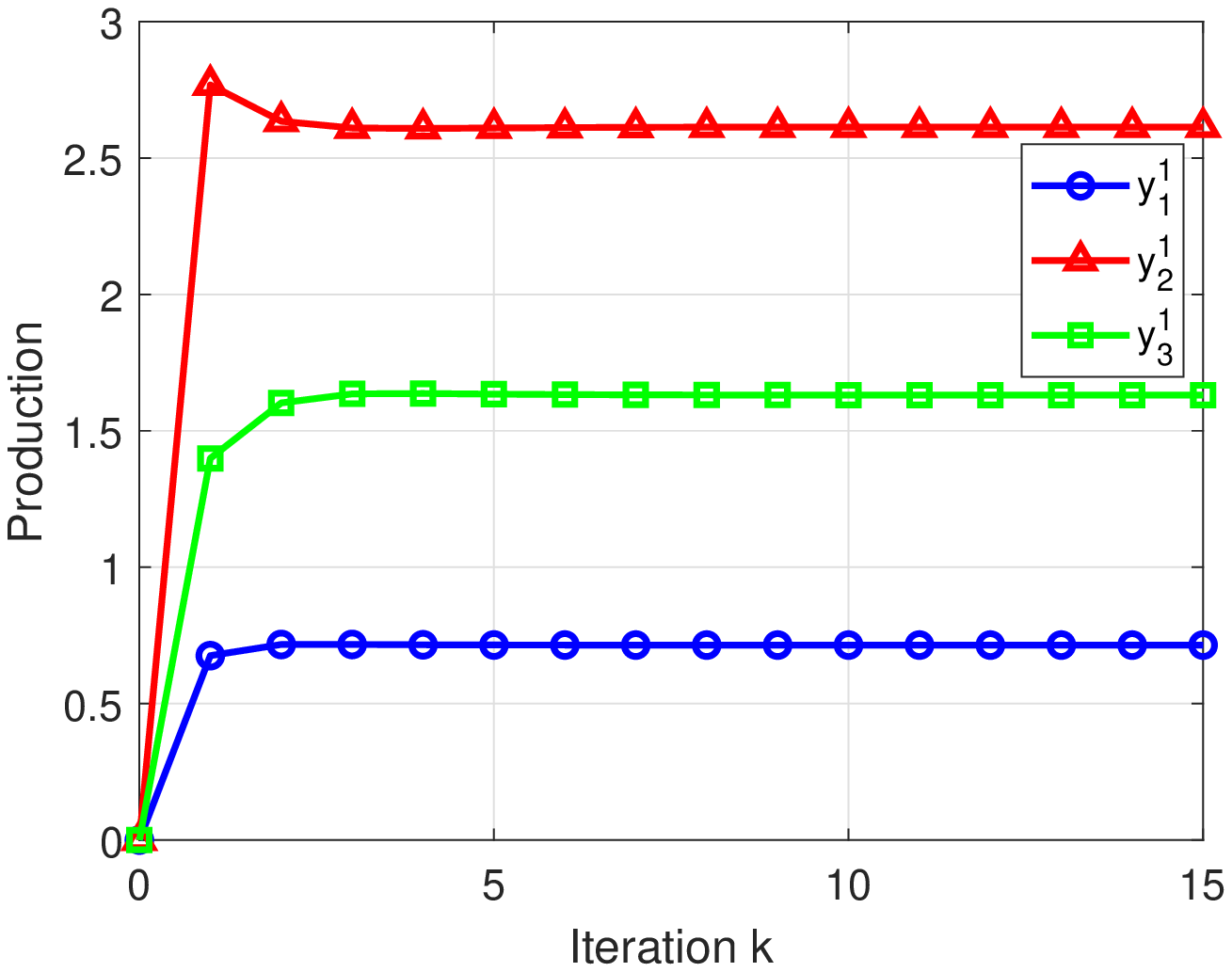}
			\caption{Firm 1}
		\end{subfigure}
		\hspace{-0.2in}\begin{subfigure}[b]{0.35\linewidth}
			\includegraphics[width=\linewidth]{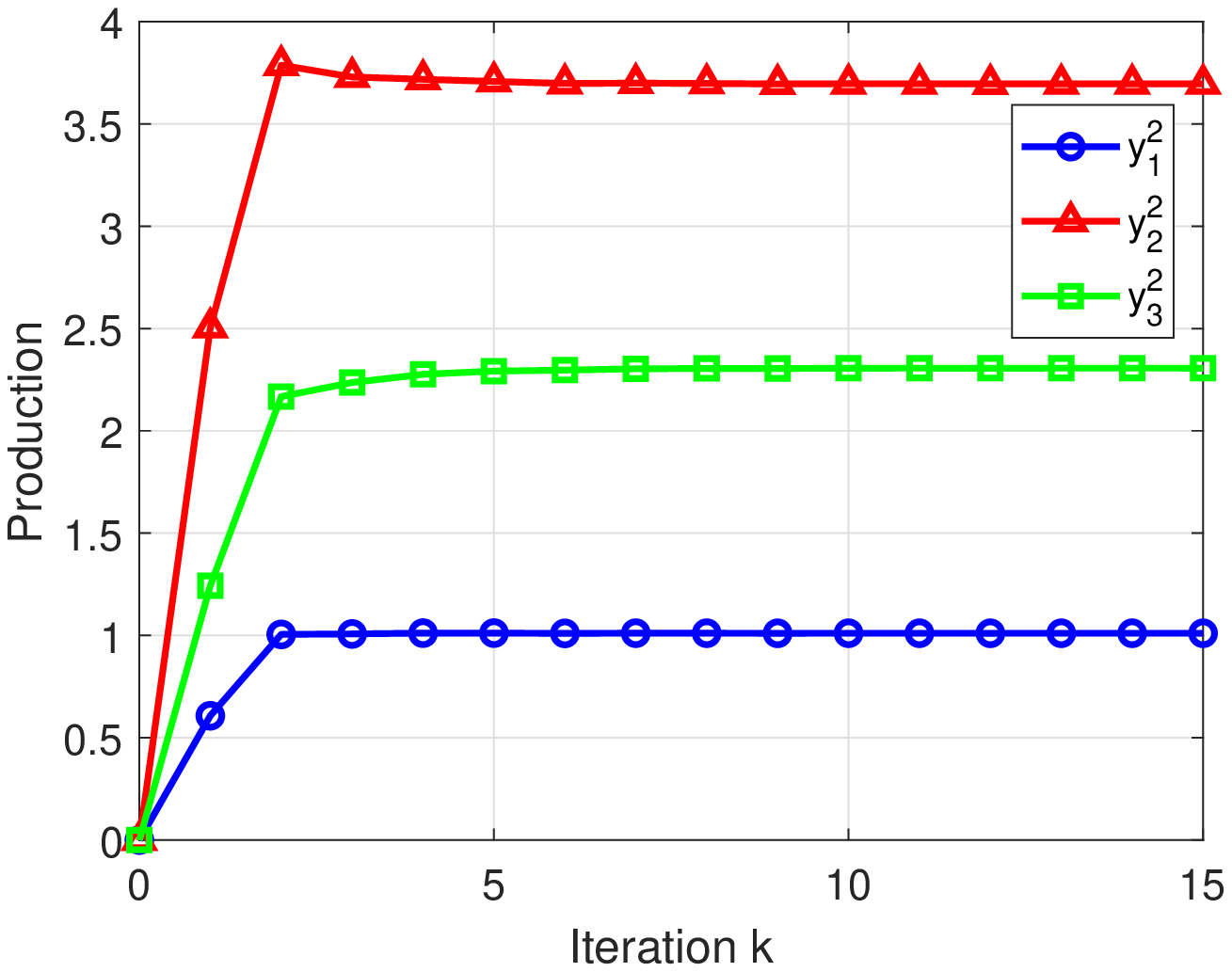}
			\caption{Firm 2}
		\end{subfigure}

	    	\centering
	        \hspace{-0.1in}\begin{subfigure}[b]{0.35\linewidth}
				\includegraphics[width=\linewidth]{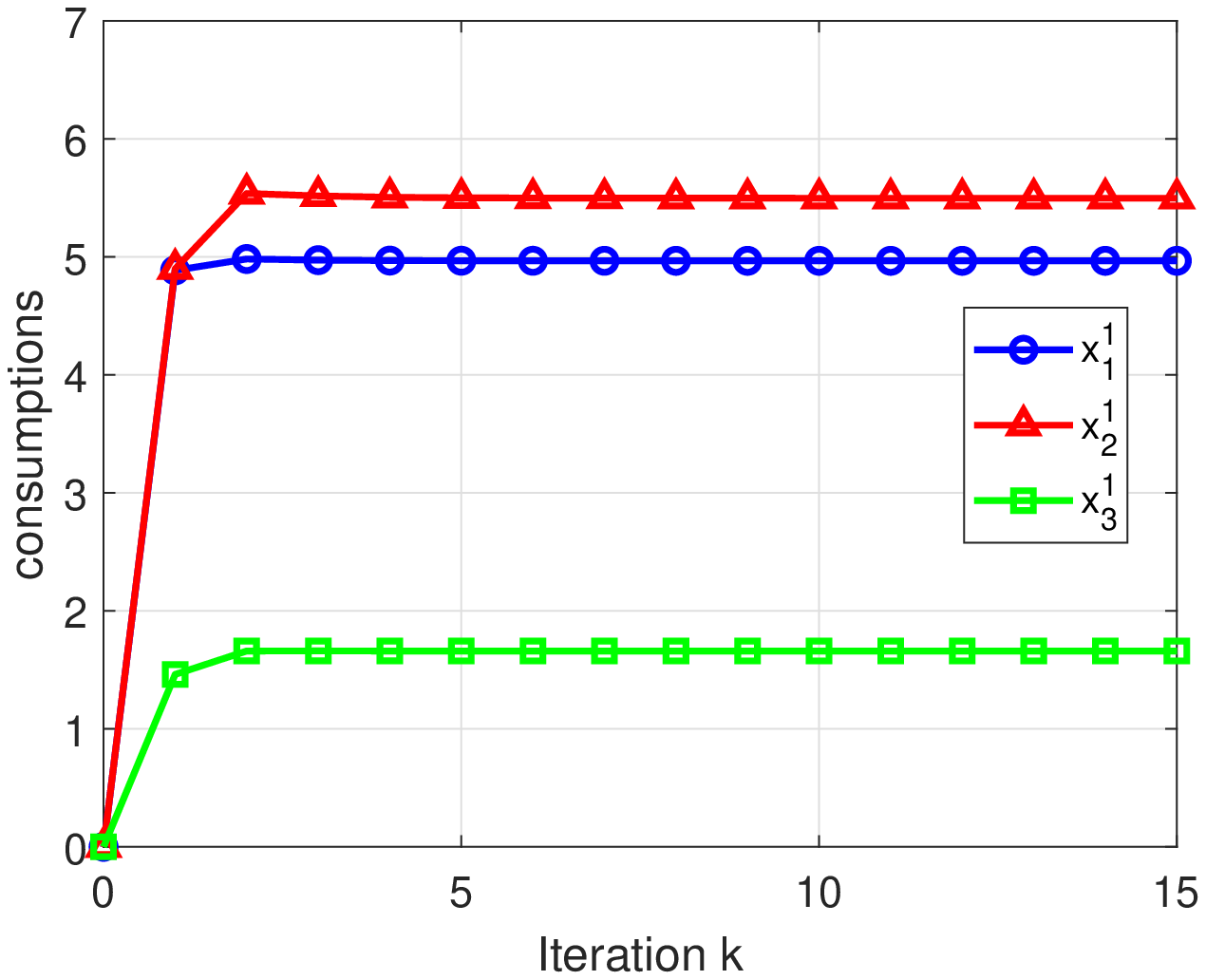}
				\caption{Consumer 1}
			\end{subfigure}
			\hspace{-0.2in}\begin{subfigure}[b]{0.35\linewidth}
				\includegraphics[width=\linewidth]{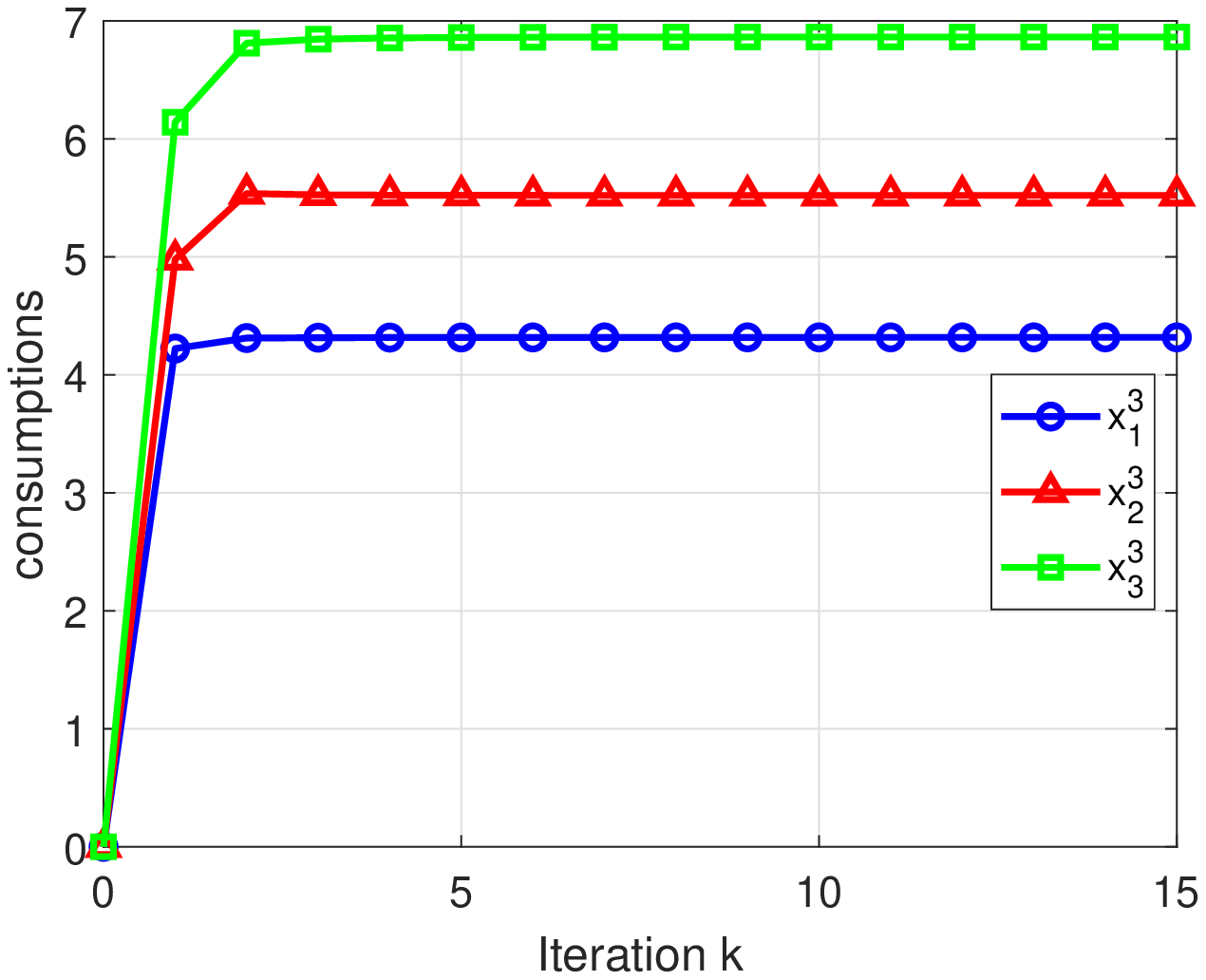}
				\caption{Consumer 3}
			\end{subfigure}	
			\hspace{-0.2in}\begin{subfigure}[b]{0.35\linewidth}
				\includegraphics[width=\linewidth]{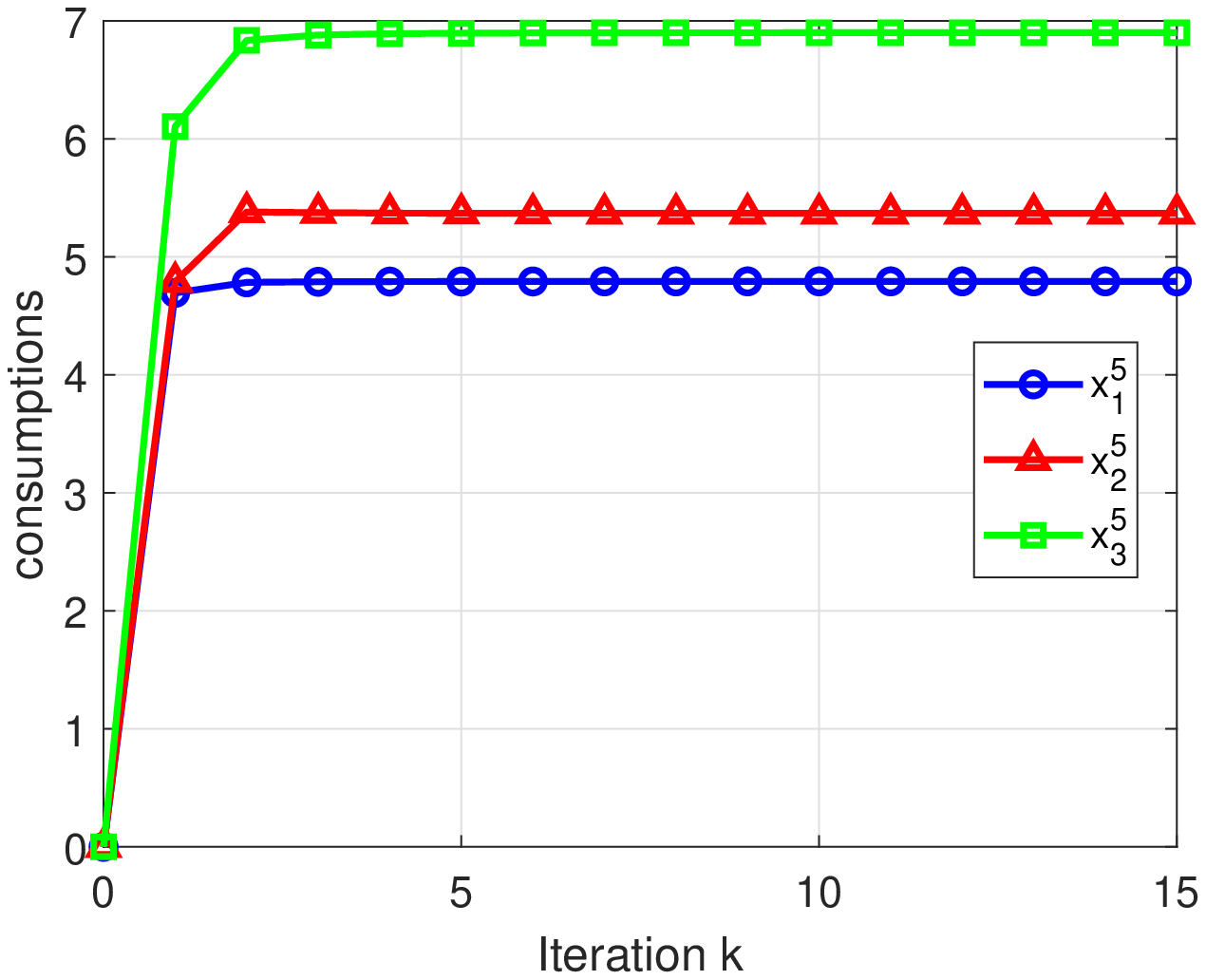}
				\caption{Consumer 5}		
			\end{subfigure}
			\caption{The sequence of decision variables of each player.} \label{fig_ex2_2}
		\end{figure}

\subsection*{Problem A.18 (Electricity market model).}
		
This electricity market model originally proposed by \cite{pang2005quasi}, and further discussed in \cite{nabetani2011parametrized}. There are two companies. Each company has an electricity plant in two out of three possible regions which are represented by the nodes of a graph. The goal is to maximize the profit of the company. The model has 18 variables, but we only present the reduced formulation with 12 variables. The reduction comes from the fact that both companies have plants on only 2 of the 3 nodes. We use the following abbreviations:
\[
\begin{aligned}
S_{1}  =  40-\frac{40}{500}(x_{1}+x_{4}+x_{7}+x_{10}), \
S_{2}  =  35-\frac{35}{400}(x_{2}+x_{5}+x_{8}+x_{11}), \
S_{3}  =  32-\frac{32}{600}(x_{3}+x_{6}+x_{9}+x_{12}).
\end{aligned}
\]
Player 1 has 6 variables $(x_{1},\ldots,x_{6})=(x_{1}^{1},\ldots,x_{6}^{1})$ and minimizes the objective function 
\[
\theta_{1}\left(x\right)=\left(15-S_{1}\right)\left(x_{1}+x_{4}\right)+\left(15-S_{2}\right)\left(x_{2}+x_{5}\right)+\left(15-S_{3}\right)\left(x_{3}+x_{6}\right),
\]
and player 1 has the nonnegativity constraints, $x_{1},\ldots,x_{6} \geq 0$, capacity constraints
\[
x_{1}+x_{2}+x_{3} \leq 100,\quad 
x_{4}+x_{5}+x_{6}  \leq  50,
\]
and coupling constraints $S_{j}-S_{i}\leq1, \ \forall i,j=1,2,3\;\mathrm{with}\;i\neq j$.
		
Player 2 has 6 variables $(x_{7},\ldots,x_{12})=(x_{1}^{2},\ldots,x_{6}^{2})$ and minimizes its objective 
\[
\theta_{2}\left(x\right)=\left(15-S_{1}\right)\left(x_{7}+x_{10}\right)+\left(15-S_{2}\right)\left(x_{8}+x_{11}\right)+\left(15-S_{3}\right)\left(x_{9}+x_{12}\right),
\]
and the nonnegativity constraints, $x_{7},\ldots,x_{12}\geq0$, capacity constraints
\[
x_{7}+x_{8}+x_{9} \leq 100, \quad
x_{10}+x_{11}+x_{12}  \leq  50,
\]
and coupling constraints $S_{j}-S_{i}\leq1, \ \forall i,j=1,2,3\;\mathrm{with}\;i\neq j$.
		
We only report numerical results for player 1 since player 2 has the same results. Figure \ref{fig:ex3} shows that company 1's P-Lagrangian is convergent, and Algorithm \ref{algorithm1} converges to a saddle-point that is a Nash equilibrium. 
	\begin{figure}[H]	
		\centering
		\hspace{-0.2in}	\includegraphics[scale=0.38]{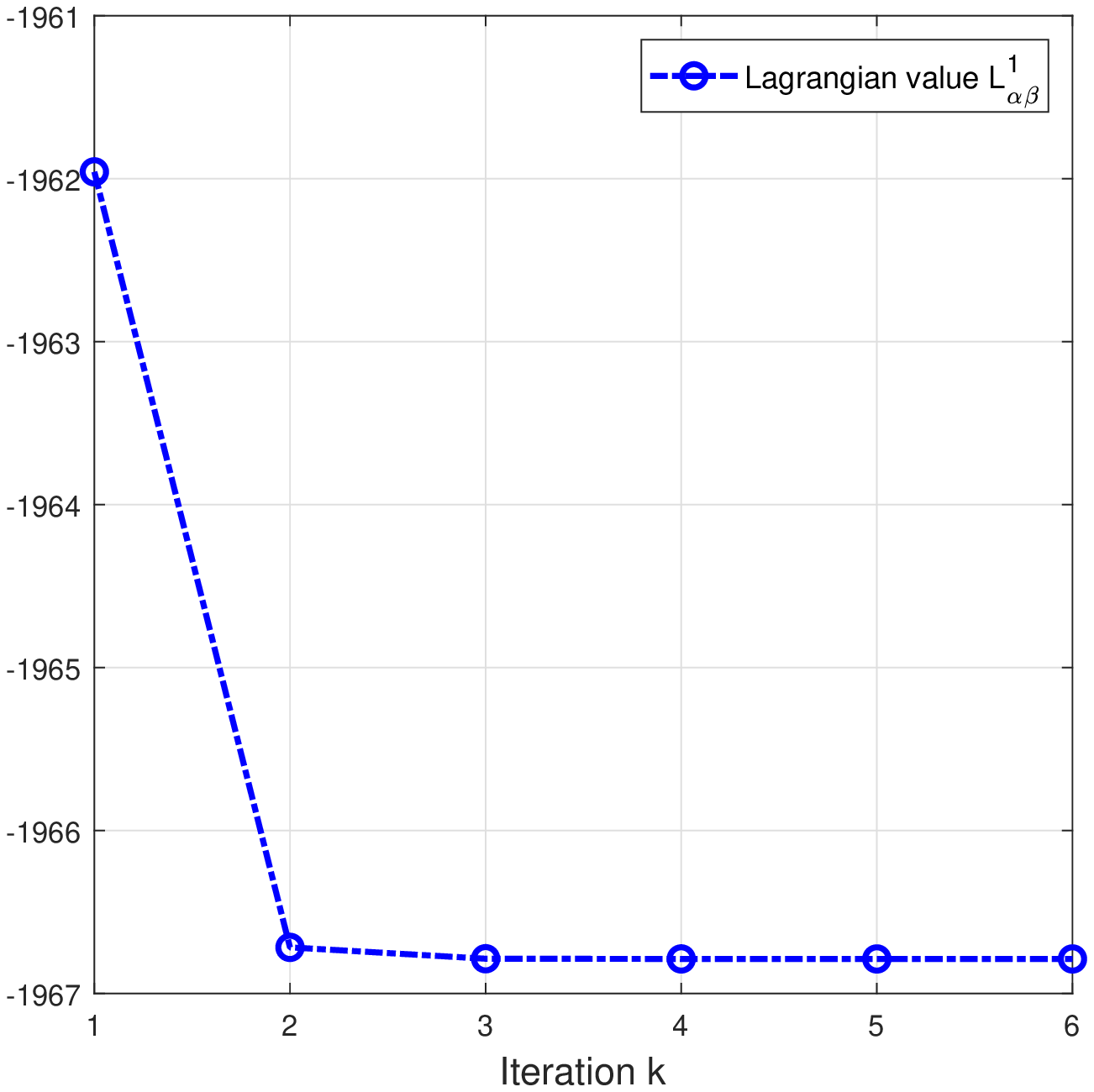}
		\includegraphics[scale=0.38]{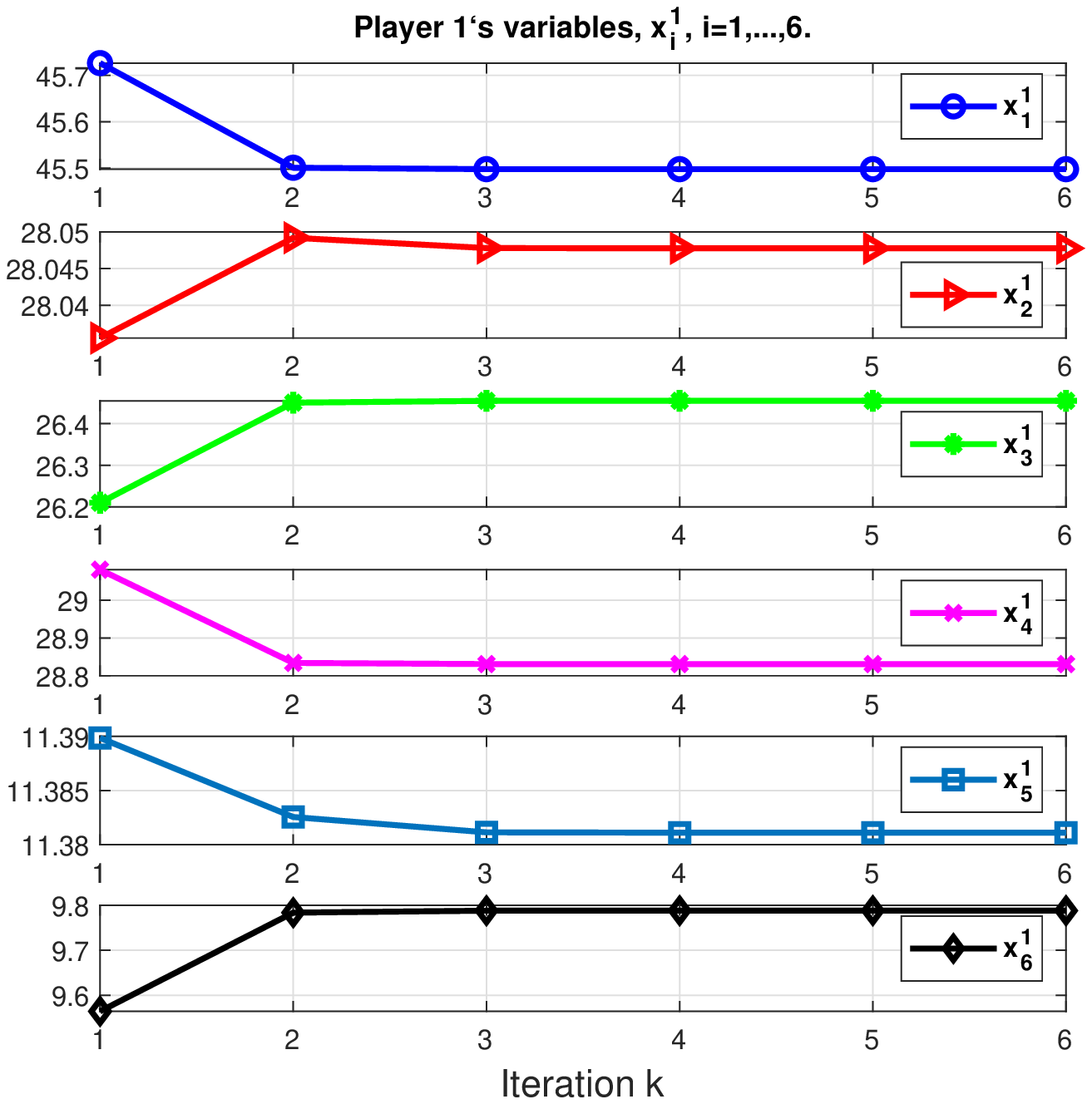}
			\vspace{-0.1in}
			\caption{{\small$\mathbf{x}^{1,0}=0$} converges to  {\small$\mathbf{x}^{1,\ast}=(45.4976,28.0478,26.4547,28.8309,11.3811,9.7880)$.}} \label{fig:ex3}
	\end{figure}
		
\section{Conclusions}

In this paper, we proposed a novel algorithmic framework for computing an equilibrium of generalized continuous Nash games (GNEPs) with theoretical guarantees based on the Proximal-Perturbed Lagrangian function. We have shown that the proposed method has significant advantages over existing approaches in both theoretical and computational perspectives; it does not require any boundedness assumptions and is the first development of an algorithm to solve \emph{general} GNEPs in a \emph{distributed} manner. The numerical results supported our theoretical findings. Possible future research is to extend our methodology to compute equilibria in stochastic Nash games with coupling constraints, which will result in a broader application domain.

\begin{acks}
The author would like to express his deep gratitude to John Birge for numerous insightful discussions and thoughtful suggestions on this work. The author would also like to thank the three anonymous reviewers for their valuable comments and suggestions that helped improve the paper.
\end{acks}

\bibliographystyle{ACM-Reference-Format}
\bibliography{references.bib}

\appendix
\begin{appendices}

\section*{\Large{Appendix}}
\vspace{0.1in}

\section{Proofs of Theorems in Section \ref{sec2}}

Before studying the relation between a saddle point of the P-Lagrangian \eqref{eq:MAL} and an equilibrium of the GNEP \eqref{eq:ef}, let us observe the following properties of $\mathcal{L}_{\alpha\beta}^{\nu}\left(x^{\nu},\boldsymbol{x}^{-\nu},z^{\nu},\lambda^{\nu},\mu^{\nu}\right)$.

\renewcommand\thetheorem{(a)}
\begin{observation}\label{ob_a}
Notice that the inner minimization in (\ref{eq:modfiedLagrangeDual}) can be split into two parts as follows:
\begin{multline}
    \underset{\lambda^{\nu}\in\mathbb{R}_{+}^{m_{\nu}},\mu^{\nu}\in\mathbb{R}^{m_{\nu}}}{\mathrm{max}}  
    \left\{ \underset{x^{\nu}\in\mathcal{X}_{\nu}}{\mathrm{min}} \left\lbrace \theta_{\nu}\left(x^{\nu},x^{-\nu}\right) + \left(\lambda^{\nu}\right)^{T}g^{\nu}\left(x^{\nu},x^{-\nu}\right) \right\rbrace \right.\\
    \left. + \underset{z^{\nu}\in\mathbb{R}^{m_{\nu}}}{\mathrm{min}} \left\lbrace -\left(\lambda^{\nu}-\mu^{\nu}\right)^T z^{\nu}+\frac{\alpha_{\nu}}{2}\left\| z^{\nu}\right\|^2 \right\rbrace - \frac{\beta_{\nu}}{2}\left\| \lambda^{\nu}-\mu^{\nu}\right\|^2 \right\}. \notag
\end{multline}
Denote by $z^{\nu} \left(\lambda^{\nu},\mu^{\nu}\right)$ as a unique solution of the problem, $\underset{z^{\nu}\in\mathbb{R}^{m_{\nu}}}{\mathrm{min}}\left\{ -\left(\lambda^{\nu}-\mu^{\nu}\right)^T z^{\nu} + \frac{\alpha_{\nu}}{2}\left\| z^{\nu} \right\|^2\right\}$ for given $\left(\lambda^{\nu},\mu^{\nu}\right)$. If we  minimize $\left\{ -\left(\lambda^{\nu}-\mu^{\nu}\right)^T z^{\nu}+\frac{\alpha_{\nu}}{2}\left\| z^{\nu} \right\|^2 \right\}$ with respect to $z^{\nu}$, we have
\begin{displaymath}
	z^{\nu}\left(\lambda^{\nu},\mu^{\nu}\right) = \frac{\lambda^{\nu}-\mu^{\nu}}{\alpha_{\nu}} \left(\Longleftarrow \left(\mu^{\nu}-\lambda^{\nu}\right)+\alpha_{\nu}z^{\nu}=0\right).
\end{displaymath}	
Recall that based on the above optimality condition for $z^{\nu}$, we can add a quadratic (regularization) term $-\frac{\beta_{\nu}}{2}\left\Vert \lambda^{\nu}-\mu^{\nu}\right\Vert^{2}$ to make the Lagrangian strongly concave in the multipliers $\lambda^{\nu}$ (for fixed $\mu^{\nu}$) and in  $\mu^{\nu}$ (for fixed $\lambda^{\nu}$) since it vanishes at the unique solution $z^{\nu,\ast}=0$. Substituting $z^{\nu}\left(\lambda^{\nu},\mu^{\nu}\right)$ into $\mathcal{L}_{\alpha\beta}^{\nu}\left(x^{\nu},{x}^{-\nu},z^{\nu},\lambda^{\nu},\mu^{\nu}\right)$, $\mathcal{L}_{\alpha\beta}^{\nu}$ reduces to 
\begin{equation}\label{eq:reducedAL}
	\mathcal{L}_{\alpha\beta}^{\nu}\left(x^{\nu},{x}^{-\nu},z^{\nu}\left(\lambda^{\nu},\mu^{\nu}\right),\lambda^{\nu},\mu^{\nu}\right)
	=\theta_{\nu}\left(x^{\nu},{x}^{-\nu}\right)+\left(\lambda^{\nu}\right)^{T}g^{\nu}\left(x^{\nu},{x}^{-\nu}\right)-\frac{1+\alpha_{\nu}\beta_{\nu}}{2\alpha_{\nu}}\left\Vert \lambda^{\nu}-\mu^{\nu}\right\Vert ^{2}.
\end{equation}
Then the P-Lagrangian dual problem can be expressed as 
\begin{equation}
	\underset{\lambda^{\nu} \in \mathbb{R}_{+}^{m_{\nu}},\mu^{\nu}\in\mathbb{R}^{m_{\nu}}}{\mathrm{max}}\left\{ \mathcal{D}_{\alpha\beta}^{\nu}\left(\lambda^{\nu},\mu^{\nu}\right) := D_{0}^{\nu}\left(\lambda^{\nu}\right)-\frac{1+\alpha_{\nu}\beta_{\nu}}{2\alpha_{\nu}}\left\Vert \lambda^{\nu}-\mu^{\nu}\right\Vert ^{2}\right\},
\end{equation}
where $D_{0}^{\nu}\left(\lambda^{\nu}\right)=\underset{x^{\nu}\in \mathcal{X}_{\nu}}{\mathrm{min}}\left\{ \theta_{\nu}\left(x^{\nu},x^{-\nu}\right)+\left(\lambda^{\nu}\right)^{T}g^{\nu}\left(x^{\nu},x^{-\nu}\right)\right\}$, which is identical to the standard dual function associated with the original problem \eqref{eq:of}. Thus the P-Lagrangian dual function $\mathcal{D}_{\alpha\beta}^{\nu}\left(\lambda^{\nu},\mu^{\nu}\right)$ is maximized jointly in $\lambda^{\nu}$ and $\mu^{\nu}$ if and only if $\lambda^{\nu}$ maximizes $D_{0}^{\nu}\left(\lambda^{\nu}\right)$ and $\lambda^{\nu}=\mu^{\nu}$. This implies that the multiplier  $\lambda^{\nu,\ast}$ for the constraint $g^{\nu}\left(x^{\nu},x^{-\nu}\right)-z^{\nu}\leq0$ in extended problem $\eqref{eq:ef}$ is precisely to the multiplier $\eta^{\nu,\ast}$ for the constraint $g^{\nu}\left(x^{\nu},x^{-\nu}\right)\leq0$ in problem \eqref{eq:of}.
\end{observation}

\renewcommand\thetheorem{(b)}
\begin{observation}\label{ob_b}
If we maximizing $\mathcal{L}_{\alpha\beta}^{\nu}(x^{\nu},\boldsymbol{x}^{-\nu},z^{\nu},\lambda^{\nu},\mu^{\nu})$ with respect to $\mu^{\nu}$, we get
\[
	\nabla_{\mu^{\nu}}\mathcal{L}_{\alpha\beta}^{\nu}\left(x^{\nu},\boldsymbol{x}^{-\nu},z^{\nu},\lambda^{\nu},\mu^{\nu}\right)=z^{\nu}+\beta_{\nu}\left(\lambda^{\nu}-\mu^{\nu}\right)=0,
\]
which, together with the fact that  $\lambda^{\nu,\ast}=\mu^{\nu,\ast}$, implies that $z^{\nu}=0$ at the maximizers $\left(\lambda^{\nu,\ast},\mu^{\nu,\ast}\right)$ and $\beta_{\nu}>0$. 
\end{observation}

Using the Observations \ref{ob_a} and \ref{ob_b}, we now establish the equivalence between a saddle point of $\mathcal{L}_{\alpha\beta}^{\nu}$ and an equilibrium of the GNEP \eqref{eq:of}. 

\subsection{Proof of Theorem \ref{lem_saddle_to_eq}} \label{A.1}
\begin{proof}
Using the reduced form of P-Lagrangian \eqref{eq:reducedAL}, we have
\begin{equation}
		\begin{aligned}\mathcal{L}_{\alpha\beta}^{\nu}\left(\mathbf{x}^{\ast},z^{\nu,\ast}\left(\lambda^{\nu},\mu^{\nu}\right),\lambda^{\nu},\mu^{\nu}\right) & =\theta_{\nu}\left(\mathbf{x}^{\ast}\right)+\left(\lambda^{\nu}\right)^{T}g^{\nu}\left(\mathbf{x}^{\ast}\right)-\frac{1+\alpha_{\nu}\beta_{\nu}}{2\alpha_{\nu}}\left\Vert \lambda^{\nu}-\mu^{\nu}\right\Vert ^{2}\\
		& \leq\mathcal{L}_{\alpha\beta}^{\nu}\left(\mathbf{x}^{\ast},z^{\nu,\ast}\left(\lambda^{\nu,\ast},\mu^{\nu,\ast}\right),\lambda^{\nu,\ast},\mu^{\nu,\ast}\right).
		\end{aligned}
		\label{eq:lem_saddle_to_eq_1}
\end{equation}
First, we prove that $\mathbf{x}^{\ast}=(x^{\nu,\ast},\boldsymbol{x}^{-\nu,\ast})$ is feasible for problem \eqref{eq:of}. Suppose by contradiction that $\mathbf{x}^{\ast}$ is infeasible, i.e., $g_{i}^{\nu}\left(\mathbf{x}^{\ast}\right)>0$ for some $i$. Then there exist some $\lambda_{i}^{\nu}$ such that $\lambda_{i}^{\nu}g_{i}^{\nu}\left(\mathbf{x}^{\ast}\right)\rightarrow\infty$ as $\lambda_{i}^{\nu}\rightarrow\infty$. This implies $\mathcal{L}_{\alpha\beta}^{\nu}\left(\mathbf{x}^{\ast},z^{\nu,\ast},\lambda^{\nu},\mu^{\nu}\right)\rightarrow\infty$ by taking the limit as $\lambda_{i}^{\nu}\rightarrow\infty$ with $\lambda_{i}^{\nu}=\mu_{i}^{\nu}$ to maximize the left-hand side of the first inequality in \eqref{eq:lem_saddle_to_eq_1}, which is a contradiction with the first inequality in \eqref{eq:saddle}. 
Therefore $g_{i}^{\nu}(\mathbf{x}^{\ast})\leq0$ for all $i=1,\ldots,m_{\nu}$. By the definition  $\mathcal{L}_{\alpha\beta}^{\nu}\left(\mathbf{x}^{\ast},z^{\nu,\ast},\lambda^{\nu,\ast},\mu^{\nu,\ast}\right)=\mathrm{sup}_{\lambda^{\nu}\geq0,\mu^{\nu}}\mathcal{L}_{\alpha\beta}^{\nu}\left(\mathbf{x}^{\ast},z^{\nu,\ast},\lambda^{\nu},\mu^{\nu}\right)$ with the fact that $g^{\nu}(\mathbf{x}^{\ast})\leq0$ and  $\lambda^{\nu,\ast}\geq0$, we have $(\lambda^{\nu,\ast})^{T}g^{\nu}(\mathbf{x}^{\ast})=0$ and $\lambda^{\nu,\ast}=\mu^{\nu,\ast}$. It  follows that
$$\mathcal{L}_{\alpha\beta}^{\nu}\left(\mathbf{x}^{\ast},z^{\nu,\ast},\lambda^{\nu,\ast},\mu^{\nu,\ast}\right)=\theta_{\nu}\left(\mathbf{x}^{\ast}\right).	\label{eq:lem_saddle_to_eq_2}$$
	
Next, let $x^{\nu}\in X_{\nu}\left(\boldsymbol{x}^{-\nu,\ast}\right)$ be any feasible solution to problem \eqref{eq:of}. For any feasible $x^{\nu}$ and $\lambda_{i}^{\nu}\geq0$, since $g_{i}^{\nu}\left(x^{\nu},\boldsymbol{x}^{-\nu,\ast}\right)\leq0$, we have \begin{equation} \label{eq:lem_saddle_to_eq_3}
	\left(\lambda^{\nu}\right)^{T}g^{\nu}\left(x^{\nu},\boldsymbol{x}^{-\nu,\ast}\right)-\frac{1+\alpha_{\nu}\beta_{\nu}}{2\alpha_{\nu}}\left\Vert \lambda^{\nu}-\mu^{\nu}\right\Vert ^{2}\leq\left(\lambda^{\nu}\right)\cdot0-\frac{1+\alpha_{\nu}\beta_{\nu}}{2\alpha_{\nu}}\left\Vert \lambda^{\nu}-\mu^{\nu}\right\Vert ^{2}\leq 0.
\end{equation}
From Observation \ref{ob_b} that $z^{\nu}=0$ when $\lambda^{\nu,\ast}=\mu^{\nu,\ast}$ for any $\beta_{\nu}>0$, we have
\begin{equation} \label{eq:lem_saddle_to_eq_4}
	-\left(\lambda^{\nu,\ast}-\mu^{\nu,\ast}\right)^{T}z^{\nu}+\frac{\alpha_{\nu}}{2}\left\Vert z^{\nu}\right\Vert ^{2}=0.
\end{equation}
The second inequality of the saddle point condition   \eqref{eq:saddle} yields
\begin{displaymath}
        \begin{aligned}\theta_{\nu}\left(\mathbf{x}^{\ast}\right) 
		&\leq
		\mathcal{L}_{\alpha\beta}^{\nu}\left(x^{\nu},\boldsymbol{x}^{-\nu,\ast},z^{\nu},\lambda^{\nu,\ast},\mu^{\nu,\ast}\right) \\
		& = \theta_{\nu}\left(x^{\nu},\boldsymbol{x}^{-\nu,\ast}\right)+\underbrace{\left(\lambda^{\nu,\ast}\right)^{T}g^{\nu}\left(x^{\nu},\boldsymbol{x}^{-\nu,\ast}\right)-\frac{\beta_{\nu}}{2}\left\Vert \lambda^{\nu,\ast}-\mu^{\nu,\ast}\right\Vert ^{2}}_{\leq0} -\underbrace{\left(\lambda^{\nu,\ast}-\mu^{\nu,\ast}\right)^{T}z^{\nu}+\frac{\alpha_{\nu}}{2}\left\Vert z^{\nu}\right\Vert ^{2}}_{=0} \\
		&
		\leq
		\theta_{\nu}\left(x^{\nu},\boldsymbol{x}^{-\nu,\ast}\right),
		\end{aligned}
\end{displaymath}
where the last inequality is from \eqref{eq:lem_saddle_to_eq_3} and \eqref{eq:lem_saddle_to_eq_4}. Hence, $\mathbf{x}^{\ast}=(x^{\nu,\ast},\boldsymbol{x}^{-\nu,\ast})$ is a GNE of problem \eqref{eq:of}.
\end{proof}

\subsection{Proof of Theorem \ref{lem_eq_to_saddle}} \label{A.2}
\begin{proof}
First, we show that the first inequality in \eqref{eq:lem_eq_to_saddle_sp} holds. From the feasibility of an equilibrium $\mathbf{x}^{\ast}$, we have for any $\lambda^{\nu}\in\mathbb{R}_{+}^{m_{\nu}}$, $\mu^{\nu}\in\mathbb{R}^{m_{\nu}}$ and $\alpha_{\nu},\beta_{\nu}>0$
\begin{equation}\label{eq:lem_eq_to_saddle_1}
	\left(\lambda^{\nu}\right)^{T}g^{\nu}\left(\mathbf{x}^{\ast}\right)-\frac{1+\alpha_{\nu}\beta_{\nu}}{2\alpha_{\nu}}\left\Vert \lambda^{\nu}-\mu^{\nu}\right\Vert ^{2}\leq 0,
\end{equation}
implying that 
\[
	\theta_{\nu}(\mathbf{x}^{\ast})+(\lambda^{\nu})^{T}g^{\nu}(\mathbf{x}^{\ast})-\frac{1+\alpha_{\nu}\beta_{\nu}}{2\alpha_{\nu}}\left\Vert \lambda^{\nu}-\mu^{\nu}\right\Vert ^{2}
	=\mathcal{L}_{\alpha\beta}^{\nu}(\mathbf{x}^{\ast},z^{\nu,\ast}(\lambda^{\nu},\mu^{\nu}),\lambda^{\nu},\mu^{\nu}) 
	\leq \theta_{\nu}(\mathbf{x}^{\ast}).
\]
On the other hand, since there exists a pair of multipliers $\left(\lambda^{\nu,\ast},\mu^{\nu,\ast}\right)$ maximizing $\mathcal{L}_{\alpha\beta}^{\nu}(\mathbf{x}^{\ast},z^{\nu,\ast},\lambda^{\nu},\mu^{\nu})$, we also have for {\(\lambda^{\nu}=\mu^{\nu}=0\)}
\[
\begin{aligned}
	\mathcal{L}_{\alpha\beta}^{\nu}\left(\mathbf{x}^{\ast},z^{\nu,\ast}\left(0,0\right),0,0\right) & 
	\leq  \mathcal{L}_{\alpha\beta}^{\nu}\left(\mathbf{x}^{\ast},z^{\nu,\ast},\lambda^{\nu,\ast},\mu^{\nu,\ast}\right) \\
	& = \theta_{\nu}\left(\mathbf{x}^{\ast}\right)+\left(\lambda^{\nu,\ast}\right)^{T}g^{\nu}\left(\mathbf{x}^{\ast}\right)-\frac{1+\alpha_{\nu}\beta_{\nu}}{2\alpha_{\nu}}\left\Vert \lambda^{\nu,\ast}-\mu^{\nu,\ast}\right\Vert ^{2},
\end{aligned}
\]
which together with the fact that $\mathcal{L}_{\alpha\beta}^{\nu}\left(\mathbf{x}^{\ast},z^{\nu,\ast}\left(0,0\right),0,0\right)=\theta_{\nu}\left(\mathbf{x}^{\ast}\right)$ gives
\begin{equation} \label{eq:lem_eq_to_saddle_2}
	\left(\lambda^{\nu,\ast}\right)^{T}g^{\nu}\left(\mathbf{x}^{\ast}\right)-\frac{1+\alpha_{\nu}\beta_{\nu}}{2\alpha_{\nu}}\left\Vert \lambda^{\nu,\ast}-\mu^{\nu,\ast}\right\Vert ^{2}\geq0.
\end{equation}
Combining \eqref{eq:lem_eq_to_saddle_1} and \eqref{eq:lem_eq_to_saddle_2}, we obtain
\begin{displaymath}
	\left(\lambda^{\nu,\ast}\right)^{T}g^{\nu}\left(\mathbf{x}^{\ast}\right)-\frac{1+\alpha_{\nu}\beta_{\nu}}{2\alpha_{\nu}}\left\Vert \lambda^{\nu,\ast}-\mu^{\nu,\ast}\right\Vert ^{2}=0.
\end{displaymath}
It thus follows that 
\[\mathcal{L}_{\alpha\beta}^{\nu}\left(\mathbf{x}^{\ast},z^{\nu,\ast},\lambda^{\nu,\ast},\mu^{\nu,\ast}\right)=\theta_{\nu}\left(\mathbf{x}^{\ast}\right).
\]
Hence, the first inequality in \eqref{eq:lem_eq_to_saddle_sp} holds. Furthermore, by using the fact that  $g^{\nu}\left(\mathbf{x}^{\ast}\right)\leq0$ and $\lambda^{\nu,\ast}\geq0$  with \eqref{eq:lem_eq_to_saddle_2}, we have that $0\geq\left(\lambda^{\nu,\ast}\right)^{T}g^{\nu}\left(\mathbf{x}^{\ast}\right)\geq\frac{1+\alpha_{\nu}\beta_{\nu}}{2\alpha_{\nu}}\left\Vert \lambda^{\nu,\ast}-\mu^{\nu,\ast}\right\Vert ^{2}\geq0,$ which implies the  multiplier $\lambda^{\nu,\ast}$ satisfies the complementarity slackness $\left(\lambda^{\nu,\ast}\right)^{T}g^{\nu}(x^{\nu,\ast},x^{-\nu,\ast})=0$ and $\lambda^{\nu,\ast}=\mu^{\nu,\ast}$. Thus the maximizer $\lambda^{\nu,\ast}$ is equivalent to the Lagrange multiplier $\eta^{\nu,\ast}$ satisfying the KKT conditions \eqref{eq:KKT} for the original GNEP \eqref{eq:of}.
	
Next, we show that $\mathcal{L}_{\alpha\beta}^{\nu}(\mathbf{x}^{\ast},z^{\nu,\ast},\lambda^{\nu,\ast},\mu^{\nu,\ast})\leq\mathcal{L}_{\alpha\beta}^{\nu}(x^{\nu},\boldsymbol{x}^{-\nu,\ast},z^{\nu},\lambda^{\nu,\ast},\mu^{\nu,\ast})$ in \eqref{eq:lem_eq_to_saddle_sp} holds.
Noticing that $\nabla_{z^{\nu}}L_{\alpha}^{\nu}\left(z^{\nu},\lambda^{\nu,\ast},\mu^{\nu,\ast}\right) = -\left(\lambda^{\nu,\ast}-\mu^{\nu,\ast}\right) + \alpha_{\nu}z^{\nu,\ast}=0$ and $\lambda^{\nu,\ast}=\mu^{\nu,\ast}$, we have a unique minimizer $z^{\nu,\ast}=0$ for some $\alpha_{\nu}>0$. From the convexity $\theta_{\nu}\left(x^{\nu},\boldsymbol{x}^{-\nu,\ast}\right)$ and $g^{\nu}\left(x^{\nu},\boldsymbol{x}^{-\nu,\ast}\right)$ in $x^{\nu}$, we have
\begin{flalign*}
	& \theta_{\nu}\left(x^{\nu},\boldsymbol{x}^{-\nu,\ast}\right)\geq  \theta_{\nu}\left(\mathbf{x}^{\ast}\right)+\nabla_{x^{\nu}}\theta_{\nu}\left(\mathbf{x}^{\ast}\right)^{T}\left(x^{\nu}-x^{\nu,\ast}\right),\\
	& 
	g^{\nu}\left(x^{\nu},\boldsymbol{x}^{-\nu,\ast}\right)
	\geq  g^{\nu}\left(\mathbf{x}^{\ast}\right)+\nabla_{x^{\nu}}g^{\nu}\left(\mathbf{x}^{\ast}\right)^{T}\left(x^{\nu}-x^{\nu,\ast}\right).
\end{flalign*}
Then we have
\[
\begin{aligned}
	&  \mathcal{L}_{\alpha\beta}^{\nu}\left(x^{\nu},\boldsymbol{x}^{-\nu,\ast},z^{\nu},\lambda^{\nu,\ast},\mu^{\nu,\ast}\right) \\
	& \geq \theta_{\nu}\left(\mathbf{x}^{\ast}\right)+\left(\lambda^{\nu,\ast}\right)^{T}g^{\nu}\left(\mathbf{x}^{\ast}\right)
	+\left(\nabla_{x^{\nu}}\theta_{\nu}\left(\mathbf{x}^{\ast}\right)+\sum_{i=1}^{m_{\nu}}\lambda_{i}^{\nu,\ast}\nabla_{x^{\nu}} g_{i}^{\nu}\left(\mathbf{x}^{\ast}\right)\right)^{T}\left(x^{\nu}-x^{\nu,\ast}\right)\\
	& \quad-\left(\lambda^{\nu,\ast}-\mu^{\nu,\ast}\right)^{T}z^{\nu}+\frac{\alpha_{\nu}}{2}\left\Vert z^{\nu}\right\Vert ^{2}-\frac{\beta_{\nu}}{2}\left\Vert \lambda^{\nu,\ast}-\mu^{\nu,\ast}\right\Vert ^{2}\\
	& \geq\theta_{\nu}\left(\mathbf{x}^{\ast}\right)+\left(\lambda^{\nu,\ast}\right)^{T}g^{\nu}\left(\mathbf{x}^{\ast}\right)+\frac{\alpha_{\nu}}{2}\left\Vert z^{\nu}\right\Vert ^{2} \\
	& \geq\theta_{\nu}\left(\mathbf{x}^{\ast}\right)=\mathcal{L}_{\alpha\beta}^{\nu}(\mathbf{x}^{\ast},z^{\nu,\ast}=0,\lambda^{\nu,\ast},\mu^{\nu,\ast}).
\end{aligned}
\]
Hence, $\left(\mathbf{x}^{\ast},z^{\nu,\ast}=0\right)$ satisfies the second inequality of \eqref{eq:lem_eq_to_saddle_sp}.
\end{proof}

\section{Proof of Lemma \ref{prop_converge_inner}}\label{B}

\begin{proof}
(a) Fix $k\geq0$, and omit the iterates $\left(z^{\nu,k}, \lambda^{\nu,k}, \mu^{\nu,k}\right)$ for simplicity in the proof. 
Let  $\widehat{x}^{\nu,k}$ be $\nu$th component of  the unique solution $\widehat{\mathbf{x}}^{k}$. From the fixed-point characterization of $\widehat{x}^{\nu,k}$ 
\[
\widehat{x}^{\nu,k} = \mathcal{P}_{\mathcal{X}_{\nu}} \left[ \widehat{x}^{\nu,k} - \sigma_{\nu}  \nabla_{x^{\nu}}\widehat{\mathcal{L}_{\alpha\beta}^{\nu}}\left(\widehat{\mathbf{x}}^{\nu,k}; \mathbf{x}^{k}\right) \right]
\]
and the contraction property of projection operator $\mathcal{P}_{\mathcal{X}_{\nu}} \left[ \bullet\right]$, we have that for all $\nu=1,\ldots,N$
{\small\begin{equation} \label{eq:prop_inner_eq1}
\begin{aligned}	
\quad \left\|{u}^{\nu,k,l+1} -\widehat{x}^{\nu,k}  \right\|^2 
& =
\left\|\mathcal{P}_{\mathcal{X}_{\nu}} \left[ u^{\nu,k,l} - \sigma_{\nu}  \nabla_{x^{\nu}}\widehat{\mathcal{L}_{\alpha\beta}^{\nu}}\left(\mathbf{u}^{k,l}; \mathbf{x}^{k}\right) \right] -\widehat{x}^{\nu,k} \right\|^2\\
& =
\left\|\mathcal{P}_{\mathcal{X}_{\nu}} \left[ u^{\nu,k,l} - \sigma_{\nu}  \nabla_{x^{\nu}}\widehat{\mathcal{L}_{\alpha\beta}^{\nu}}\left(\mathbf{u}^{k,l}; \mathbf{x}^{k}\right) \right]  - \mathcal{P}_{\mathcal{X}_{\nu}} \left[ \widehat{x}^{\nu,k} - \sigma_{\nu}  \nabla_{x^{\nu}}\widehat{\mathcal{L}_{\alpha\beta}^{\nu}}\left( \widehat{\mathbf{x}}^{k}; \mathbf{x}^{k}\right) \right]  \right\|^2
\\
& \leq
\left\| \left[ u^{\nu,k,l} - \sigma_{\nu}  \nabla_{x^{\nu}}\widehat{\mathcal{L}_{\alpha\beta}^{\nu}}\left(\mathbf{u}^{k,l}; \mathbf{x}^{k}\right) \right] 
- 
\left[ \widehat{x}^{\nu,k} - \sigma_{\nu}  \nabla_{x^{\nu}}\widehat{\mathcal{L}_{\alpha\beta}^{\nu}}\left( \widehat{\mathbf{x}}^{k}; \mathbf{x}^{k}\right) \right]  \right\|^2.
\end{aligned}
\end{equation}}
By expanding the last term on the right, the above inequality can be rewritten as 
{\small\begin{equation} \label{eq:prop_inner_eq2}
\begin{aligned}
\left\| {u}^{\nu,k,l+1} - \widehat{x}^{\nu,k}  \right\|^2 
& \leq
\left\| {u}^{\nu,k,l} -\widehat{x}^{\nu,k}  \right\|^2 
- 2 \sigma_{\nu}
\left(   \nabla_{x^{\nu}}\widehat{\mathcal{L}_{\alpha\beta}^{\nu}} \left(\mathbf{u}^{k,l}; \mathbf{x}^{k}\right) 
- 
\nabla_{x^{\nu}}\widehat{\mathcal{L}_{\alpha\beta}^{\nu}}\left(\widehat{\mathbf{x}}^{k}; \mathbf{x}^{k}\right) \right)^T \left({u}^{\nu,k,l} -\widehat{x}^{\nu,k} \right) \\
&
\quad + \sigma_{\nu}^2 \left\| \nabla_{x^{\nu}}\widehat{\mathcal{L}_{\alpha\beta}^{\nu}}\left(\mathbf{u}^{k,l}; \mathbf{x}^{k}\right) 
- 
\nabla_{x^{\nu}}\widehat{\mathcal{L}_{\alpha\beta}^{\nu}} \left( \widehat{\mathbf{x}}^{k}; \mathbf{x}^{k} \right) \right\|^2.
\end{aligned}
\end{equation}}
Since  {\small $\widehat{\mathcal{L}_{\alpha \beta}^{\nu}}$} is strongly convex in $\mathbf{x}=\left(x^{\nu}, x^{-\nu} \right) $ with constant $c_{\nu}$ and  {\small $\nabla_{\mathbf{x}}\widehat{\mathcal{L}_{\alpha \beta}^{\nu}}$} is Lipschitz continuous with constant $\widehat{L}_{\nu}$  (see properties \ref{itm:p3} and \ref{itm:p4} in Remark \ref{rem_prop_approx_AL}), we can
estimate the second term and third term on right-hand side of \eqref{eq:prop_inner_eq2} 
{\small
\begin{equation} \label{eq:prop_inner_eq3}
\left(   \nabla_{x^{\nu}}\widehat{\mathcal{L}_{\alpha\beta}^{\nu}}\left(\mathbf{u}^{k,l}; \mathbf{x}^{k}\right) 
- 
\nabla_{x^{\nu}}\widehat{\mathcal{L}_{\alpha\beta}^{\nu}}\left( \widehat{\mathbf{x}}^{k}; \mathbf{x}^{k}\right) \right)^T \left({u}^{\nu,k,l} -\widehat{x}^{\nu,k} \right) \geq
c_{\nu} \left\| {u}^{\nu,k,l} -\widehat{x}^{\nu,k} \right\|^2,
\end{equation}}
{\small\begin{equation} \label{eq:prop_inner_eq4}
\left\| \nabla_{x^{\nu}}\widehat{\mathcal{L}_{\alpha\beta}^{\nu}}\left(\mathbf{u}^{k,l}; \mathbf{x}^{k}\right) 
- 
\nabla_{x^{\nu}}\widehat{\mathcal{L}_{\alpha\beta}^{\nu}}\left( \widehat{\mathbf{x}}^{k}; \mathbf{x}^{k}\right) \right\| \leq \widehat{L}_{\nu} \left\|  {u}^{\nu,k,l} -\widehat{x}^{\nu,k} \right\|.
\end{equation}}
Note that since {\small $\widehat{\mathcal{L}_{\alpha \beta}^{\nu}}$} is a proximal linearized function with the quadratic term $\frac{\gamma_{\nu}}{2}\left\Vert \mathbf{x}-\mathbf{x}^{k}\right\Vert^{2}$, we can take $\gamma_{\nu}=c_\nu$. Substituting \eqref{eq:prop_inner_eq3} with $c_\nu=\gamma_{\nu}$ and \eqref{eq:prop_inner_eq4} into \eqref{eq:prop_inner_eq2} yields
\[
\begin{aligned}
\left\| {u}^{\nu,k,l+1} - \widehat{x}^{\nu,k}  \right\|^2
\leq 
\left( 1 - 2 \gamma_{\nu}\sigma_{\nu} +\sigma_{\nu}^2 \widehat{L}_{\nu}^2\right)
\left\| {u}^{\nu,k,l} - \widehat{x}^{\nu,k}  \right\|^2. 
\end{aligned}
\]
Notice that $\left( 1 - 2\gamma_{\nu} \sigma_{\nu} + \sigma_{\nu}^2 \widehat{L}_{\nu}^2\right) \geq 0$ is satisfied since $\widehat{L}_{\nu} \geq \gamma_{\nu} $. Now, setting $\widehat{\sigma}:=\underset{\nu=1,\ldots,N}{\mathrm{max}}\sigma_{\nu}$ and observing that  $\left(1 - 2\gamma_{\nu} \widehat{\sigma} + \widehat{\sigma}^2 \widehat{L}_{\nu}^2 \right)  \leq \left(1 - 2\gamma_{\mathrm{min}} \widehat{\sigma} + \widehat{\sigma}^2 \widehat{L}_{\mathrm{max}}^2 \right)$,  where $\gamma_{\mathrm{min}}=\underset{\nu=1,\ldots,N}{\mathrm{min}}\gamma_{\nu}$ and  $\widehat{L}_{\mathrm{max}}=\underset{\nu=1,\ldots,N}{\mathrm{max}}\widehat{L}_{\nu}$, it immediately follows that  
\[
\begin{aligned}
\left\| {u}^{\nu,k,l+1} - \widehat{x}^{\nu,k}  \right\|^2
\leq 
\left(1 - 2\gamma_{\mathrm{min}} \widehat{\sigma} + \widehat{\sigma}^2 \widehat{L}_{\mathrm{max}}^2 \right)
\left\| {u}^{\nu,k,l} - \widehat{x}^{\nu,k}  \right\|^2. 
\end{aligned}
\]
Thus, for  $0<\widehat{\sigma}<\left(2\gamma_{\mathrm{min}}^2\right)/\widehat{L}_{\mathrm{max}}$ implying  that $\left(1 - 2\gamma_{\mathrm{min}} \widehat{\sigma} + \widehat{\sigma}^2 \widehat{L}_{\mathrm{max}}^2 \right)<1$, we  obtain 
\begin{equation} \label{eq:prop_inner_eq5}
\begin{aligned}
\left\| {u}^{\nu,k,l+1} - \widehat{x}^{\nu,k}  \right\|
\leq
\boldsymbol{\tau} \left\| {u}^{\nu,k,l} - \widehat{x}^{\nu,k}  \right\|, \quad 0<\boldsymbol{\tau} < 1,
\end{aligned}
\end{equation}
where $\boldsymbol{\tau}= \sqrt{ 1 - 2\gamma_{\mathrm{min}} \widehat{\sigma} + \widehat{\sigma}^2 \widehat{L}_{\mathrm{max}}^2}$.
Therefore, by summing over  the above inequality \eqref{eq:prop_inner_eq5} for all players  from $\nu=1$ to $N$, we deduce the desired result \eqref{eq:prop_converge_result}. 
\vspace{0.07in}

(b) From the property \ref{itm:p1} in Remark \ref{rem_prop_approx_AL} with $\mathbf{y} = \mathbf{x}^{k}$, we know that 
\begin{equation} \label{eq:in_des_1}
\widehat{\mathcal{L}_{\alpha\beta}^{\nu}} \left( \mathbf{x}^{k},z^{\nu,k},\lambda^{\nu,k},\mu^{\nu,k};\mathbf{x}^{k}\right) = \mathcal{L}_{\alpha\beta}^{\nu} \left( \mathbf{x}^{k},z^{\nu,k},\lambda^{\nu,k},\mu^{\nu,k}\right), \ \ \nu=1,\ldots,N.   
\end{equation}
Since $\mathbf{x}^{k} \neq \widehat{\mathbf{x}}^{k}$ and $\mathbf{u}^{k,l} \rightarrow \widehat{\mathbf{x}}^{k}$ by the result (a), the inner gradient projection \eqref{eq:distributed_pd_vi} can find a point ${\mathbf{u}}^{k,l+1}$ close to $\widehat{\mathbf{x}}^k$  such that 
\begin{equation} \label{eq:in_des_2}
\widehat{\mathcal{L}_{\alpha\beta}^{\nu}}\left(\mathbf{u}^{k,l+1},z^{\nu,k},\lambda^{\nu,k},\mu^{\nu,k};\mathbf{x}^{k}\right) 
<
\widehat{\mathcal{L}_{\alpha\beta}^{\nu}}\left(\mathbf{x}^{k},z^{\nu,k},\lambda^{\nu,k},\mu^{\nu,k};\mathbf{x}^{k}\right), \ \ \nu=1,\ldots,N,
\end{equation}
in a finite number of iterations. In addition, by the property \ref{itm:p2} in Remark \ref{rem_prop_approx_AL} with $\mathbf{y} = \mathbf{u}^{k,l+1}$, we have that for any $\gamma_\nu \geq L_\nu$,
\begin{equation}  \label{eq:in_des_3}
\mathcal{L}_{\alpha\beta}^{\nu}\left(\mathbf{u}^{k,l+1},z^{\nu,k},\lambda^{\nu,k},\mu^{\nu,k}\right)
\leq
\widehat{\mathcal{L}_{\alpha\beta}^{\nu}}\left(\mathbf{u}^{k,l+1},z^{\nu,k},\lambda^{\nu,k},\mu^{\nu,k};\mathbf{x}^{k}\right), \ \ \nu=1,\ldots,N. 
\end{equation}
Combining \eqref{eq:in_des_1}, \eqref{eq:in_des_2}, and \eqref{eq:in_des_3} yields  
\begin{equation} 
	\mathcal{L}_{\alpha\beta}^{\nu}\left(\mathbf{u}^{k,l+1},z^{\nu,k},\lambda^{\nu,k},\mu^{\nu,k}\right)
	<
	\mathcal{L}_{\alpha\beta}^{\nu}\left(\mathbf{x}^{k},z^{\nu,k},\lambda^{\nu,k},\mu^{\nu,k}\right), \ \ \nu=1,\ldots,N. \notag 
\end{equation}
Hence, we can derive the desired decrease property of every  ${\mathcal{L}_{\alpha\beta}^{\nu}}$ during inner iterations.
\end{proof}

\section{Proofs of Key Results in Section \ref{sec4.1}}

\subsection{Proof of Lemma \ref{lem_bound_multi}} \label{C.1}
\proof
Note that since $\mathcal{L}_{\alpha\beta}^{\nu}$ is strongly concave in $\lambda^{\nu}$ for fixed $(\mathbf{x}, z^{\nu}, \mu^{\nu})$, there exists a unique maximizer, denoted by $\widehat{\lambda}^{\nu}\left(\mathbf{x},z^{\nu},\mu^{\nu}\right)$,  such that $$\mathcal{L}_{\alpha\beta}^{\nu}(\mathbf{x},z^{\nu},\widehat{\lambda}^{\nu}(\mathbf{x},z^{\nu},\mu^{\nu}),\mu^{\nu})
={\mathrm{max}}_{\lambda^{\nu}\in\mathbb{R}^{m_{\nu}}_{+}}\mathcal{L}_{\alpha\beta}^{\nu}\left(\mathbf{x},z^{\nu},\lambda^{\nu},\mu^{\nu}\right).$$ 
From the update of $\lambda^{\nu,k}$ defined (as maximizer) in Step \ref{step3}, we have
\[\begin{aligned}\label{eq:multi_x_relation_1}
&
\nabla_{\lambda^{\nu}}\mathcal{L}_{\alpha\beta}^{\nu}\left(\mathbf{x}^{k+1},z^{\nu,k+1},\lambda^{\nu,k+1},\mu^{\nu,k}\right)^{T}\left(\lambda^{\nu,k}-\lambda^{\nu,k+1}\right)\leq0,\\ 
&
\nabla_{\lambda^{\nu}}\mathcal{L}_{\alpha\beta}^{\nu}\left(\mathbf{x}^{k},z^{\nu,k},\lambda^{\nu,k},\mu^{\nu,k-1}\right)^{T}\left(\lambda^{\nu,k+1}-\lambda^{\nu,k}\right)\leq0.
\end{aligned}
\]
By the definition $z^{\nu,k+1}=(\lambda^{\nu,k}-\mu^{\nu,k})/\alpha_{\nu}$ in Step \ref{step2} with $\lambda^{\nu,k+1}=\mu^{\nu,k+1}$ in Step \ref{step3}, we have   $z^{\nu,k+1}=z^{\nu,k}=0$ for the same starting point $\lambda^{\nu,0}=\mu^{\nu,0}$. Adding the above inequalities and a direct computation of $\nabla_{\lambda^{\nu}}\mathcal{L}_{\alpha\beta}^{\nu}$ give
\begin{equation} \label{eq:bound_multi_1}
\begin{aligned} 
&
\left(\nabla_{\lambda^{\nu}}\mathcal{L}_{\alpha\beta}^{\nu}\left(\mathbf{x}^{k+1},z^{\nu,k+1},\lambda^{\nu,k+1},\mu^{\nu,k}\right)
-\nabla_{\lambda^{\nu}}\mathcal{L}_{\alpha\beta}^{\nu}\left(\mathbf{x}^{k},z^{\nu,k},\lambda^{\nu,k},\mu^{\nu,k-1}\right)\right)^{T}\left(\lambda^{\nu,k}-\lambda^{\nu,k+1}\right) \\
&
=\left(g^{\nu}\left(\mathbf{x}^{k+1}\right) -g^{\nu}\left(\mathbf{x}^{k}\right) 
-\beta_{\nu}\left(\lambda^{\nu,k+1}-\lambda^{\nu,k}\right)
+\beta_{\nu}\left(\mu^{\nu,k}-\mu^{\nu,k-1}\right)\right)^{T}
\left(\lambda^{\nu,k}-\lambda^{\nu,k+1}\right) \\
&
=\left(g^{\nu}\left(\mathbf{x}^{k+1}\right)-g^{\nu}\left(\mathbf{x}^{k}\right)\right)^{T}\left(\lambda^{\nu,k}-\lambda^{\nu,k+1}\right)
+ \beta_{\nu}\left\Vert\lambda^{\nu,k+1}-\lambda^{\nu,k} \right\Vert ^{2}  +\beta_{\nu}\underbrace{\left(\mu^{\nu,k}-\mu^{\nu,k-1}\right)^{T}\left(\lambda^{\nu,k}-\lambda^{\nu,k+1}\right)}_{ \geq \left\Vert \lambda^{\nu,k+1}-\lambda^{\nu,k} \right\Vert^2 } 
\leq0,
\end{aligned}
\end{equation}
where, to bound the third term, we used Lemma 1(a) in \citet{nedic2010constrained}; $(x-y )^{T}(x-P[x] )\geq \left\|P[x]-x \right\|^{2}$ with $x=\mu^{\nu,k}$, $y=\mu^{\nu,k-1}$, $P[x]=\lambda^{\nu,k+1}$, and the fact  $\mu^{\nu,k}=\lambda^{\nu,k}$.  Specifically, since  $\lambda^{\nu,k+1}$ maximizes $\mathcal{L}_{\alpha\beta}^{\nu}(\mathbf{x}^{k+1},z^{\nu,k+1},\lambda^{\nu},\mu^{\nu,k}) = \theta_{\nu}(\mathbf{x}^{k+1})+(\lambda^{\nu})^{T}g^{\nu}(\mathbf{x}^{k+1})-\frac{\beta_{\nu}}{2}\left\Vert \lambda^{\nu}-\mu^{\nu,k}\right\Vert ^{2}$, we have  
\begin{multline*}
\theta_{\nu}(\mathbf{x}^{k+1})+\left(\lambda^{\nu,k+1}\right) ^{T}g^{\nu}(\mathbf{x}^{k+1})-\frac{\beta_{\nu}}{2}\left\Vert \lambda^{\nu,k+1}-\mu^{\nu,k}\right\Vert ^{2}  \\ 
\geq
\theta_{\nu}(\mathbf{x}^{k+1})+\left( \widehat{\lambda}^{\nu}(\mathbf{x}^{k+1}, z^{k+1}, \mu^{\nu,k})\right) ^{T}g^{\nu}(\mathbf{x}^{k+1})-\frac{\beta_{\nu}}{2}\left\Vert \widehat{\lambda}^{\nu}(\mathbf{x}^{k+1}, z^{k+1}, \mu^{\nu,k})-\mu^{\nu,k}\right\Vert ^{2},
\end{multline*} 
and $\left(\lambda^{\nu,k+1}\right) ^{T}g^{\nu}(\mathbf{x}^{k+1})=\widehat{\lambda}^{\nu}\left(\mathbf{x}^{k+1}, z^{k+1}, \mu^{\nu,k}\right)^{T}g^{\nu}(\mathbf{x}^{k+1})$. It thus follows that 
\[
\left\Vert \lambda^{\nu,k+1}-\mu^{\nu,k}\right\Vert 
\leq 
\left\Vert \widehat{\lambda}^{\nu}(\mathbf{x}^{k+1}, z^{k+1}, \mu^{\nu,k})-\mu^{\nu,k}\right\Vert,
\]
which, by definition of projection {\cite[Section 3.4] {bertsekas1989parallel}}, means that $\lambda^{\nu,k+1}$ can be viewed as the projection of $\mu^{k}$ onto  the solution set $\widehat{\lambda}^{\nu}(\mathbf{x}^{k+1}, z^{k+1}, \mu^{\nu,k})$. We thus see that 
\[
\left(\mu^{\nu,k} -\mu^{\nu,k-1}\right)^T
\left(\lambda^{\nu,k}-\lambda^{\nu,k+1}\right) =\left(\lambda^{\nu,k} -\lambda^{\nu,k-1}\right)^T
\left(\lambda^{\nu,k}-\lambda^{\nu,k+1}\right) \geq  \left\Vert \lambda^{\nu,k+1}-\lambda^{\nu,k} \right\Vert^2 \geq 0.
\]
Note that using the Cauchy-Schwarz inequality, we also get  that $\left\Vert \lambda^{\nu,k} -\lambda^{\nu,k-1}\right\Vert \geq \left\Vert \lambda^{\nu,k+1}-\lambda^{\nu,k} \right\Vert$ implying the stable sequence of the multipliers.  Rearranging terms in \eqref{eq:bound_multi_1}, we obtain  
\[
\beta_{\nu}\left\Vert\lambda^{\nu,k+1} - \lambda^{\nu,k} \right\Vert ^{2} 
\leq\left(g^{\nu}(\mathbf{x}^{k+1}) -g^{\nu}(\mathbf{x}^{k})\right)^{T} \left(\lambda^{\nu,k+1}-\lambda^{\nu,k}\right), 
\]
which leads to 
\[
\begin{aligned}
\left\Vert \lambda^{\nu,k+1}-\lambda^{\nu,k}\right\Vert \overset{(i)}{\leq} \frac{1}{\beta_{\nu}}\left\Vert g^{\nu}(\mathbf{x}^{k+1})-g^{\nu}(\mathbf{x}^{k})\right\Vert 
\overset{(ii)}{\leq}
\frac{L_{g^{\nu}}}{\beta_{\nu}} \left\Vert \mathbf{x}^{k+1} -\mathbf{x}^{k} \right\Vert.
\end{aligned}
\]
where $\left(i\right)$ follows from the Cauchy-Schwarz inequality; $\left(ii\right)$ is from the continuous differentiability of $g^{\nu}(\mathbf{x})$ (Assumption \ref{assumption1}), implying that $g^{\nu}(\mathbf{x})$ is Lipschitz continuous with constant $L_{g^{\nu}}$. Squaring both sides of the inequality gives the desired result \eqref{eq:multiplier_x_relation}. 
\endproof

\subsection{Proof of Lemma \ref{lem_lag_behavior}} \label{C.2}
\proof
Consider the difference of two consecutive sequences of $\mathcal{L}_{\alpha \beta}^{\nu}$:
\begin{equation} \label{eq:lem_lag_behavior_1}
\begin{aligned}
&
\mathcal{L}_{\alpha \beta}^{\nu}\left(\mathbf{x}^{k+1},z^{\nu,k+1},\lambda^{\nu,k+1},\mu^{\nu,k+1}\right) -\mathcal{L}_{\alpha \beta}^{\nu}\left(\mathbf{x}^{k},z^{\nu,k},\lambda^{\nu,k},\mu^{\nu,k}\right) \\
&
= 
\left[  \mathcal{L}_{\alpha \beta}^{\nu}\left(\mathbf{x}^{k+1},z^{\nu,k+1},\lambda^{\nu,k},\mu^{\nu,k}\right) 
-
\mathcal{L}_{\alpha \beta}^{\nu}\left(\mathbf{x}^{k},z^{\nu,k},\lambda^{\nu,k},\mu^{\nu,k}\right) \right] \\
&
\quad +\left[ \mathcal{L}_{\alpha \beta}^{\nu}\left(\mathbf{x}^{k+1},z^{\nu,k+1},\lambda^{\nu,k+1},\mu^{\nu,k+1}\right) 
- \mathcal{L}_{\alpha \beta}^{\nu}\left(\mathbf{x}^{k+1},z^{\nu,k+1},\lambda^{\nu,k},\mu^{\nu,k}\right)\right].
\end{aligned}
\end{equation}

We estimate the two terms on the right-hand side of \eqref{eq:lem_lag_behavior_1} one by one. For the first term,   using the descent lemma (Lemma \ref{lem_descent}), we have 
\begin{equation} \label{eq:lem_lag_behavior_2}
\mathcal{L}_{\alpha \beta}^{\nu}\left(\mathbf{x}^{k+1}\right) 
\leq
\mathcal{L}_{\alpha \beta}^{\nu}\left(\mathbf{x}^{k}\right)
+\nabla_{\mathbf{x}}\mathcal{L}_{\alpha \beta}^{\nu}\left(\mathbf{x}^{k}\right)^{T}\left(\mathbf{x}^{k+1}-\mathbf{x}^{k}\right)
+\frac{L_\nu}{2}\left\Vert \mathbf{x}^{k+1}-\mathbf{x}^{k}\right\Vert ^{2}.
\end{equation}
Here, we omitted $\left( z^{\nu,k}, \lambda^{\nu,k}, \mu^{\nu,k}\right)$ for simplicity. By Step \ref{step1} of Algorithm 1, we have 
\[
\begin{aligned}
\widehat{\mathcal{L}_{\alpha \beta}^{\nu}}\left(\mathbf{x}^{k+1};\mathbf{x}^{k}\right) &
= \mathcal{L}_{\alpha \beta}^{\nu}\left(\mathbf{x}^{k}\right)
+
\nabla_{\mathbf{x}}\mathcal{L}_{\alpha \beta}^{\nu}\left(\mathbf{x}^{k}\right)^{T}\left(\mathbf{x}^{k+1} - \mathbf{x}^{k}\right)
+\frac{\gamma_{\nu}}{2}\left\Vert \mathbf{x}^{k+1} - \mathbf{x}^{k}\right\Vert ^{2} 
\leq 
\mathcal{L}_{\alpha \beta}^{\nu}\left(\mathbf{x}^{k}\right).
\end{aligned}
\]
Thus we obtain
\[
\nabla_{\mathbf{x}}\mathcal{L}_{\alpha \beta}^{\nu}\left(\mathbf{x}^{k}\right)^{T}\left(\mathbf{x}^{k+1} - \mathbf{x}^{k}\right)
\leq -\frac{\gamma_{\nu}}{2}\left\Vert \mathbf{x}^{k+1} - \mathbf{x}^{k}\right\Vert ^{2}.
\]
By substituting the above expression into \eqref{eq:lem_lag_behavior_2} and using the definition of $z^{\nu,k+1}$ in Step \ref{step2}, we get 
\begin{equation} \label{eq:lem_lag_behavior_3}
\mathcal{L}_{\alpha \beta}^{\nu}\left(\mathbf{x}^{k+1},z^{\nu,k+1},\lambda^{\nu,k},\mu^{\nu,k}\right) -
\mathcal{L}_{\alpha \beta}^{\nu}\left(\mathbf{x}^{k},z^{\nu,k},\lambda^{\nu,k},\mu^{\nu,k}\right) 
\leq 
-\frac{1}{2}\left(\gamma_{\nu}- L_{\nu} \right) \left\Vert \mathbf{x}^{k+1}-\mathbf{x}^{k}\right\Vert ^{2}.
\end{equation}

Next, consider the second term: 
\begin{equation}\label{eq:lem_lag_behavior_4} 
\begin{aligned}
&
\mathcal{L}_{\alpha \beta}^{\nu}\left(\mathbf{x}^{k+1},z^{\nu,k+1},\lambda^{\nu,k+1},\mu^{\nu,k+1}\right) -\mathcal{L}_{\alpha \beta}^{\nu}\left(\mathbf{x}^{k+1},z^{\nu,k+1},\lambda^{\nu,k},\mu^{\nu,k}\right) \\ 
&
=\left( \lambda^{\nu,k+1} -\lambda^{\nu,k} \right)^{T} g^{\nu}\left( \mathbf{x}^{k+1}\right) 
- \frac{\beta_{\nu}}{2} \underbrace{\left\Vert \lambda^{\nu,k+1}-\mu^{\nu,k+1} \right\Vert^{2}}_{=0} +\frac{\beta_\nu}{2} \underbrace{\left\Vert \lambda^{\nu,k}-\mu^{\nu,k} \right\Vert^{2}}_{=0}, 
\end{aligned}\end{equation}
where it follows from Step \ref{step3} that the second and third terms on the right-hand side are zero. 

We now focus on deriving an upper bound for the term $\left( \lambda^{\nu,k+1} -\lambda^{\nu,k} \right)^{T} g^{\nu}\left( \mathbf{x}^{k+1}\right)$. To this end, we need to consider two cases: $\mu^{\nu,k}+\frac{1}{\beta_{\nu}} g^{\nu}(\mathbf{x}^{k+1}) \geq 0$ and $\mu^{\nu,k}+\frac{1}{\beta_{\nu}} g^{\nu}(\mathbf{x}^{k+1}) < 0$. 

\textit{Case 1}.  $\mu^{\nu,k}+\frac{1}{\beta_{\nu}} g^{\nu}(\mathbf{x}^{k+1}) \geq 0$. By the definition of   $\lambda^{\nu,k+1}=\left[\mu^{\nu,k}+\frac{1}{\beta_{\nu}} g^{\nu}(\mathbf{x}^{k+1}) \right]^{+}$ and $\lambda^{\nu,k}=\mu^{\nu,k}$ from Step \ref{step3}, we  obtain 
\begin{equation}\label{eq:lem_lag_behavior_5}
\left( \lambda^{\nu,k+1} -\lambda^{\nu,k} \right)^{T} g^{\nu}( \mathbf{x}^{k+1}) 
=
\beta_{\nu}\left\Vert \lambda^{\nu,k+1}-\lambda^{\nu,k} \right\Vert^{2}.
\end{equation}

\textit{Case 2}.  $\mu^{\nu,k}+\frac{1}{\beta_{\nu}} g^{\nu}(\mathbf{x}^{k+1}) < 0$. Note that in this case, $\lambda^{\nu,k+1}=0$ and $\mathbf{x}^{k+1}$ is feasible because $g(\mathbf{x}^{k+1}) <0$. 
For convenience, we define
\[
\triangle_k
:= (\lambda^{\nu,k+1})^T g^{\nu}(\mathbf{x}^{k+1})
-(\lambda^{\nu,k})^T g^{\nu}(\mathbf{x}^{k}).
\]
Then, by subtracting and adding the term $\left(\lambda^{\nu,k}\right)^T g^{\nu}(\mathbf{x}^{k+1})$ to the right-hand side and using the fact that $\lambda^{\nu,k+1}=0$, it follows that
\begin{align}
\left\| \triangle_k \right\| 
& = \left\| \left(\lambda^{\nu,k+1} -\lambda^{\nu,k}\right)^T g^{\nu}(\mathbf{x}^{k+1})
+ (\lambda^{\nu,k})^T \left(g^{\nu}(\mathbf{x}^{k+1})-g^{\nu}(\mathbf{x}^{k})\right) \right\|  \notag \\
& \geq 
\left\| \left(\lambda^{\nu,k+1} -\lambda^{\nu,k}\right)^T g^{\nu}(\mathbf{x}^{k+1})\right\| 
- 
\left\| \left(\lambda^{\nu,k}-\lambda^{\nu,k+1}\right)^T \left(g^{\nu}(\mathbf{x}^{k+1})-g^{\nu}(\mathbf{x}^{k})\right) \right\|. \label{eq:lem_lag_behavior_6}
\end{align}
From the feasibility of $\mathbf{x}^{k+1}$, we have that for any $\lambda^{\nu}\in\mathbb{R}_{+}^{m_{\nu}}$, $\mu^{\nu}\in\mathbb{R}^{m_{\nu}}$ and $\beta_\nu >0$
\[
\left(\lambda^{\nu}\right)^{T}g^{\nu}(\mathbf{x}^{k+1})
-\frac{\beta_{\nu}}{2} \left\Vert \lambda^{\nu}-\mu^{\nu}\right\Vert ^{2}\leq0,
\]
Thus, we can get with $\lambda^{\nu}=\lambda^{\nu,k+1}$ and $\mu^{\nu}=\lambda^{\nu,k}$ that 
\[
\left(\lambda^{\nu,k+1}\right)^{T}g^{\nu}(\mathbf{x}^{k+1}) \leq 
\frac{\beta_{\nu}}{2}\left\Vert \lambda^{\nu,k+1}-\lambda^{\nu,k}\right\Vert ^{2}.
\]
On the other hand, since  $\lambda^{\nu,k} \geq 0$ maximizes $\mathcal{L}_{\alpha\beta}^{\nu}(\mathbf{x}^{k},z^{\nu,k},\lambda^{\nu},\mu^{\nu,k-1})=\theta_{\nu}\left(\mathbf{x}^{k}\right)+\left(\lambda^{\nu}\right)^{T} g^{\nu}\left(\mathbf{x}^{k}\right)-\frac{\beta_{\nu}}{2}\left\Vert \lambda^{\nu}-\mu^{\nu,k-1}\right\Vert ^{2}$ for given $\left(\mathbf{x}^k,\mu^{\nu,k-1}\right) $ 
and  the third term,
\[-\frac{\beta_{\nu}}{2}\left\Vert\lambda^{\nu,k}-\mu^{\nu,k-1}\right\Vert^2 = 
\begin{cases}
-\frac{1}{2\beta_{\nu}}\left\Vert g^{\nu}(\mathbf{x}^k)\right\Vert^2 & \mathrm{if} \ \ \mu^{\nu,k-1}+\frac{1}{\beta_{\nu}} g^{\nu}(\mathbf{x}^{k}) \geq 0 \\
-\frac{\beta_{\nu}}{2}\left\Vert \mu^{\nu,k-1}\right\Vert^2 & \mathrm{otherwise},
\end{cases}\]
is a  given constant, we have that
\[
\left(\lambda^{\nu,k}\right)^{T}g^{\nu}\left(\mathbf{x}^{k}\right) \geq 0.
\]
Hence,
\begin{equation} \label{eq:lem_lag_behavior_7}
\left\| \triangle_k \right\| 
=\left\| \left(\lambda^{\nu,k+1}\right)^{T}g^{\nu}\left(\mathbf{x}^{k+1}\right) - \left(\lambda^{\nu,k}\right)^{T}g^{\nu}\left(\mathbf{x}^{k}\right) \right\| 
\leq 
\frac{\beta_{\nu}}{2} \left\Vert \lambda^{\nu,k+1}-\lambda^{\nu,k}\right\Vert ^{2},
\end{equation}
Combining \eqref{eq:lem_lag_behavior_6} and \eqref{eq:lem_lag_behavior_7} and invoking Lemma \ref{lem_bound_multi}, we obtain 
\begin{equation} \label{eq:lem_lag_behavior_8}
\begin{aligned}
\left\| \left(\lambda^{\nu,k+1} -\lambda^{\nu,k}\right)^T g^{\nu}(\mathbf{x}^{k+1})\right\| 
& 
\leq 
\left\| \left(\lambda^{\nu,k}-\lambda^{\nu,k+1}\right)^T \left(g^{\nu}(\mathbf{x}^{k+1})-g^{\nu}(\mathbf{x}^{k})\right) \right\| 
+ \frac{\beta_{\nu}}{2} \left\Vert \lambda^{\nu,k+1}-\lambda^{\nu,k}\right\Vert ^{2}\\
&
\leq
L_g\left\| \lambda^{\nu,k+1}-\lambda^{\nu,k}\right\| 
\left\| \mathbf{x}^{k+1}-\mathbf{x}^{k} \right\| +\frac{L_{g^{\nu}}^2}{2\beta_\nu} \left\| \mathbf{x}^{k+1}-\mathbf{x}^{k} \right\| ^2 \\
&
\leq \frac{3L_{g^{\nu}}^2}{2\beta_\nu} \left\| \mathbf{x}^{k+1}-\mathbf{x}^{k} \right\| ^2 .
\end{aligned}
\end{equation}
Notice that the above upper bound on $\left\| \left(\lambda^{\nu,k+1} -\lambda^{\nu,k}\right)^T g^{\nu}(\mathbf{x}^{k+1})\right\|$ includes the upper bound in {Case} 1. 
Therefore, by combining   \eqref{eq:lem_lag_behavior_3}, \eqref{eq:lem_lag_behavior_4} and \eqref{eq:lem_lag_behavior_8}, we obtain the desired result: 
\[
\label{eq:lem_sufficient_decrease_Lagrangian_6}
\mathcal{L}_{\alpha\beta}^{\nu}(\mathbf{x}^{k+1},z^{\nu,k+1},\lambda^{\nu,k+1},\mu^{\nu,k+1}) 
\leq \mathcal{L}_{\alpha\beta}^{\nu}(\mathbf{x}^{k},z^{\nu,k},\lambda^{\nu,k},\mu^{\nu,k})
-\frac{1}{2}\left(\gamma_{\nu} -L_{\nu} -\frac{3L_{g}^{2}}{\beta_{\nu}}\right) \left\Vert \mathbf{x}^{k+1}-\mathbf{x}^{k}\right\Vert ^{2},
\]
which implies that the sequence $\left\lbrace \mathcal{L}_{\alpha\beta}^{\nu}\left(\mathbf{x}^{k},z^{\nu,k},\lambda^{\nu,k},\mu^{\nu,k}\right)\right\rbrace$ is monotonically decreasing if $\gamma_{\nu}$ is chosen such that 
$\gamma_{\nu} > L_{\nu}+\frac{3L_{g^{\nu}}^{2}}{\beta_{\nu}}$ with a suitable choice of $\beta_{\nu}>0$.

\vspace{0.05in}
Next, we show that $\{ \mathcal{L}_{\alpha\beta}^{\nu} \} $ is a convergent sequence. we know from Lemma \ref{lem_eq_to_saddle} that a saddle point of $\mathcal{L}_{\alpha \beta}^{\nu}$ exists. Let $\left(\mathbf{x}^{\ast},z^{\nu,\ast},\lambda^{\nu,\ast},\mu^{\nu,\ast}\right) $ be a saddle point of $\mathcal{L}_{\alpha\beta}^{\nu}(\mathbf{x},z^{\nu},\lambda^{\nu},\mu^{\nu})$. By the updating rules for $\left( \lambda^{\nu,k+1}, \mu^{\nu,k+1}\right)$ defined as maximizers for updated $\left( \mathbf{x}^{k+1}, z^{k+1} \right) $ and the saddle point condition \eqref{eq:saddle}, we see that  
\begin{equation} \label{eq:thm2_x_a}
\mathcal{L}_{\alpha\beta}^{\nu}\left(\mathbf{x}^{k+1},z^{\nu,k+1},\lambda^{\nu,k+1},\mu^{\nu,k+1}\right) 
\geq \mathcal{L}_{\alpha\beta}^{\nu}\left(\mathbf{x}^{k+1},z^{\nu,k+1},\lambda^{\nu,\ast}\mu^{\nu,\ast}\right) 
\geq
\mathcal{L}_{\alpha\beta}^{\nu}\left(\mathbf{x}^{\ast},z^{\nu,\ast},\lambda^{\nu,\ast},\mu^{\nu,\ast}\right)>-\infty,
\end{equation}
which implies that the sequence $\left\lbrace \mathcal{L}_{\alpha\beta}^{\nu}\left(\mathbf{x}^{k},z^{\nu,k},\lambda^{\nu,k},\mu^{\nu,k}\right) \right\rbrace$ is lower bounded by a finite value of $\mathcal{L}_{\alpha\beta}^{\nu}\left(\mathbf{x}^{\ast},z^{\nu,\ast},\lambda^{\nu,\ast},\mu^{\nu,\ast}\right)$. Therefore, with the choice of  $\gamma_{\nu}$  such that 
$\gamma_{\nu}\geq L_{\nu}+\frac{3L_{g^{\nu}}^{2}}{\beta_{\nu}}$, the sequence $\left\lbrace \mathcal{L}_{\alpha\beta}^{\nu}\left(\mathbf{x}^{k},z^{\nu,k},\lambda^{\nu,k},\mu^{\nu,k}\right) \right\rbrace$ converges to a finite limit, denoted by $\underline{\mathcal{L}^{\nu}}$, as $k \rightarrow \infty$.
\endproof

\subsection{Proof of Theorem \ref{thm2_x}} \label{C.3}
\proof 
(a)
Recall from Lemma \ref{lem_eq_to_saddle} that a saddle point $(\mathbf{x}^{\ast},z^{\nu,\ast},\lambda^{\nu,\ast},\mu^{\nu,\ast})$ of  $\mathcal{L}_{\alpha\beta}^{\nu}(\mathbf{x},z^{\nu},\lambda^{\nu},\mu^{\nu})$ exists. We know from \eqref{eq:thm2_x_a} that $\mathcal{L}_{\alpha\beta}^{\nu}(\mathbf{x}^{k},z^{\nu,k},\lambda^{\nu,k},\mu^{\nu,k})$ is lower bounded by $\mathcal{L}_{\alpha\beta}^{\nu}\left(\mathbf{x}^{\ast},z^{\nu,\ast},\lambda^{\nu,\ast},\mu^{\nu,\ast}\right)$. In addition, since $\mathcal{L}_{\alpha\beta}^{\nu}\left(\mathbf{x}^{k},z^{\nu,k},\lambda^{\nu,k},\mu^{\nu,k}\right)$ is nonincreasing, it is also upper bounded by a finite value, i.e,  $\mathcal{L}_{\alpha\beta}^{\nu}(\mathbf{x}^{k},z^{\nu,k},\lambda^{\nu,k},\mu^{\nu,k}) < \infty$. We thus have
\[
\begin{aligned}
-\infty 
< \mathcal{L}_{\alpha\beta}^{\nu}(\mathbf{x}^{k+1},z^{\nu,k+1},\lambda^{\nu,k+1},\mu^{\nu,k+1})
= \theta(\mathbf{x}^{k+1}) + (\lambda^{\nu,k+1})^T g(\mathbf{x}^{k+1}) 
< + \infty.
\end{aligned}
\]
Hence, the sequence $\left\lbrace \mathbf{x}^{k} \right\rbrace$ is  bounded due to the coercivity of $\theta_{\nu}(\mathbf{x})$  (Assumption \ref{assumption_coercive}) with the facts that $\lambda^{\nu,k+1} \geq 0$ and $\left( \lambda^{\nu,k+1}\right) ^T g^{\nu}(\mathbf{x}^{k+1}) \geq 0$.
\vspace{0.1in}

(b)
Note that since the function $\lambda^{\nu} \rightarrow \mathcal{L}_{\alpha\beta}(\mathbf{x}^{},z^{\nu},\lambda^{\nu},\mu^{\nu})$ is strongly concave,  there exists parameter $c_\nu^\lambda>0$ such that for $\lambda^{\nu,\ast},\lambda^{\nu,k+1} \in \mathbb{R}^{m_\nu}_{+}$ and for given $(\mathbf{x}^{k+1},z^{\nu,k+1}, \mu^{\nu,k})$
\[
\begin{aligned}
	\mathcal{L}_{\alpha\beta}^{\nu}(\lambda^{\nu,\ast},\mu^{\nu,k}) 
	&
	\leq 
	\mathcal{L}_{\alpha\beta}^{\nu}(\lambda^{\nu,k+1},\mu^{\nu,k}) 
	+\nabla_{\lambda^{\nu}}\mathcal{L}_{\alpha\beta}^{\nu}(\lambda^{\nu,k+1},\mu^{\nu,k})^{T}(\lambda^{\nu,\ast}-\lambda^{\nu,k+1}) 
	-\frac{c_{\nu}^{\lambda}}{2}\left\Vert  \lambda^{\nu,k+1}-\lambda^{\nu,\ast}\right\Vert ^{2}\\
	\mathcal{L}_{\alpha\beta}^{\nu}(\lambda^{\nu,k+1},\mu^{\nu,k}) 
	& \leq  
	\mathcal{L}_{\alpha\beta}^{\nu}(\lambda^{\nu,\ast},\mu^{\nu,k}) 
	+\nabla_{\lambda^{\nu}}\mathcal{L}_{\alpha\beta}^{\nu}(\lambda^{\nu,\ast},\mu^{\nu,k})^{T}(\lambda^{\nu,k+1}-\lambda^{\nu,\ast}) 
	-\frac{c_{\nu}^{\lambda}}{2}\left\Vert  \lambda^{\nu,k+1}-\lambda^{\nu,\ast}\right\Vert ^{2},
	\end{aligned}
\]
where we omitted $(\mathbf{x}^{k+1},z^{\nu,k+1})$ for notational simplicity. Adding the above two inequalities yields 
\[\label{eq:thm_bound_multi_eq1} 
	c_{\nu}^{\lambda}\left\Vert \lambda^{\nu,k+1}-\lambda^{\nu,\ast}\right\Vert^{2}
	\leq
	\left( \nabla_{\lambda^{\nu}}\mathcal{L}_{\alpha\beta}^{\nu}(\lambda^{\nu,k+1},\mu^{\nu,k}) 
	-\nabla_{\lambda^{\nu}}\mathcal{L}_{\alpha\beta}^{\nu}(\lambda^{\nu,\ast},\mu^{\nu,k})\right) ^{T}\left(\lambda^{\nu,\ast}-\lambda^{\nu,k+1}\right).
\]
Using the Cauchy-Schwarz inequality and the triangle inequality, we obtain
\[\label{eq:thm_bound_multi_eq2}
	\begin{aligned}
	\left\Vert \lambda^{\nu,k+1}-\lambda^{\nu,\ast}\right\Vert 
	& \leq
	\frac{1}{c^{\lambda}_{\nu}}
	\left\|\nabla_{\lambda^{\nu}}\mathcal{L}_{\alpha\beta}^{\nu}\left(\lambda^{\nu,k+1},\mu^{\nu,k}\right) 
	-\nabla_{\lambda^{\nu}}\mathcal{L}_{\alpha\beta}^{\nu}(\lambda^{\nu,\ast},\mu^{\nu,k}) \right\| 
	=
	\frac{1}{c^{\lambda}_{\nu}}
	\left\Vert \beta_\nu \left( \lambda^{\nu,\ast} - \lambda^{\nu,k+1}\right)\right\Vert \\
	& \overset{(a)}{\leq} 
	\frac{1}{c^{\lambda}_{\nu}} \left\Vert 
	\beta_\nu \left(\lambda^{\nu,\ast} -\mu^{\nu,k}\right) - g^\nu (\mathbf{x}^{k+1})\right\Vert \\
	& \leq 
	\frac{\beta_\nu}{c^{\lambda}_{\nu}}  \left\Vert   \lambda^{\nu,k} - \lambda^{\nu,\ast} \right\Vert
	+\frac{1}{c^{\lambda}_{\nu}} \left\Vert g^\nu (\mathbf{x}^{k+1})\right\Vert.
\end{aligned}
\]
where the inequality $(a)$ comes from the definition of $\lambda^{\nu,k+1}$ in Step \ref{step3}, implying that $\lambda^{\nu,k+1} \geq \mu^{\nu,k} +\frac{1}{\beta_\nu}g(\mathbf{x}^{k+1})$. Since $\{\mathbf{x}^{k}\}$ is bounded and $g^{\nu}(\mathbf{x})$ is continuous differentiable (Assumption \ref{assumption1}), there exists $D_{\nu}>0$ such that  $\left\Vert g^{\nu}\left(\mathbf{x}^{k+1}\right)\right\Vert \leq D_{\nu}$. From the update of  $\mu^{\nu,k+1}=\lambda^{\nu,k+1}$ in Step \ref{step3}, we have  $\lambda^k=\mu^k$ for any $k \geq 1$. By taking  $c^{\lambda}_{\nu} =\beta_{\nu}$ we have
\[\label{eq:thm_bound_multi}
\begin{aligned}
	\left\Vert \lambda^{\nu,k+1}-\lambda^{\nu,\ast}\right\Vert 
	& \leq
	\left\Vert  \lambda^{\nu,k} - \lambda^{\nu,\ast}\right\Vert +\frac{D_{\nu}}{\beta_{\nu}}.
	\end{aligned}
\]
Therefore, the sequence $\left\lbrace \lambda^{\nu,k}\right\rbrace $ is bounded on any subset of $\mathbb{R}_{+}^{m_{\nu}}$.
\vspace{0.1in}

(c) Invoking Lemma \ref{lem_lag_behavior}, we have that for all $k\geq1$
\begin{displaymath}\label{eq:thm2_x_c_1}
\rho_{\nu}
\left\Vert 
\mathbf{x}^{k+1}-\mathbf{x}^{k} \right\Vert ^{2}
\leq
\mathcal{L}_{\alpha\beta}^{\nu}\left(\mathbf{x}^{k},z^{\nu,k},\lambda^{\nu,k},\mu^{\nu,k}\right)-\mathcal{L}_{\alpha\beta}^{\nu}\left(\mathbf{x}^{k+1},z^{\nu,k+1},\lambda^{\nu,k+1},\mu^{\nu,k+1}\right),
\end{displaymath}
where $\rho_{\nu}:=\frac{1}{2}\left( \gamma_{\nu}- L_{\nu}-\frac{3L_{g^{\nu}}^{2}}{\beta_{\nu}}\right)\geq0$. Summing the above inequality over $k=1,\ldots,K$, we obtain
\[
\begin{aligned}
\quad\sum_{k=1}^{K}\left\Vert \mathbf{x}^{k+1}-\mathbf{x}^{k}\right\Vert ^{2} 
& 
\leq \frac{1}{\rho_{\nu}} \left( \mathcal{L}_{\alpha\beta}^{\nu}\left(\mathbf{x}^{1},z^{\nu,1},\lambda^{\nu,1},\mu^{\nu,1}\right)-\mathcal{L}_{\alpha\beta}^{\nu}\left(\mathbf{x}^{K+1},z^{\nu,K+1},\lambda^{\nu,K+1},\mu^{\nu,K+1}\right)\right)  \\
&
\leq \frac{1}{\rho_{\nu}}  \left( \mathcal{L}_{\alpha\beta}^{\nu}\left(\mathbf{x}^{1},z^{\nu,1},\lambda^{\nu,1},\mu^{\nu,1}\right) -\theta_{\nu}\left(\mathbf{x}^{\ast}\right)\right) ,
\end{aligned}
\]
where the last inequality comes from \eqref{eq:thm2_x_a} with the fact that   $\mathcal{L}_{\alpha\beta}^{\nu}(\mathbf{x}^{\ast},z^{\nu,\ast},\lambda^{\nu,\ast},\mu^{\nu,\ast})=\theta_{\nu}\left(\mathbf{x}^{\ast}\right)$. Letting $K\rightarrow\infty$ yields
\begin{displaymath}
\sum_{k=1}^{\infty}
\left\Vert \mathbf{x}^{k+1}-\mathbf{x}^{k} \right\Vert ^{2}<\infty,
\end{displaymath}
from which, along with Lemma \ref{lem_bound_multi}, it also  follows immediately that  $\sum_{k=1}^{\infty}\left\Vert \lambda^{\nu,k+1}-\lambda^{\nu,k}\right\Vert ^2< \infty$ and  
$\sum_{k=1}^{\infty}\left\Vert \mu^{\nu,k+1}-\mu^{\nu,k}\right\Vert ^2< \infty$. Therefore,  we can deduce the desired results in \eqref{eq:thm2_x_c_result}.
\endproof

\section{Proofs of Main Convergence Results in Section \ref{sec4.2}}
\subsection{Proof of Theorem \ref{thm_limit_saddle}} \label{D.1}
\proof
By Theorem \ref{thm2_x}, the sequence $\left\lbrace \left(\mathbf{x}^{k}, z^{\nu,k}, \lambda^{\nu,k}, \mu^{\nu,k}\right) \right\rbrace$ is bounded, so there exists at least one limit point. Let $(\overline{\mathbf{x}}, \overline{z},\overline{\lambda},\overline{\mu})$ be a limit point of $\left\lbrace (\mathbf{x}^{k}, z^{\nu,k}, \lambda^{\nu,k}, \mu^{\nu,k}) \right\rbrace$, and let $\left\lbrace \left(\mathbf{x}^{k_j}, z^{\nu,k_j}, \lambda^{\nu,k_j}, \mu^{\nu,k_j}\right) \right\rbrace$ be a subsequence converging to $(\overline{\mathbf{x}}, \overline{z},\overline{\lambda},\overline{\mu})$ as $j \rightarrow \infty$. 
From Theorem \ref{thm2_x}\ref{itm:thm2_c},
it also follows that  $\left\lbrace \left(\mathbf{x}^{k_j+1}, z^{\nu,k_j+1}, \lambda^{\nu,k_j+1}, \mu^{\nu,k_j+1}\right) \right\rbrace \rightarrow (\overline{\mathbf{x}}, \overline{z},\overline{\lambda},\overline{\mu})$  as $j \rightarrow \infty$.

First, we show that a limit point $(\overline{\mathbf{x}},\overline{z}^{\nu},\overline{\lambda}^{\nu},\overline{\mu}^{\nu})$ satisfies the second inequality of the saddle point condition \eqref{eq:saddle}. Because  $\left\lbrace x^{\nu,k_j+1} \right\rbrace \rightarrow \overline{x}^{\nu}$ and $\left\lbrace x^{\nu,k_j} \right\rbrace \rightarrow \overline{x}^{\nu}$ as $j \rightarrow \infty$, we have from Step \ref{step1} that
\[
\overline{x}^{\nu}=\mathcal{P}_{\mathcal{X}_{\nu}}\left[ \overline{x}^{\nu}-\sigma_{\nu} \nabla_{x^{\nu}} \widehat{\mathcal{L}_{\alpha\beta}^{\nu}}(\overline{\mathbf{x}},\overline{z}^{\nu},\overline{\lambda}^{\nu},\overline{\mu}^{\nu};\overline{\mathbf{x}}) \right]. 
\]	
The limit point $\overline{x}^{\nu}$ is equivalent to a solution of the VI \cite[Prop. 1.5.8]{facchinei2007finite}:
\[
\nabla_{x^{\nu}} \widehat{\mathcal{L}_{\alpha\beta}^{\nu}} (\overline{\mathbf{x}},\overline{z}^{\nu},\overline{\lambda}^{\nu},\overline{\mu}^{\nu}; \overline{\mathbf{x}})^{T}\left({x}^{\nu} - \overline{x}^{\nu}\right) \geq 0, \quad \forall {x}^{\nu} \in \mathcal{X}_{\nu}.
\]
Using the fact that   $\nabla_{x^{\nu}}\widehat{\mathcal{L}_{\alpha\beta}^{\nu}} ( \overline{\mathbf{x}},\overline{z}^{\nu},\overline{\lambda}^{\nu},\overline{\mu}^{\nu};\overline{\mathbf{x}})
=\nabla_{x^{\nu}}\mathcal{L}_{\alpha\beta}^{\nu} (\overline{\mathbf{x}},\overline{z}^{\nu},\overline{\lambda}^{\nu},\overline{\mu}^{\nu})$ and the convexity of $\mathcal{L}_{\alpha\beta}^{\nu}$ with respect to $x^{\nu}$,  we obtain the first-order optimality condition for  $\mathcal{L}_{\alpha\beta}^{\nu}$:
\[
\nabla_{x^{\nu}} \mathcal{L}_{\alpha\beta}^{\nu}(\overline{\mathbf{x}},\overline{z}^{\nu},\overline{\lambda}^{\nu},\overline{\mu}^{\nu})^{T}\left({x}^{\nu} - \overline{x}^{\nu}\right) \geq 0, \quad \forall {x}^{\nu} \in \mathcal{X}_{\nu}.
\]
Equivalently,
\[
{\mathcal{L}_{\alpha\beta}^{\nu}} ( \overline{x}^{\nu},\overline{x}^{-\nu},\overline{z}^{\nu},\overline{\lambda}^{\nu},\overline{\mu}^{\nu})
\leq
{\mathcal{L}_{\alpha\beta}^{\nu}} ( x^{\nu},\overline{x}^{-\nu},\overline{z}^{\nu},\overline{\lambda}^{\nu},\overline{\mu}^{\nu}),
\]
which implies that   $(\overline{\mathbf{x}},\overline{z}^{\nu},\overline{\lambda}^{\nu},\overline{\mu}^{\nu})$ satisfies the second inequality of the saddle point condition  \eqref{eq:saddle}. 

Similarly, by the definitions of $\lambda^{\nu,k+1}$ and $\mu^{\nu,k+1}$ (as  maximizers) in Step \ref{step3}, the limit points $(\overline{\lambda}^{\nu},\overline{\mu}^{\nu})$ maximize  $\mathcal{L}_{\alpha\beta}^{\nu}(\overline{\mathbf{x}},{z}^{\nu}({\lambda}^{\nu},{\mu}^{\nu}),{\lambda}^{\nu},{\mu}^{\nu})$. We thus see that
\begin{align}
\nabla_{\lambda^{\nu}} \mathcal{L}_{\alpha\beta}^{\nu}(\overline{\mathbf{x}},\overline{z}^{\nu},\overline{\lambda}^{\nu},\overline{\mu}^{\nu})^{T}({\lambda}^{\nu} - \overline{\lambda}^{\nu}) &\leq 0, \quad \forall {\lambda}^{\nu} \in \mathbb{R}^{m_\nu}_+,  \notag \\
\nabla_{\mu^{\nu}} \mathcal{L}_{\alpha\beta}^{\nu}(\overline{\mathbf{x}},\overline{z}^{\nu},\overline{\lambda}^{\nu},\overline{\mu}^{\nu})^{T}({\mu}^{\nu} - \overline{\mu}^{\nu}) &\leq 0, \quad \forall {\mu}^{\nu} \in \mathbb{R}^{m_\nu}. \notag
\end{align}
Consequently, $(\overline{\mathbf{x}},\overline{z}^{\nu},\overline{\lambda}^{\nu},\overline{\mu}^{\nu})$ satisfies the first inequality of the saddle point condition  \eqref{eq:saddle}.
\endproof

Before proceeding with global convergence, let us provide some preliminaries, which are central to our global convergence analysis. 
\renewcommand\thetheorem{5}
\begin{lemma}[Uniformized K\L{} Property {\cite[Lemma 6]{bolte2014proximal}}] \label{lem_uniformized_KL_property}
	Let $\Omega$ be a compact set and let $\Psi: \mathbb{R}^n \rightarrow (-\infty, \infty]$ be proper, lower semicontinuous function. Assume that $\Psi$ is constant on $\Omega$ and satisfies the K\L{} property at each point of $\Omega$.  Then there exist $\varepsilon>0$, $\delta$ and $\varphi \in \Phi_\delta$  such that for all $\overline{u}$ in $\Omega$ and all $u$ in the following intersection:
	\begin{equation} \label{eq:lem_uniform_KL_1}
	\left\{ u \in \mathbb{R}^n: \mathrm{dist}(u, \Omega) < \varepsilon \right\} \cap \left[ \Psi(\overline{u}) < \Psi(u) < \Psi(\overline{u}) + \delta\right]
	\end{equation}
	one has,
	\begin{equation} \label{eq:lem_uniform_KL_2}
	\varphi^{\prime}(\Psi(u) -\Psi(\overline{u})) \cdot \mathrm{dist}(0,\partial \Psi(u)) \geq 1.
	\end{equation}
\end{lemma}
Note that if the function $\Psi$ is continuously differentiable and $\Psi(\overline{u})=0$, the inequality \eqref{eq:lem_uniform_KL_2} can be rewritten as 
\begin{equation} \label{eq:uniform_KL_setting}
\varphi^{\prime}(\Psi(u) -\Psi(\overline{u})) \left\|\nabla \Psi(u) \right\| \geq 1. \notag
\end{equation}
With the uniformized K\L{} property, we can prove that the generated sequence has finite length, and hence the \emph{whole} sequence  converges to a saddle point. 
The techniques developed in \cite{bolte2014proximal}  are extended to our smooth constrained game setting with some modifications.

In order to exploit Lemma \ref{lem_uniformized_KL_property} for proving global convergence, we need to use the size of the gradient of the P-Lagrangian, denoted by $\widetilde{\nabla}\mathcal{L}_{\alpha\beta}$, and derive an upper bound on the gradient. Noting that  $\theta_\nu$ and $g^\nu$ are continuously differentiable,  $x^\nu \in \mathcal{X}_\nu$, and  $\lambda^\nu \in \mathbb{R}^{m_\nu}_+$, we consider the following projected gradients of $\mathcal{L}_{\alpha\beta}^\nu$ in $x^\nu$ and $\lambda^\nu$ for $x^\nu$-component and $\lambda^\nu$-component of $\widetilde{\nabla}\mathcal{L}_{\alpha\beta}$:
\begin{align} \label{eq:project_gradient_x}
\widetilde{\nabla}_{x^\nu}\mathcal{L}_{\alpha\beta}^\nu(\mathbf{x},z^\nu,\lambda^\nu,\mu^\nu)
&:= x^\nu - \mathcal{P}_{\mathcal{X}_\nu}\left[  x^\nu -  \nabla_{x^\nu}{\mathcal{L}}_{\alpha\beta}^\nu(\mathbf{x},z^\nu,\lambda^\nu,\mu^\nu) \right], \notag \\	
\widetilde{\nabla}_{\lambda^\nu}\mathcal{L}_{\alpha\beta}^\nu(\mathbf{x},z^\nu,\lambda^\nu,\mu^\nu)
&:= \lambda^\nu -\left[  \lambda^\nu +   \nabla_{\lambda^\nu}{\mathcal{L}}_{\alpha\beta}^\nu(\mathbf{x},z^\nu,\lambda^\nu,\mu^\nu) \right]^+. \notag
\end{align}
Let us now define the \emph{projected gradient} of $\mathcal{L}_{\alpha\beta}^\nu$ at $(\mathbf{x}^{k+1},z^{\nu,k+1},\lambda^{\nu,k+1},\mu^{\nu,k+1})$ as
\begin{equation} \label{eq:projected_gradient}
\widetilde{\nabla} \mathcal{L}_{\alpha\beta}^\nu(\mathbf{w}^{\nu,k+1}):= 
\left(\begin{aligned}
& q^{k+1}_{x^{\nu}} \\
& q^{k+1}_{z^{\nu}} \\
& q^{k+1}_{\lambda^{\nu}} \\
& q^{k+1}_{\mu^{\nu}}
\end{aligned}\right) 
= \left(\begin{aligned}
& x^{\nu,k+1} - \mathcal{P}_{\mathcal{X}_\nu}\left[  {x}^{\nu,k+1} -  \nabla_{x^\nu}\mathcal{L}_{\alpha \beta}^\nu(\mathbf{x}^{k+1},z^{\nu,k+1},\lambda^{\nu,k+1},\mu^{\nu,k+1}) \right] \\
&\nabla_{z^\nu}\mathcal{L}_{\alpha \beta}^\nu(\mathbf{x}^{k+1},z^{\nu,k+1},\lambda^{\nu,k+1},\mu^{\nu,k+1}) \\
&\lambda^{\nu,k+1} - \left[  \lambda^{\nu,k+1} +  \nabla_{\lambda^\nu}\mathcal{L}_{\alpha \beta}^\nu(\mathbf{x}^{k+1},z^{\nu,k+1},\lambda^{\nu,k+1},\mu^{\nu,k+1}) \right]^+ \\
&\nabla_{\mu^\nu}\mathcal{L}_{\alpha \beta}^\nu(\mathbf{x}^{k+1},z^{\nu,k+1},\lambda^{\nu,k+1},\mu^{\nu,k+1})
\end{aligned}\right). 
\end{equation}

It is clear that if $\widetilde{\nabla} \mathcal{L}_{\alpha\beta}^\nu(\mathbf{w}^{\nu,k+1}) \rightarrow 0$, then a saddle point of $\mathcal{L}_{\alpha\beta}^\nu(\mathbf{w}^{\nu})$ is obtained.
In what follows, we derive an upper bound on  $\widetilde{\nabla} \mathcal{L}_{\alpha\beta}^\nu(\mathbf{w}^{\nu,k+1})$ in terms of the  generated iterates by Algorithm \ref{algorithm1}.

Recall that the conditions \eqref{eq:assumption_lipschitz_1}, \eqref{eq:assumption_lipschitz_2}, \eqref{eq:assumption_lipschitz_3}, and  \eqref{eq:assumption_lipschitz_4} in  Assumption \ref{assumption_lipschitz} imply  there exist constants $M_{\nabla\theta_{\nu}}$ and $M_{\nabla{g^{\nu}}}$ such that 
 \begin{subequations}
		\begin{align}
		\left\Vert \nabla_{{x}^\nu}{\theta_{\nu}\left(\mathbf{x}_{1}\right)}-\nabla_{{x}^\nu}{\theta_{\nu}\left(\mathbf{x}_{2}\right)}\right\Vert 
		&\leq 
		M_{\nabla\theta_{\nu}}\left\Vert \mathbf{x}_{1}-\mathbf{x}_{2}\right\Vert ,\quad\forall\mathbf{x}_{1},\mathbf{x}_{2}\in\mathcal{X}_{\nu} \times \mathcal{X}_{-\nu},  \label{eq:uniform_lipsch3} \\
		\left\Vert \nabla_{{x}^\nu}{g^{\nu}\left(\mathbf{x}_{1}\right)}-\nabla_{{x}^\nu}{g^{\nu}\left(\mathbf{x}_{2}\right)}\right\Vert 
		&\leq M_{\nabla g^{\nu}}\left\Vert \mathbf{x}_{1}-\mathbf{x}_{2}\right\Vert ,\quad\forall\mathbf{x}_{1},\mathbf{x}_{2}\in\mathcal{X}_{\nu} \times \mathcal{X}_{-\nu}, \label{eq:uniform_lipsch4} 
		\end{align}
\end{subequations}
where $M_{\nabla \theta_\nu}= L_{\nu}(\theta_\nu) + L_{-\nu}(\theta_{\nu})$ and  $M_{\nabla g^\nu}= L_{\nu}(g^\nu) + L_{-\nu}(g^\nu)$. 
\renewcommand\thetheorem{6}
\begin{lemma} \label{lem_upper_bound_gradient}
	Let $\left\{\mathbf{w}^{\nu,k} \right\}_{\nu=1}^N$ be the sequence generated by Algorithm \ref{algorithm1}. Then, for every $\nu=1,\ldots,N$,  there exist constant $C_\nu>0$ such that for all $k  \geq 0$
	\begin{equation} \label{eq:lem_upper_bound_gradient}
	\left\| \widetilde{\nabla} \mathcal{L}_{\alpha\beta}^\nu(\mathbf{w}^{\nu,k+1}) \right\|
	\leq
	C_\nu \left\|\mathbf{x}^{k+1} -\mathbf{x}^{k}\right\|. 
	\end{equation}	
\end{lemma}

\proof
We first estimate an upper bound for the norm of component $q^{k+1}_{x^{\nu}}$ in $ \widetilde{\nabla}\mathcal{L}_{\alpha\beta}^\nu(\mathbf{w}^{\nu,k+1})$. Recall that  there exists a unique solution $\widehat{\mathbf{x}}^{k}$ of the approximation subproblem $\mathrm{VI}^{k}(\mathbf{X},\widehat{\mathbf{L}}^{k})$ in \eqref{eq:approx_VI} at each iteration $k$ (Lemma \ref{prop_converge_inner}), and denote by $\widehat{x}^{\nu,k}$ the $\nu$th component of  $\widehat{\mathbf{x}}^{k}$.
From the fixed-point characterization of $\widehat{x}^{\nu,k}$, we know that for every $\nu=1,\ldots,N,$
\begin{align}
\widehat{x}^{\nu,k} 
& = \mathcal{P}_{\mathcal{X}_{\nu}} \left[\widehat{x}^{\nu,k} -  \nabla_{x^{\nu}}\widehat{\mathcal{L}_{\alpha\beta}^{\nu}}\left(\widehat{\mathbf{x}}^{k},z^{\nu,k},\lambda^{\nu,k},\mu^{\nu,k}; \mathbf{x}^{k}\right) \right] \notag \\ 
& = \mathcal{P}_{\mathcal{X}_\nu}\left[\widehat{x}^{\nu,k} - \left( \nabla_{x^{\nu}} \theta_\nu(\mathbf{x}^k) + \nabla_{x^{\nu}} g^\nu(\mathbf{x}^k) \lambda^{\nu,k} + \gamma_\nu (\widehat{x}^{\nu,k} - {x}^{\nu,k}) \right) \right]. \notag
\end{align}
Hence, 
\[
\begin{aligned}
\left\| q^{k+1}_{x^\nu} \right\| 
&= \left\| x^{\nu,k+1} - \widehat{x}^{\nu,k} + \widehat{x}^{\nu,k} - \mathcal{P}_{\mathcal{X}_\nu}\left[ x^{\nu,k+1} -  \nabla_{{x}^\nu}{\mathcal{L}_{\alpha\beta}^\nu}(\mathbf{x}^{k+1},z^{\nu,k+1},\lambda^{\nu,k+1},\mu^{\nu,k+1}) \right] \right\|  \\
&= 
\left\| x^{\nu,k+1} - \widehat{x}^{\nu,k}\right\| 
+ \left\| \mathcal{P}_{\mathcal{X}_\nu}\left[\widehat{x}^{\nu,k} -\nabla_{x^{\nu}} \theta_\nu(\mathbf{x}^{k}) - \nabla_{x^{\nu}} g^\nu(\mathbf{x}^{k}) \lambda^{\nu,k} - \gamma_\nu (\widehat{x}^{\nu,k} - {x}^{\nu,k})\right] \right. \\
& \hspace{2in}\left. 
- \mathcal{P}_{\mathcal{X}_\nu}\left[ x^{\nu,k+1} -\nabla_{x^{\nu}} \theta_\nu(\mathbf{x}^{k+1}) - \nabla_{x^{\nu}} g^\nu(\mathbf{x}^{k+1}) \lambda^{\nu,k+1}\right] \right\| \\
& \overset{(a)}{\leq} 
\left\| x^{\nu,k+1} - \widehat{x}^{\nu,k}\right\| 
+ \left\|\left[  \widehat{x}^{\nu,k} -\nabla_{x^{\nu}} \theta_\nu(\mathbf{x}^{k}) - \nabla_{x^{\nu}} g^\nu(\mathbf{x}^{k}) \lambda^{\nu,k} - \gamma_\nu (\widehat{x}^{\nu,k} - {x}^{\nu,k})\right]  \right. \\
& \hspace{2in}\left. 
- \left[  x^{\nu,k+1}  -\nabla_{x^{\nu}} \theta_\nu(\mathbf{x}^{k+1}) - \nabla_{x^{\nu}} g^\nu(\mathbf{x}^{k+1}) \lambda^{\nu,k+1}\right] \right\| \\
& \overset{(b)}{\leq}
(2+\gamma_\nu)\left\| \mathbf{x}^{k+1}-\mathbf{x}^{k} \right\| 
+ \left\| \nabla_{x^{\nu}} \theta_\nu(\mathbf{x}^{k+1})- \nabla_{x^{\nu}} \theta_\nu(\mathbf{x}^{k}) 
+  \nabla_{x^{\nu}} g^\nu(\mathbf{x}^{k+1}) \lambda^{\nu,k+1}  -\nabla_{x^{\nu}} g^\nu(\mathbf{x}^{k}) \lambda^{\nu,k} \right\|,
\end{aligned}
\]
where $(a)$ follows from the nonexpansive property of the
projection operator, and  $(b)$ is due to  the facts that $\left\| x^{\nu,k+1} - \widehat{x}^{\nu,k}\right\| \leq \left\| x^{\nu,k+1} - {x}^{\nu,k}\right\|$ and $ \left\| x^{\nu,k+1} - {x}^{\nu,k}\right\| \leq \left\|\mathbf{x}^{k+1}-\mathbf{x}^{k} \right\|$. Then,
by adding and subtracting  $g^\nu(\mathbf{x}^{k}) \lambda^{\nu,k+1}$ and using the triangle inequality, we obtain 
\begin{align}
\left\| q^{k+1}_{x^\nu} \right\| 
&\leq 
(2+\gamma_\nu)  \left\| \mathbf{x}^{k+1}-\mathbf{x}^{k}\right\| + \left\| \nabla_{x^{\nu}} \theta_\nu(\mathbf{x}^{k+1})- \nabla_{x^{\nu}} \theta_\nu(\mathbf{x}^{k}) \right\| \notag \\
&
\qquad \quad
+
\left\| \nabla_{x^{\nu}} g^\nu(\mathbf{x}^{k+1}) \lambda^{\nu,k+1}  -\nabla_{x^{\nu}} g^\nu(\mathbf{x}^{k}) \lambda^{\nu,k+1} \right\| 
+ \left\| \nabla_{x^{\nu}} g^\nu(\mathbf{x}^{k}) \lambda^{\nu,k+1} - \nabla_{x^{\nu}} g^\nu(\mathbf{x}^{k}) \lambda^{\nu,k}  \right\| \notag \\
& \leq
(2+\gamma_\nu)  \left\| \mathbf{x}^{k+1}-\mathbf{x}^{k}\right\|
+ M_{\nabla \theta_\nu} \left\| \mathbf{x}^{k+1}-\mathbf{x}^{k}\right\| 
+ M_{\nabla g^\nu} B_{\lambda^\nu} \left\| \mathbf{x}^{k+1}-\mathbf{x}^{k} \right\| 
+ R_{g^\nu} \left\| \lambda^{\nu,k+1} - \lambda^{\nu,k}\right\|    \notag \\
& \leq
\left(2+\gamma_\nu + M_{\nabla \theta_\nu}  +M_{\nabla g^\nu} B_{\lambda^\nu} +\frac{R_{g^\nu} L_{g^\nu}}{\beta_\nu}\right) \left\| \mathbf{x}^{k+1}-\mathbf{x}^{k} \right\|, \label{eq:lem_norm_dx}
\end{align}
where the second inequality follows from the Lipschitz  continuity of $\nabla_{x^{\nu}} \theta_\nu$ and $\nabla_{x^{\nu}} g^\nu$ (see  \eqref{eq:uniform_lipsch3} and \eqref{eq:uniform_lipsch4}), and the boundedness of $\left\lbrace \mathbf{x}^k \right\rbrace $ and $\left\lbrace \lambda^{\nu,k} \right\rbrace $ implying that there exist constants  $B_{\lambda^\nu}:=\mathrm{max}_{k \in \mathbb{N}}\left\|\lambda^{\nu,k} \right\|$ and $R_{g^\nu}:=\mathrm{max}_{k \in \mathbb{N}}\left\|\nabla_{x^\nu}g^\nu(\mathbf{x}^{k}) \right\|$; in the last inequality we used that $\left\Vert \lambda^{\nu,k+1}- \lambda^{\nu,k} \right\Vert \leq \frac{L_{g^\nu}}{\beta_\nu} \left\Vert \mathbf{x}^{k+1}-\mathbf{x}^{k} \right\Vert$ (Lemma \ref{lem_bound_multi}).

Next, notice that from the definition of $\lambda^{\nu,k+1} \in \mathbb{R}^{m_\nu}_+$ as a maximizer for given $(\mathbf{x}^{k+1}, z^{\nu,k+1}, \mu^{\nu,k})$, $\lambda^{\nu,k+1}$ is also characterized by 
\[
\lambda^{\nu,k+1}=\left[  \lambda^{\nu,k+1} +  \nabla_{\lambda^\nu}\mathcal{L}_{\alpha \beta}^\nu(\mathbf{x}^{k+1},z^{\nu,k+1},\lambda^{\nu,k+1},\mu^{\nu,k}) \right]^+,
\]
which, together with the nonexpansive property of the projection onto $ \mathbb{R}^{m_\nu}_+$ and Lemma \ref{lem_bound_multi}, yields
\begin{align}
\left\| q^{k+1}_{\lambda^\nu} \right\| 
&= \left\| \left[ \lambda^{\nu,k+1} + (g^\nu(\mathbf{x}^{k+1}) - z^{\nu,k+1}) -\beta_\nu (\lambda^{\nu,k+1}-\mu^{\nu,k})\right]^+ \right. \notag \\
&\hspace{1.35in} 
\left.  - \left[ \lambda^{\nu,k+1} + (g^\nu(\mathbf{x}^{k+1}) - z^{\nu,k+1}) -\beta_\nu (\lambda^{\nu,k+1}-\mu^{\nu,k+1})\right]^+  \right\| \notag\\
& \leq
\left\| \beta_\nu ( \mu^{\nu,k+1} -\mu^{\nu,k}) \right\|  
\leq {L_{g^\nu}} \left\Vert \mathbf{x}^{k+1}-\mathbf{x}^{k}\right\|.  \label{eq:lem_norm_dlambda} 
\end{align}

In addition, recalling that the definitions of $z^{\nu,k+1}$ in Step \ref{step2} and  $\mu^{\nu,k+1}$ in Step \ref{step3}, we have
\begin{align}
\left\| q^{k+1}_{z^\nu} \right\| 
&= \left\|  (\mu^{\nu,k+1}-\lambda^{\nu,k+1}) + \alpha_\nu z^{\nu,k+1}  \right\|=0, \label{eq:lem_norm_dz} \\
\left\| q^{k+1}_{\lambda^\nu} \right\| 
&= \left\| z^{\nu,k+1} + \beta_\nu (\lambda^{\nu,k+1}-\mu^{\nu,k+1}) \right\|  =0.  \label{eq:lem_norm_dmu}
\end{align}	
Therefore, summing the inequalities \eqref{eq:lem_norm_dx} and  \eqref{eq:lem_norm_dlambda},  we deduce for all $k\geq 0 $ 
\[
\left\| \widetilde{\nabla} \mathcal{L}(\mathbf{w}^{\nu,k+1}) \right\|
=\sum_{i=x^\nu,z^\nu, \lambda^\nu,\mu^\nu} \left\| {q}_i^{\nu,k+1} \right\|
\leq 
C_\nu \left\| \mathbf{x}^{k+1} -\mathbf{x}^{k}\right\|
\]
with the positive constant
\[ 
C_\nu=2+\gamma_\nu + M_{\nabla \theta_\nu}  +M_{\nabla g^\nu} B_{\lambda^\nu} +\frac{R_{g^\nu} L_{g^\nu}}{\beta_\nu} + L_{g^\nu}.  
\]
This completes the proof.
\endproof

\subsection{Proof of Theorem \ref{thm_global_convergence}} \label{D.2}
\proof
Let $\overline{\mathbf{w}}^{\nu}:=(\overline{\mathbf{x}},\overline{z}^{\nu},\overline{\lambda}^{\nu},\overline{\mu}^{\nu})$ be a limit point of the sequence $\left\lbrace \mathbf{w}^{\nu,k}=(\mathbf{x}^{k},z^{\nu,k},\lambda^{\nu,k},\mu^{\nu,k})\right\rbrace $ that is bounded for every $\nu=1,\ldots,N$. Then, by the continuity of $\mathcal{L}_{\alpha\beta}^{\nu}$, we have
\begin{equation} \label{eq:thm_global_convergence_eq1}
\underset{k \rightarrow \infty} {\mathrm{lim}} \mathcal{L}_{\alpha\beta}^{\nu}(\mathbf{w}^{\nu,k})
= 
\mathcal{L}_{\alpha\beta}^{\nu} (\overline{\mathbf{w}}^\nu).
\end{equation}
In the following, we consider two cases: 
\vspace{0.05in}

\textit{Case 1.}
Suppose that there exists an integer $\bar{k}$ such that 
$\mathcal{L}_{\alpha\beta}^{\nu} (\mathbf{w}^{\nu,\bar{k}})= 
\mathcal{L}_{\alpha\beta}^{\nu} (\overline{\mathbf{w}}^\nu)$ for  $\nu=1,\ldots,N$. Since the sequence $\left\lbrace \mathcal{L}_{\alpha\beta}^{\nu}\right\rbrace$ is nonincreasing, we have that  
$\mathcal{L}_{\alpha\beta}^{\nu} (\mathbf{w}^{\nu,k})
= 
\mathcal{L}_{\alpha\beta}^{\nu} (\overline{\mathbf{w}}^\nu)$ for all $k \geq \bar{k}$. 
Then, we have from Lemma \ref{lem_lag_behavior} that for any $t\geq 0$
\begin{equation}  
\rho_\nu \left\Vert \mathbf{x}^{k+t}-\mathbf{x}^{k} \right\Vert ^{2}
\leq
\mathcal{L}_{\alpha\beta}^\nu(\mathbf{w}^{\nu,k}) -\mathcal{L}_{\alpha\beta}^\nu(\mathbf{w}^{\nu,k+t}) =0, \notag
\end{equation}
which leads to 
\begin{equation}
\mathbf{x}^{k+1}-\mathbf{x}^{k} =0, \qquad  \forall  k \geq \bar{k},
\end{equation}
From Lemma \ref{lem_bound_multi}, we also obtain that  $\lambda^{\nu,k+1}-\lambda^{\nu,k}=0$ and  $\mu^{\nu,k+1}-\mu^{\nu,k}=0$ for all $k \geq \bar{k}$.  Therefore, $\left\lbrace \mathbf{w}^{\nu,k}=(\mathbf{x}^{k}, z^{\nu,k}, \lambda^{\nu,k}, \mu^{\nu,k}) \right\rbrace$  must be eventually constant (stationary), and  it thus has finite length. 
\vspace{0.05in}

\textit{Case 2.}
Consider now the case where such an integer $\bar{k}$ does not exist (and every $\left\lbrace  \mathbf{w}^{\nu,k} \right\rbrace$ is nonstationary) for $\nu=1,\ldots,N$. In this case, we first show that the P-Lagrangian $\mathcal{L}_{\alpha\beta}^{\nu}$ is finite and constant on the set of all limit points $\omega_\nu({\mathbf{w}}^{\nu}_0)$ of the sequence $\left\lbrace \mathbf{w}^{\nu,k}\right\rbrace$, and then apply Lemma \ref{lem_uniformized_KL_property} to show that $\left\lbrace  \mathbf{w}^{\nu,k} \right\rbrace$ is a Cauchy sequence and convergent. 

First, note that since  the sequence $\{\mathcal{L}_{\alpha\beta}^{\nu}\}$ is nonincreasing, we have that $\mathcal{L}_{\alpha\beta}^{\nu}({\mathbf{w}}^{\nu,k})> \mathcal{L}_{\alpha\beta}^{\nu}(\overline{\mathbf{w}}^\nu)$ for all $k$. 
This, along with \eqref{eq:thm_global_convergence_eq1}, implies that there exists an integer $k_0 $ large enough such that for any  $\varepsilon>0$ and $\delta>0$ in Lemma \ref{lem_uniformized_KL_property} 
\begin{equation}
\mathcal{L}_{\alpha\beta}^{\nu} (\overline{\mathbf{w}}^\nu)
< \mathcal{L}_{\alpha\beta}^{\nu} (\mathbf{w}^{\nu,k}) 
< \mathcal{L}_{\alpha\beta}^{\nu} (\overline{\mathbf{w}}^\nu) +\delta 
\quad \mathrm{and} \quad 
\mathrm{dist}(\mathbf{w}^{\nu,k},\omega({\mathbf{w}}_{0}^{\nu}) )< \varepsilon \quad \mathrm{for \ all} \ k \geq k_0, 
\end{equation}	
where the second comes from the fact that ${\mathrm{lim}}_{k \rightarrow \infty}\mathrm{dist}(\mathbf{w}^{\nu,k},\omega({\mathbf{w}}^{\nu}_0))=0$ (see Theorem \ref{thm_limit_saddle}). Thus  $\left\lbrace \mathbf{w}^{\nu,k}\right\rbrace$ belongs to the intersection in \eqref{eq:lem_uniform_KL_1} with	$\Omega=\omega_\nu({\mathbf{w}}^{\nu}_{0})$ for all $k \geq k_0 $, and $\Omega=\omega_\nu({\mathbf{w}}^{\nu}_{0})$ is nonempty and compact. Recall that $\{\mathcal{L}_{\alpha\beta}^\nu\}$ is bounded below by the value of  $\mathcal{L}_{\alpha\beta}^\nu$ at a saddle point, and hence $\{\mathcal{L}_{\alpha\beta}^\nu\}$ converges to a finite limit, denoted by $\underline{\mathcal{L}^{\nu}}$. It then  follows from  \eqref{eq:thm_global_convergence_eq1} that $\underline{\mathcal{L}^{\nu}}=\mathcal{L}_{\alpha\beta}^{\nu} (\overline{\mathbf{w}}^\nu)$, which shows that  $\mathcal{L}_{\alpha\beta}^{\nu}$ is finite and constant on $\omega^\nu(\overline{\mathbf{w}}^{\nu}_0)$.

Thus, since $\mathcal{L}_{\alpha\beta}^{\nu}$ is a KL function, by applying Lemma \ref{lem_uniformized_KL_property} with $\Omega=\omega^\nu({\mathbf{w}}^{\nu}_0)$ and 
$\partial \Psi(u)=\widetilde{\nabla} \mathcal{L}_{\alpha\beta}^{\nu} (\mathbf{w}^{\nu,k})$, 
we get that for any $k > k_0$ 
\begin{equation}
\varphi^{\prime}\left( \mathcal{L}_{\alpha\beta}^{\nu} (\mathbf{w}^{\nu,k}) -\mathcal{L}_{\alpha\beta}^{\nu} (\overline{\mathbf{w}}^\nu)\right)\cdot 
\mathrm{dist} \left(0, \widetilde{\nabla} \mathcal{L}_{\alpha\beta}^{\nu} (\mathbf{w}^{\nu,k}) \right) \geq 1, \notag
\end{equation}
which combined with Lemma \ref{lem_upper_bound_gradient} gives
\begin{equation}\label{eq:thm_global_case2_1}
\varphi^{\prime}\left( \mathcal{L}_{\alpha\beta}^{\nu} (\mathbf{w}^{\nu,k}) -\mathcal{L}_{\alpha\beta}^{\nu} (\overline{\mathbf{w}}^\nu)\right) 
\geq
\frac{1}{\mathrm{dist} \left(0, \widetilde{\nabla} \mathcal{L}_{\alpha\beta}^{\nu} (\mathbf{w}^{\nu,k}) \right)}
\geq
\frac{1}
{C_\nu  \left\Vert \mathbf{x}^{k} -\mathbf{x}^{k-1} \right\Vert}. 
\end{equation}
On the other hand, since $\varphi$ is concave function, we know that
\begin{multline}
\varphi\left( \mathcal{L}_{\alpha\beta}^{\nu} (\mathbf{w}^{\nu,k}) -\mathcal{L}_{\alpha\beta}^{\nu} (\overline{\mathbf{w}}^\nu)\right) 
-\varphi\left( \mathcal{L}_{\alpha\beta}^{\nu} (\mathbf{w}^{\nu,k+1}) -\mathcal{L}_{\alpha\beta}^{\nu} (\overline{\mathbf{w}}^\nu)\right) \\
\geq
\varphi^{\prime}\left(\mathcal{L}_{\alpha\beta}^{\nu} (\mathbf{w}^{\nu,k}) -\mathcal{L}_{\alpha\beta}^{\nu} (\overline{\mathbf{w}}^\nu)\right)\left(\mathcal{L}_{\alpha\beta}^{\nu} (\mathbf{w}^{\nu,k}) -\mathcal{L}_{\alpha\beta}^{\nu} ({\mathbf{w}^{\nu,k+1}})\right). \notag 
\end{multline}
For convenience, we define for any $p,q \in \mathbb{N}$
\[
\triangle_{p,q}:=\varphi\left(\mathcal{L}_{\alpha\beta}^{\nu}(\mathbf{w}^{\nu,p}) -\mathcal{L}_{\alpha\beta}^{\nu} (\overline{\mathbf{w}}^\nu)\right)
-\varphi\left(\mathcal{L}_{\alpha\beta}^{\nu} (\mathbf{w}^{\nu,q}) -\mathcal{L}_{\alpha\beta}^{\nu} (\overline{\mathbf{w}}^\nu)\right).
\]
Then we get
\begin{equation} \label{eq:thm_global_case2_2}
\triangle_{k,k+1} \geq \varphi^{\prime}\left(\mathcal{L}_{\alpha\beta}^{\nu} (\mathbf{w}^{\nu,k}) - \mathcal{L}_{\alpha\beta}^{\nu} (\overline{\mathbf{w}}^\nu)\right)
\left(\mathcal{L}_{\alpha\beta}^{\nu} (\mathbf{w}^{\nu,k}) -\mathcal{L}_{\alpha\beta}^{\nu} ({\mathbf{w}^{\nu,k+1}})\right).
\end{equation}
Recalling from Lemma \ref{lem_lag_behavior} that  $\mathcal{L}_{\alpha\beta}^{\nu} (\mathbf{w}^{\nu,k}) - \mathcal{L}_{\alpha\beta}^{\nu} ({\mathbf{w}^{\nu,k+1}}) \geq {\rho_\nu} \left\Vert \mathbf{x}^{k+1} - \mathbf{x}^{k} \right\Vert^2$, we combine \eqref{eq:thm_global_case2_1} and \eqref{eq:thm_global_case2_2} to obtain
\begin{equation}
\triangle_{k,k+1} \geq \frac{\rho_\nu\left\Vert \mathbf{x}^{k+1} -\mathbf{x}^{k} \right\Vert^2}
{C_\nu \left\Vert \mathbf{x}^{k} -\mathbf{x}^{k-1} \right\Vert}. \notag
\end{equation}
Multiplying the above inequality by $\frac{C_\nu}{\rho_\nu}\left\Vert \mathbf{x}^{k} -\mathbf{x}^{k-1} \right\Vert$  gives 
\[
\left\Vert \mathbf{x}^{k+1} -\mathbf{x}^{k} \right\Vert^2 
\leq 
\xi_\nu \triangle_{k,k+1} \left\Vert \mathbf{x}^{k} -\mathbf{x}^{k-1} \right\Vert \quad \mathrm{where} \ \  \xi_\nu:=\frac{C_\nu}{\rho_\nu},
\]
and hence
$
2\left\Vert \mathbf{x}^{k+1} -\mathbf{x}^{k} \right\Vert 
\leq 
2\sqrt{\xi_\nu \triangle_{k,k+1} \left\Vert \mathbf{x}^{k} -\mathbf{x}^{k-1} \right\Vert}.
$
Using the inequality $2\sqrt{ab} \leq a+b$ for any $a,b \geq 0$ with  $a=\left\Vert \mathbf{x}^{k} -\mathbf{x}^{k-1} \right\Vert$ and $b=\xi_\nu \triangle_{k,k+1}$, we have
\begin{equation} \label{eq:finite_length_core_ineq}
2\left\Vert \mathbf{x}^{k+1} -\mathbf{x}^{k} \right\Vert 
\leq 
\left\Vert \mathbf{x}^{k} -\mathbf{x}^{k-1} \right\Vert
+ \xi_\nu \triangle_{k,k+1}.
\end{equation}
Now we show that for any $k > k_0$ the following inequality holds:
\[
2 \sum_{l={k_0}+1}^k \left\Vert \mathbf{x}^{l+1} -\mathbf{x}^{l} \right\Vert
\leq 
\left\Vert \mathbf{x}^{k_0+1} -\mathbf{x}^{k_0} \right\Vert + \xi_\nu \triangle_{k_0+1,k+1}.
\]
By summing \eqref{eq:finite_length_core_ineq} over $l=k_0+1,\ldots, k$, we have
\begin{align}
2 \sum_{l=k_0 + 1}^k \left\Vert \mathbf{x}^{l+1} -\mathbf{x}^{l} \right\Vert
& \leq 
\sum_{l=k_0 +1 }^k\left\Vert \mathbf{x}^{l} -\mathbf{x}^{l-1} \right\Vert
+ \xi_\nu \sum_{l=k_0 +1 }^k \triangle_{l,l+1} \notag \\
&\leq
\sum_{l=k_0 +1 }^k\left\Vert \mathbf{x}^{l+1} -\mathbf{x}^{l} \right\Vert
+ \left\Vert \mathbf{x}^{k_0+1} -\mathbf{x}^{k_0} \right\Vert 
+ \xi_\nu \sum_{l=k_0 +1 }^k \triangle_{l,l+1} \label{eq:thm_global_case2_3}
\end{align}
and using fact that $\triangle_{p,q}+\triangle_{q,r}=\triangle_{p,r}$ for all $p,q,r \in \mathbb{N}$, we get 
\begin{align}
\sum_{l=k_0 +1 }^k\triangle_{l,l+1}
=\triangle_{k_0+1,k+1} 
& = \varphi\left(\mathcal{L}_{\alpha\beta}^{\nu} (\mathbf{w}^{\nu,k_0+1}) -\mathcal{L}_{\alpha\beta}^{\nu} (\overline{\mathbf{w}}^\nu)\right)
-\varphi\left(\mathcal{L}_{\alpha\beta}^{\nu} (\mathbf{w}^{\nu,k_0+2}) -\mathcal{L}_{\alpha\beta}^{\nu} (\overline{\mathbf{w}}^\nu)\right)  \notag \\
& \leq
\varphi\left(\mathcal{L}_{\alpha\beta}^{\nu} (\mathbf{w}^{\nu,k_0+1}) -\mathcal{L}_{\alpha\beta}^{\nu} (\overline{\mathbf{w}}^\nu)\right) < \infty, \label{eq:thm_global_case2_4}
\end{align}
where the last inequality follows from the fact that $\varphi \geq 0$. Plugging  \eqref{eq:thm_global_case2_4} into  \eqref{eq:thm_global_case2_3} and rearranging terms, we obtain
\begin{align}
\sum_{l=k_0 + 1}^k \left\Vert \mathbf{x}^{l+1} -\mathbf{x}^{l} \right\Vert
\leq
\left\Vert \mathbf{x}^{k_0+1} -\mathbf{x}^{k_0} \right\Vert 
+ \xi_\nu \varphi\left(\mathcal{L}_{\alpha\beta}^{\nu} (\mathbf{w}^{\nu,k_0+1}) -\mathcal{L}_{\alpha\beta}^{\nu} (\overline{\mathbf{w}}^\nu )\right)< \infty. \label{eq:thm_global_case2_5}
\end{align}
Since the right-hand side of \eqref{eq:thm_global_case2_5} does not depends $k$,  the sequence $\left\lbrace \mathbf{x}^{k} \right\rbrace $ has finite length, i.e.,
\begin{equation}
\sum_{k=1}^{\infty} \left\Vert \mathbf{x}^{k+1} -\mathbf{x}^{k} \right\Vert <  \infty. \notag
\end{equation}
This implies $ \left\lbrace \mathbf{x}^k \right\rbrace$ is a Cauchy sequence and thus a convergent sequence. By Lemma \ref{lem_bound_multi},  the multiplier sequences $\{ \lambda^{\nu,k} \}$ and $\{\mu^{\nu,k} \}$ are also Cauchy. 
Therefore, we conclude that the whole sequence $\left\lbrace ( \mathbf{x}^{k},z^{\nu,k},\lambda^{\nu,k},\mu^{\nu,k}) \right\rbrace$ converges to a saddle point  $(\overline{\mathbf{x}},\overline{z}^{\nu},\overline{\lambda}^{\nu},\overline{\mu}^{\nu})$  of $\mathcal{L}_{\alpha \beta}^{\nu}$, $\nu=1,\ldots, N$.
\endproof

We end by noting that convergence rate of the generated sequence described in Theorem \ref{thm_global_convergence} can be easily derived by applying the generic rate of convergence result in \cite{attouch2009convergence}.

\renewcommand\thetheorem{6}
\begin{theorem}[Convergence Rate]
Suppose that every $\theta_\nu$ and $g^\nu$ satisfy the KL property, where the desingularizing function $\varphi$ of $\mathcal{L}^{\nu}_{\alpha \beta}$ is of the form: $\varphi(s)=cs^{1-\delta}, \ c>0, \  \delta \in [0,1).$ Let  $\overline{\mathbf{w}}^{\nu}:=(\overline{\mathbf{x}},\overline{z}^{\nu},\overline{\lambda}^{\nu},\overline{\mu}^{\nu})$ be the limit point of $\left\lbrace \mathbf{w}^{\nu,k}:=\left( {x}^{\nu,k},z^{\nu,k},\lambda^{\nu,k},\mu^{\nu,k}\right) \right\rbrace_{\nu=1}^{N}$. Then the following convergence rates hold:
\begin{enumerate}[label=(\alph*)]
	\item If $\delta=0$, then $\left\lbrace \mathbf{w}^{\nu,k} \right\rbrace_{k \in \mathbb{N}}$ converges to $\overline{\mathbf{w}}^{\nu}$ in a finite number of steps.
		
	\item If $\delta \in (0,1/2]$, then  $\left\| \mathbf{w}^{\nu,k} - \overline{\mathbf{w}}^{\nu} \right\| \leq \widehat{C} \widehat{Q}^k$ for all $k \geq k_0$, for certain $k_0>0$.
		
	\item If $\delta \in (1/2,1)$. then $\left\| \mathbf{w}^{\nu,k} - \overline{\mathbf{w}}^{\nu} \right\| \leq \widetilde{C} k ^{-\frac{(1-\delta)}{(2\delta-1)}}$.
\end{enumerate}
\end{theorem}

\end{appendices}

\end{document}